\documentclass[a4paper,11pt]{article}

\makeatletter
\usepackage{hyperref}
\usepackage{amsfonts,dsfont,delarray,amssymb,amsmath,amsthm,a4,a4wide,color}
\usepackage{amsmath, euscript}
\usepackage[ansinew]{inputenc}

\newtheorem{theo}{Theorem}
\newtheorem{pro}{Proposition}[section]
\newtheorem{lem}[pro]{Lemma}
\newtheorem{coro}[theo]{Corollary}
\newtheorem{remark}[pro]{Remark}
\newtheorem{defi}[pro]{Definition}

\def\Xint#1{\mathchoice
   {\XXint\displaystyle\textstyle{#1}}%
   {\XXint\textstyle\scriptstyle{#1}}%
   {\XXint\scriptstyle\scriptscriptstyle{#1}}%
   {\XXint\scriptscriptstyle\scriptscriptstyle{#1}}%
   \!\int}
\def\XXint#1#2#3{{\setbox0=\hbox{$#1{#2#3}{\int}$}
     \vcenter{\hbox{$#2#3$}}\kern-.5\wd0}}
\def\dashint{\Xint-}

\DeclareMathOperator{\supp}{\mathrm{Supp}}

\def\1{\mathds{1}}
\def\admissible{{\mathcal{A}_m}}

\def\dist{\text{dist}\ }
\def\div{\mathrm{div} \ }

\def\D{\displaystyle}
\def\({\left(}
\def\){\right)}

\def\E{\Sigma}

\def\ep{\varepsilon}

\def\hal{\frac{1}{2}}

\def\indic{\mathds{1}}

\def\Ld{{L^2_{loc}(\mr^d,\mr^d)}}
\def\Lp{{L^p_{loc}(\mr^d,\mr^d)}}

\def\mn{\mathbb{N}}
\def\mr{\mathbb{R}}

\def\nab{\nabla}

\def\p{\partial}
\def\ro{\rho}

\def\supp{\text{Supp}}

\def\vp{\varphi}

\def\w{{H_n}}
\def\W{\mathcal{W}}

\def\z{\zeta}
\def\Z{Z_n^\beta}

\def\mK{{\mathcal{K}}}

\def\xbf{{\mathbf x}}

\def\I{\mathcal{E}}
\def\j{\mathbf{E}}
\def\cK{\mathcal{K}}

\def\hj{\widehat{\j}}
\def\g{w}

\def\Q{\mathbb{P}_{n,\beta}}
\def\Qk{\mathbb{P}_{n,\beta} ^{(k)}}
\def\Qone{\mathbb{P}_{n,\beta} ^{(1)}}

\def\hel{h_{n,\ell}}
\def\het{h'_{n,\eta}}

\def\ba{\mathcal{A}}
\def\bam{\overline{\mathcal{A}}_m}
\def\bai{\overline{\mathcal{A}}_1}
\def\P{\mathcal{P}}
\def\F{\mathcal{F}}
\def\En{\mathcal{E}}
\def\In{\mathcal{I}_n}

\newcommand{\Fnbeta}{\F_{n,\beta}}
\newcommand{\Fnbetae}{F_{n,\beta}}
\newcommand{\mubet}{\mu_{\beta}}

\newcommand{\R}{\mathbb{R}}

\newcommand{\wto}{\rightharpoonup}

\def\dist{\mathrm{dist}}
\def \La{\Lambda}
\def\Lat{\tilde{\Lambda}}
\def\Lah{\hat{\Lambda}}
\newcommand{\hti}{\tilde{h}}
\newcommand{\hb}{\bar{h}}
\newcommand{\EetR}{F_{\eta,R}}
\newcommand{\deleta}{\delta ^{(\eta)}}
\newcommand{\delell}{\delta ^{(\ell)}}

\newcommand{\hext}{h^{\rm rem}}
\newcommand{\hint}{V_{\eta,R}}

\newcommand{\one}{\mathds{1}}

\newcommand{\mubf}{\boldsymbol{\mu}}

\numberwithin{equation}{section}

\title{Higher Dimensional Coulomb Gases and Renormalized Energy Functionals}
\author{N. Rougerie\footnote{Universit\'e Grenoble 1 \& CNRS, LPMMC, UMR 5493, B.P. 166, 38042 Grenoble, France}  \  and S. Serfaty\footnote{UPMC Univ. Paris 6,  UMR 7598 Laboratoire Jacques-Louis Lions,  Paris, F-75005 France \newline \& Courant Institute, New York University, 251 Mercer st, NY NY 10012, USA.} 
}
\date{July 9, 2013}

\begin{document}

\maketitle

\begin{abstract}
We consider a classical system of $n$ charged particles in an external confining potential, in any dimension $d\ge 2$. The particles interact via  pairwise repulsive Coulomb forces and the coupling parameter is of order $n^{-1}$ (mean-field scaling). By a suitable splitting of the Hamiltonian, we extract the next to leading order term  in the ground state energy, beyond the mean-field limit. We show that this next order term, which characterizes the fluctuations of the system, is governed by a new ``renormalized  energy" functional providing a way to compute the total Coulomb energy of a jellium (i.e. an infinite set of point charges screened by a uniform neutralizing background), in any dimension. The renormalization that cuts out the infinite part of the  energy is achieved by  smearing out the point charges at a small scale, as in Onsager's lemma.  We obtain consequences for the statistical mechanics of the Coulomb gas: next to leading order asymptotic expansion of the free energy or partition function, 
characterizations of the Gibbs measures, estimates on the local charge fluctuations and factorization estimates for reduced densities. This extends results of Sandier and Serfaty to dimension higher than two by an alternative approach. 
\end{abstract}

\tableofcontents

\section{Introduction}

We study the equilibrium properties of a classical Coulomb gas or ``one-component plasma": a system of $n$ classical charged particles living in the full space of dimension $d\ge 2$, interacting via Coulomb forces and confined by an external electrostatic potential $V$. We will be interested in the mean-field regime where the number $n$ of particles is large and the pair-interaction strength (coupling parameter) scales as the inverse of $n$. We study the ground states of the system as well its statistical mechanics when temperature is added. Denoting $x_1,\ldots,x_n$ the positions of the particles, the total energy at rest of such a system is given by the Hamiltonian 
\begin{equation}\label{wn}
\w(x_1, \dots, x_n)=   \sum_{i\neq j} \g (x_i-x_j)  +n \sum_{i=1}^n V(x_i)
\end{equation}
where 
\begin{equation}\label{eq:coul ker}
\begin{cases}
\displaystyle \g(x)=\frac{1}{|x|^{d-2}}& \text{if}\  d \ge 3\\
\displaystyle \g(x)= -   \log |x|& \text{if } \ d=2
\end{cases}
\end{equation} 
is  a multiple of the Coulomb potential in dimensions $d\ge 2$, i.e. we have
\begin{equation}
-\Delta \g= c_d \delta_0  \end{equation}
with
\begin{equation}\label{defc}
c_2 = 2\pi, \qquad c_d = (d-2)|\mathbb{S}^{d-1}| \ \text{when} \ d\ge 3
\end{equation}
and $\delta_0$ is the Dirac mass at the origin. The one-body potential $V:\mr^d \to \mr$ is a continuous function, growing at infinity (\emph{confining} potential). More precisely, we assume 
\begin{equation}\label{eq:trap pot}
\begin{cases}
\displaystyle \lim_{|x|\to \infty} V(x) = +\infty  & \text{ if }  d \ge 3\\
\displaystyle \lim_{|x|\to \infty} \left(\frac{V(x)}{2} - \log |x|\right) = +\infty  & \text{ if }  d=2.
\end{cases}
\end{equation} 
Note the factor $n$ in front of the one-body term (second term) in \eqref{wn} that puts us in a mean-field scaling where the one-body energy and the two-body energy (first term) are of the same order of magnitude. This choice is equivalent to demanding that the pair-interaction strength be of order $n^{-1}$. One can always reduce to this situation in the particular case where the trapping potential $V$ has some homogeneity, which is particularly important in applications. 
We will not treat at all the case of one-dimensional Coulomb gases (where the interaction kernel $\g$ is $|x|$), since this case has been shown to be essentially completely solvable a long time ago \cite{am,len1,len2,BL,Kun}.

Classical Coulomb systems are    fundamental  systems of statistical mechanics, since they can be seen as a toy model for matter, containing the truly long-range nature of electrostatic interactions. Studies in this direction include \cite{sm,LO,jlm,PS}, see \cite{ln} for a review. Another motivation is that, as was pointed out by Wigner \cite{wigner} and exploited by Dyson \cite{dyson}, two-dimensional Coulomb systems are directly related to Gaussian random matrices, more precisely the Ginibre ensemble, and such random matrix models have also received much attention for their own sake. A similar connection exists between ``log-gases" in dimension 1 and the GUE and GOE ensembles of random matrices, as well as more indirectly to orthogonal polynomial ensembles. For more details on these aspects we refer to \cite{For}, and for an introduction to the random matrix aspect to the texts \cite{agz,mehta,deift}. A recent trend in random matrix theory is the study of universality with respect to the entries' 
statistics, i.e. the fact that results for Gaussian entries carry over to the general case, see e.g. \cite{TV,ESY}.

We are interested in equilibrium properties of the system in the regime $n\to \infty$, that is on the large particle number asymptotics of the ground state and the Gibbs state at given temperature. In the former case we consider configurations $(x_1,\ldots,x_n)$ that minimize the total energy \eqref{wn}. We will denote 
\begin{equation}\label{eq:gse}
E_n := \min_{\R ^{dn}} \w 
\end{equation}
the ground state energy. It is well-known (we give references below) that to leading order 
\begin{equation}\label{eq:gse first order}
E_n = n ^2 \En[\mu_0] (1+o(1)) 
\end{equation}
in the limit $n\to \infty$ where 
\begin{equation}\label{eq:def MF ener}
\En[\mu] = \iint_{\mr^d\times \mr^d} \g(x-y) \, d\mu(x)\, d\mu(y)+ \int_{\mr^d}V(x)\, d\mu(x)
\end{equation}
is the mean-field energy functional defined for Radon measures $\mu$, and $\mu_0$ (the equilibrium measure) is the minimizer of $\En$ amongst probability measures on $\R ^d$. In this paper we quantify precisely the validity of \eqref{eq:gse first order} and characterize the next to leading order correction. We also study the consequences of these asymptotics on minimizing and thermal configurations. By the latter we mean the Gibbs state at inverse temperature $\beta$, i.e. the probability law
\begin{equation}\label{eq:defi Gibbs state}
\Q (x_1,\ldots,x_n) =\frac{1}{\Z} e^{-\frac{\beta}{2} \w (x_1,\ldots,x_n)} \, dx_1 \ldots dx_n
\end{equation}
where $\Z$ is a normalization constant, and we are again interested in obtaining next order expansions of the partition function $\Z$ as well as consequences for the distributions of the points according to the temperature. This program has been carried out in \cite{ss2d} in dimension $d=2$ and here we extend it  to arbitrary higher dimension -- in particular the more physical case $d=3$, and  provide at the same time a simpler  approach  to recover (most of) the results of dimension $2$. In \cite{ss2d} it was shown that the next order corrections are related to a ``renormalized energy" denoted $W$ -- so named in reference to the procedure used in its definition and related functionals used in Ginzburg-Landau theory \cite{BBH,SSbook}, but the derivation of this object was restricted to dimension 2 due to an obstruction  that we still do not know how to overcome (more precisely the derivation relies on a ``ball construction method" which crucially uses the conformal invariance of the Coulomb kernel in two 
dimensions). Here we again connect the problem to a slightly different ``renormalized energy," this time denoted  $\W$ (it is the same in good cases, but different in general) and the approach to its definition and derivation are not at all restricted by the dimension. They rely on smearing out point charges and Onsager's lemma \cite{Ons}, a celebrated tool that has been much used in the proof of the stability of matter (see \cite{LO} and \cite[Chapter 6]{LieSei}).

We choose to use an electrostatic/statistical mechanics vocabulary that is more fit to general dimensions but the reader should keep in mind the various applications of the Coulomb gas, especially in two dimensions: Fekete points in polynomial interpolation \cite{ST}, Gaussian random matrices which correspond to $\beta=2$ \cite{For}, vortex systems in classical and quantum fluids \cite{CLMP,gl13,CY,CRY}, fractional quantum Hall physics \cite{Gir,RSY1,RSY2} ...

\medskip

The easiest way to think of the limit \eqref{eq:gse first order} is as a continuum limit: the one-body potential $nV$ confines the large number $n$ of particles in a bounded region of space, so that the mean distance between points goes to zero and the empirical measure 
\begin{equation}\label{eq:defi empir}
\mu_n := \frac{1}{n} \sum_{i=1} ^n \delta_{x_i}
\end{equation}
of a minimizing configuration converges to a density $\ro$.
%\begin{equation}\label{eq:defi empir}
%\nu_n \approx \rho \in C (\R ^d).
%\end{equation}
The functional \eqref{eq:def MF ener} is nothing but the continuum energy corresponding to \eqref{wn}: the first term is the classical Coulomb interaction energy of the charge distribution $\mu$ and the second the potential energy in the potential $V$. If the interaction potential $\g$ was regular at the origin we could write 
\begin{align*}
\w (x_1,\ldots,x_n) &= n ^2 \left( \int_{\mr^d}V(x)\, d\mu_n(x) + \iint_{\mr^d\times \mr^d} \g(x-y) \, d\mu_n(x)\, d\mu_n(y)\right) - n \g(0) \\
&= n ^2 \En[\mu_n] \left( 1+ O(n ^{-1})\right)
\end{align*}
and \eqref{eq:gse first order} would  easily follow from a simple compactness argument. In the case where $\g$ has a singularity at the origin, a regularization procedure is needed but this mean-field limit result still holds true, meaning that  for minimizers of $\w$, the empirical measure $\mu_n$ converges to $\mu_0$, the minimizer of \eqref{eq:def MF ener}.  This is standard and can be found in  a variety of sources: e.g. \cite[Chap. 1]{ST} for the Coulomb kernel in dimension $2$, \cite{CGZ} for a more general setting including possibly non-coulombian kernels, or \cite{ln} for  a simple general treatment. This  leading order result is often complemented by  a much stronger  large deviations principle in the case with temperature: a large deviations principle with speed $n^2$ and good rate function $\beta \En$ holds;  see \cite{hiaipetz,BZ,hardy} for the two-dimensional Coulomb case (with $\beta=2$), which can be adapted to any finite temperature and any dimension \cite{CGZ,ln}. This is also of interest in 
the more elaborate settings of complex manifolds, cf. e.g. \cite{berman,bbn} and references therein.

\medskip

Another way to think of the mean-field limit, less immediate in the present context but more suited for generalizations in statistical and quantum mechanics, is as follows. In reality, particles are indistinguishable, and the configuration of the system should thus be described by a probability measure $\mubf(\xbf) = \mubf (x_1,\ldots,x_n)$, which is  symmetric under particle exchange:
\begin{equation}\label{eq:symmetry}
\mubf(x_1,\ldots,x_n) = \mubf (x_{\sigma(1)},\ldots,x_{\sigma(n)}) \mbox{ for any permutation } \sigma.  
\end{equation}
An optimal (i.e. most likely) configuration $\mubf_n$ is found by minimizing the $n$-body energy functional
\begin{equation}\label{eq:ener N}
\In [\mubf] := \int_{\R ^{dn}}  H_n (\xbf) \mubf(d\xbf)  
\end{equation}
amongst symmetric probability measures $\mubf \in \P_s (\R ^{dn})$ (probability measures on $\R ^{dn}$ satisfying \eqref{eq:symmetry}). It is immediate to see that $\mubf_n$ must be a convex superposition of measures of the form $\delta_{(x_1,\ldots,x_n)}$ with $(x_1,\ldots,x_n)$ minimizing $\w$ (in other words it has to be a symmetrization of some $\delta_{\xbf}$ for a minimizing configuration $\xbf$). The infimum of the functional \eqref{eq:ener N} of course coincides with 
\begin{equation*}
\inf_{\mubf \in \P_s (\R ^{dn})} \int_{\R ^{dn}} H_n (\xbf) \mubf(d\xbf)  = E_n 
\end{equation*}
and a way to understand the asymptotic formula \eqref{eq:gse first order} is to think of the minimizing $\mubf_n$ as being almost factorized 
\begin{equation}\label{eq:intro factor}
\mubf_n (x_1,\ldots,x_n) \approx \rho ^{\otimes n} (x_1,\ldots,x_n) = \prod_{j=1} ^n \rho (x_i)
\end{equation}
with a regular probability measure $\rho \in \P (\R ^d)$. Plugging this ansatz into \eqref{eq:ener N} we indeed obtain 
\[
\In [\rho ^{\otimes n}] =  n ^2 \En[\rho] \left( 1+ O(n ^{-1})\right)
\]
and the optimal choice is $\rho=\mu_0$. The mean-field limit can thus  also be understood as one where correlations amongst the particles of the system vanish in the limit $n\to \infty$, which is the meaning of the factorized ansatz. 

In this paper we shall pursue both the ``uncorellated limit'' and the ``continuum limit" points of view beyond leading order considerations. That is, we shall quantify to which precision (and in which sense) the empirical measure \eqref{eq:empirical measure} of a minimizing configuration can be approximated by $\mu_0$ and the $n$-body ground state factorizes in the form \eqref{eq:intro factor} with $\rho = \mu_0$.  Previous results related to the ``uncorellated limit'' point of view may be found in \cite{CLMP,Kie1,Kie2,KS,MS}.
Another way of viewing this is that we are looking at characterizing the ``fluctuations" of the distribution of points  around its limit measure, in other words the behavior of $n( \mu_n - \mu_0)$ where $\mu_n$ is the empirical measure \eqref{eq:defi empir}. In the  probability literature, such questions are now understood in dimension $2$  for the particular determinantal case $\beta=2$ \cite{ahm,ahm2}, and in dimension $1$ with the logarithmic interaction \cite{vv,BEY1,BEY2}. Our results are less precise (we do not exhibit exact local statistics of spacings), but they are valid for any $\beta$, any $V$, and any  dimension $d\ge 2$ (we could also treat the log gas in dimension $1$, borrowing ideas from  \cite{ss1d} to complete those we use here). We are not aware of any previous results giving any information on such fluctuations for ground states or thermal states in dimension $\ge 3$. 

\medskip

Most of our results follow from almost exact splitting formulae for the Hamiltonian \eqref{wn}  that reveal the corrections beyond leading order in \eqref{eq:gse first order}, in the spirit of \cite{ss2d}. Let us first explain  what physics governs these corrections. As already mentioned, points minimizing $\w$ tend to be densely packed in a bounded region of space (the support of $\mu_0$, that we shall denote $\E$) in the limit $n\to \infty$. Their distribution (i.e. the empirical measure) has to follow $\mu_0$ on the macroscopic scale but this requirement still leaves a lot of freedom on the configuration at the {\it microscopic scale}, that is on lengthscales  of order $n ^{-1/d}$ (the mean inter-particle distance). A natural  idea is thus to blow-up at scale $n ^{-1/d}$ in order to consider configurations where points are typically separated by distances of order unity, and investigate which microscopic configuration is favored. On such length scales, the equilibrium measure $\mu_0$ varies slowly so the 
points will want to follow a constant density given locally by the value of $\mu_0$. Since the problem is electrostatic in nature it is intuitive that the correct way to measure the distance between the configuration of points and the local value of the equilibrium measure should use the Coulomb energy. This leads to the idea that the local energy around a blow-up origin should be the electrostatic energy of what is often called a {\it  jellium} in physics: an infinite collection of interacting particles in a constant neutralizing background of opposite charge, a model originally introduced in \cite{wigner1}. At this microscopic  scale the pair-interactions will no longer be of mean-field type, their strength will be of order $1$.  The splitting formula will allow to separate exactly the Coulomb energy of this jellium as the next to leading order term, except that what will come out is more precisely some average of all the  energies of the jellium configurations obtained after blow-up around all possible 
origins.

\medskip

Of course it is a  delicate matter to define the energy of the general infinite jellium in a mathematically rigorous way: one has to take into account the pair-interaction energy of infinitely many charges, without assuming any local charge neutrality, and the overall energy may be finite only via screening effects between the charges and the neutralizing background that are difficult to quantify. This has been done  for the first time in 2D in \cite{gl13,ss2d}, the energy functional for the jellium being  the renormalized energy $W$ alluded to above.  As already mentioned, one of the main contributions of the present work is to present an alternate definition $\W$ that generalizes better to higher dimensions. The precise definition will be given later, but we can already state our asymptotic formula for the ground state energy (minimum of $\w$), where $\alpha_d$ denotes the minimum of $\W$ for a jellium of density $1$ in dimension $d$:
\begin{equation}\label{eq:gse second order}
\boxed{E_n = 
\begin{cases} 
 \displaystyle      n ^2 \En [\mu_0] + \frac{n ^{2-2/d}}{c_d}\alpha_d \int \mu_0 ^{2-2/d} (x)  dx + o \left(n ^{2-2/d}\right) \mbox{ if } d\geq 3\\ 
 \displaystyle      n ^2 \En [\mu_0] - \frac{n}{2}\log n + n\left(\frac{\alpha_2}{2\pi} -\hal \int \mu_0(x) \log \mu_0(x) \,dx \right) + o\left( n\right)
       \mbox{ if } d= 2.
      \end{cases}
}
\end{equation}
This formula encodes the double scale nature of the charge distribution: the first term is the familiar mean-field energy and is due to the points following the macroscopic distribution  $\mu_0$ -- we assume that the probability $\mu_0$ has a density that we also denote $\mu_0$ by abuse of notation. The next order correction {which happens to lie at the order $n^{2-2/d}$} governs the configurations at the microscopic scale {and the crystallization, by selecting configurations which }  minimize $\W$ (on average, with respect to the blow-up centers). 

The factors involving $\mu_0$ in the correction come from scaling, and  from the fact that the points locally see a neutralizing background whose charge density is given by the value of $\mu_0$. The scaling properties of the renormalized energy imply that the minimal energy of a jellium  with neutralizing density $\mu_0 (x)$  is $\alpha_d \mu_0 ^{2-2/d} (x)$ (respectively $\mu_0(x)\alpha_2 - \mu_0 (x) \log \mu_0(x)$ when $d=2$). Integrating this energy density on the support of $\mu_0$ leads to the formula.
 
The interpretation of the correction is thus that around (almost) any point $x$ in the support of $\mu_0$ there are approximately $\mu_0(x)$ points per unit volume, distributed so as to minimize a jellium energy with background density $\mu_0(x)$. Due to the properties of the jellium, this implies that, up to a $\mu_0(x)$-dependent rescaling, the local distribution of particles are the same around any blow-up origin $x$ in the support of $\mu_0$. This can be interpreted as a result of universality with respect to the potential $V$ in \eqref{wn}, in connection with recent works on the 1D log gas \cite{BEY1,BEY2,ss1d}.  

We remark that even the weaker result that the correction in \eqref{eq:gse second order} is exactly of order $n^{2-2/d}$ (with an extra term of order $n\log n$ in 2D due to the scaling properties of the log) did not seem to have been previously noticed  (except in \cite{ss2d} where the formula \eqref{eq:gse second order} is derived in the case $d=2$). 

\medskip

The next natural question is of course that of the nature of the minimizers of the renormalized energy $\W$. It is widely believed  that the minimizing configuration (at temperature~$0$) consists of points distributed on a regular lattice  (Wigner crystal). A proof of this is out of our present reach: crystallization problems have up to now been solved only for specific short range interaction potentials (see \cite{The,BPT,HR,Sut,Rad} and references therein) that do not cover Coulomb forces, or 1D systems \cite{BL,Kun,ss1d}. In \cite{gl13} it was shown however that in dimension $2$, \emph{if} the minimizer is a lattice, then it has to be the triangular one, called  the {\it Abrikosov lattice} in the context of superconductivity. In dimension $3$ and higher, the question is wide open, not only to prove that minimizers are  crystalline, but even to identify the optimal lattice configurations. The FCC (face centered cubic) lattice, and maybe also the BCC (body centered cubic) lattice, seem like natural 
candidates, and the optimisation of $\W$ is related to the computation of what physicists and chemists call their Madelung constants.

This  question is in fact of number-theoretic nature: in \cite{ss2d} it is shown that  the question in dimension $2$  reduces to minimizing the Epstein Zeta function $\sum_{p\in \Lambda} |p|^{-s}$ with $s>2$ among lattices $\Lambda$, a question which was in turn already solved in dimension $2$ in the 60's (see \cite{cassels,rankin,ennola,ennola2,diananda} and also \cite{montgomery} and references therein), but the same question is open in dimension $d\ge 3$ -- except for $d=8$ and $d=24$ -- and it is only conjectured that the FCC is a local minimizer, see \cite{sarns} and references therein. The connection with the minimization of $\W$ among lattices and that of the Epstein Zeta functions is not even rigorously clear in dimension $d\ge 3$. For more details, 
we refer to Section \ref{sec3} where we examine and compute  $\W$ in the class  of periodic configurations and discuss this question.
\medskip

Modulo the conjecture that the minimizer of the renormalized energy is a perfect lattice, \eqref{eq:gse second order} (and the consequences for minimizers that we state below) can be interpreted as a crystallization result for the Coulomb gas at zero temperature.

Concerning the distribution of charges at positive temperature, one should expect a transition from a crystal at low temperature to a liquid at large enough temperature. While we do not have any conclusive proof of this fact, our results on the Gibbs state \eqref{eq:defi Gibbs state} strongly suggest that the transition should happen in the regime $\beta \propto n ^{2/d-1}$. In particular we prove that (see Theorem \ref{thm:partition} below):
\begin{enumerate}
\item If $\beta \gg n ^{2/d-1}$ (low temperature regime) in the limit $n\to \infty$, then the free-energy (linked to the partition function $\Z$) is to leading order given by the mean-field energy, with the correction expressed in terms of the renormalized energy as in \eqref{eq:gse second order}. In other words, the free energy and the ground state energy agree up to corrections of smaller order than the contribution of the jellium energy. We interpret this as a first indication of crystallization in this regime. Note that for the 1D log gas in the corresponding regime, a full proof of crystallization has been provided in \cite{ss1d}. 
\item If $\beta \ll n ^{2/d-1}$ (high temperature regime) in the limit $n\to \infty$, the next-to leading order correction to the free energy is no longer given by the renormalized energy, but rather by an entropy term. We take this as a weak indication that the Gibbs state is no longer crystalline, but a more detailed analysis would be required.\end{enumerate}

Note the dependence on $d$ of the critical order of magnitude: it is of order $1$ only in 2D. Interestingly, in the main applications where the Gibbs measure of the 2D Coulomb gas arises (Gaussian random matrices and quantum Hall phases), the inverse temperature $\beta$ is a number independent of $n$, i.e. one is exactly at the transition regime.

\medskip

In the next section we proceed to state our results rigorously. Apart from what has already been mentioned above they include:
\begin{itemize}
\item Results on the ground state configurations that reveal the two-scale structure of the minimizers hinted at by \eqref{eq:gse second order}, see Theorem \ref{th1}.
\item A large deviations - type result in the low temperature regime (Theorem \ref{th4}), that shows that in the limit $\beta \gg n ^{2/d-1}$ the Gibbs measure charges only configurations whose renormalized energy converges to the minimum. When $\beta \propto n^{2/d-1} $ only configurations with a certain upper  bound on $\W$ are likely.
\item Estimates on the charge fluctuations in the Gibbs measure, in the low temperature regime again (Theorem \ref{thm:charge fluctu}). These are derived by exploiting coercivity properties of the renormalized energy.
\item Precise factorization estimates for the ground state and Gibbs measure, giving a rigorous meaning to \eqref{eq:intro factor}, that we obtain as corollaries of our estimates on charge fluctuations, see Corollary \ref{thm:marginals}.
\end{itemize}

These results, let us note, should be seen as coming in two groups. The first one reduces the crystallisation question for the trapped Coulomb gas to the same question for the jellium, in a suitable weak sense. Roughly speaking, we show that configurations for the trapped Coulomb gas look almost everywhere locally crystalline at small enough temperature, in an average sense that we make precise in Sections \ref{sec:results 0 T} and \ref{sec:results finite T}. Our second set of results aims at quantifying the deviations from mean-field theory in a stronger sense than that we use in Sections \ref{sec:results 0 T} and \ref{sec:results finite T}. Here our results (stated in Section \ref{sec:results deviations}), are not optimal, for example we can only show that particles are localized within a precision $O(n^{-1/d+2})$, whereas crystallization happens at scale $n ^{-1/d}$. The estimates we obtain are however, as far as we know, the best of their kind in $d\geq 3$. Concerning this second group 
of results we finally mention that, although we focus on the low temperature regime, our methods (or close variants) also yield explicit estimates on the Gibbs measure for any temperature. 
%There is no reason to think that they are optimal except in the low temperature regime. 
They can be used if the need for robust and explicit information on the behavior of the Gibbs measure arises, as e.g. in \cite{RSY1,RSY2} where we studied quantum Hall states with related techniques.

\medskip

{\bf Acknowledgments:} S.S. was supported by a EURYI award and would like to thank the hospitality of the Forschungsinstitut f\"ur Mathematik at the ETH Z\"urich, where part of this work was completed. N.R. thanks Gian Michele Graf and Martin Fraas for their hospitality  at the Institute for Theoretical Physics of ETH Z\"urich, and acknowledges the financial support of the ANR (project Mathostaq ANR-13-JS01-0005-01).

\section{Statement of main results}

We let the mean-field energy functional be as in \eqref{eq:def MF ener}. The minimization of $\I$ among $\mathcal{P}(\mr^d)$, the space of  probability measures on $\mr^d$,  is a standard problem in potential theory (see \cite{frostman} or \cite{ST} for $d=2$). The uniqueness of a minimizer, called the {\it equilibrium measure}  is obvious by strict convexity of $\I$ and its existence can be proven using the continuity of $V$ and the assumption \eqref{eq:trap pot} -- note that it only depends on the data of $V$ and the dimension.
 We  recall it is denoted by $\mu_0$, assumed to have a density also denoted $\mu_0(x) $ by abuse of notation,  and its support is denoted by $\Sigma$ {and usually called the {\it droplet}}.
  We will need the following additional assumptions on $\mu_0$:
\begin{align}\label{ass1}
\p \Sigma \ &\text{is} \ C^1\\
\label{ass2} \mu_0 \ \text{is } C^1 \  \text{on its support $\E$ and } \ &0<\underline{m} \le \mu_0(x) \le \overline{m}<\infty \quad \text{there}.
\end{align}
It is easy to check that these are satisfied for example if $V$ is quadratic, in which case $\mu_0$ is the characteristic function of a ball. More generally, if $V\in C^2$ then $\mu_0(x)=\Delta V(x)$ on its support. The regularity assumption is purely technical: essentially it makes the construction needed for the upper bound easier. The upper and lower bounds to $\mu_0$ on its support ensure that the ``jellium energy with background density $\mu_0 (x)$'', which will give the local energy fluctuation around $x\in \Sigma$, makes sense . 

{We note that $\mu_0$ is also related to the solution of an {\it obstacle problem} (see the beginning of Section \ref{sec4}), and if $V\in C^2$ and $\Delta V>0$ then  the droplet $\E$ coincides with the {\it coincidence set} of the obstacle problem, for which  a regularity theory is known \cite{caff}.}

\subsection{The renormalized energy}

This section is devoted to the precise definition of the renormalized energy. It is 
defined via the electric field $\j$ generated by the full charge system: a  (typically infinite) distribution of point charges in space in a constant neutralizing background.
Note first that the classical Coulomb interaction of two charge distributions (bounded Radon measures) $f$ and $g$,
\begin{equation}\label{eq:def Coul ener}
D(f,g):=\iint_{\mr^d \times \mr^d} \g(x-y)\, df(x) \, dg(y)
\end{equation}
is linked to the (electrostatic) potentials $h_f =  \g * f$, $h_g =  \g * g $ that they generate via the formula\footnote{Here we assume that all quantities are well-defined.}
\begin{equation}\label{eq:ener field}
D(f,g) = \int_{\R ^d} f h_g  = \int_{\R ^d} g h_f =\frac{1}{c_d} \int_{\R ^d} \nabla h_f \cdot \nabla h_g
\end{equation}
where we used the fact that by definition of $\g$, 
$$ -\Delta h_f = c_d f,\quad -\Delta h_g = c_d g.$$
The electric field generated by the distribution $f$ is given by $\nab h_f$, and its square norm thus  gives a constant times  the electrostatic energy density of the charge distribution $f$: 
$$D(f,f) = \frac{1}{c_d}\int_{\R ^d} |\nabla h_f| ^2.$$

The electric field generated by a jellium  is of the form described in the following definition.

\begin{defi}[\textbf{Admissible electric fields}]\label{def:adm field}\mbox{}\\ 
Let $m >0$. Let $\j$ be a vector field in $\mr^d$. We say that $\j$ belongs to the class $\bam$ if $\j = \nab h$ with
\begin{equation}\label{curlj}
-\Delta h = c_d \Big(\sum_{p\in \Lambda}N_p \delta_p - m \Big) \quad \text{in} \ \mr^d
\end{equation}
for some discrete set $\Lambda \subset \mr^d$, and $N_p$ integers in $\mn^*$.
%We say that $\j$ belongs to the class $\mathcal{A}_m$ if the same holds with $N_p$ all equal to $1$.
\end{defi}

Note that in \cite{gl13} a different convention was used, the electric field being rotated by $\pi/2$ at each point to represent a superconducting current. This made the analogy with 2D Ginzburg-Landau theory more transparent but does not generalize easily to higher dimensions. In the present definition $h$ corresponds to the electrostatic potential  generated by the jellium and $\j$ to its electric field, while the constant $m$ represents the mean number of particles per unit volume, or the density of the neutralizing background. An important difficulty is that the electrostatic energy $D(\delta_p,\delta_p)$ of a point charge density $\delta_p$, where $\delta_p$ denotes the Dirac mass at $p$, is infinite, or in other words, that  the electric field generated by point charges fails to be in $L^2_{loc}$. This is where the need for a ``renormalization" of this infinite contribution comes from.

To remedy this, we replace point charges by smeared-out charges, as in Onsager's lemma: We pick some arbitrary fixed {\it radial} nonnegative function $\ro$, supported in $B(0,1)$ and with integral $1$. For any point $p$ and $\eta>0$ we introduce   the smeared charge 
\begin{equation}\label{eq:smeared charge}
\delta_p^{(\eta)}= \frac{1}{\eta^d}\ro \left(\frac{x}{\eta}\right) *  \delta_p.
\end{equation}
Even though the value of $\W$ will depend in general on the precise choice of $\ro$,  
the results in the paper will not depend on it, implying  in particular that the value of the minimum of $\W$ will not either (thus we do not try to optimize over the possible choices). A simple example is to take 
$ \ro = \frac{1}{|B(0,1)|} \indic_{B(0,1)},$
in which case 
$$\delta_p^{(\eta)} = \frac{1}{|B(0,\eta)|} \indic_{B(0,\eta)}.$$
We also define
\begin{equation}\label{kapd}
\kappa_d:= c_d D(\delta_0^{(1)}, \delta_0^{(1)}) \quad \text{for} \  d\ge 3,  \quad \kappa_2=c_2, \qquad \gamma_2=  c_2 D(\delta_0^{(1)}, \delta_0^{(1)})  \ \text{for} \ d=2.
\end{equation} 
The numbers $\kappa_d$, $\gamma_2$,  depend only on the choice of the function $\ro$ and on the dimension. This nonsymmetric  definition is due to the fact that the  logarithm behaves differently from power functions under rescaling, and is  made to ease the formulas below.

Newton's theorem  \cite[Theorem 9.7]{LiLo} asserts that the Coulomb potentials generated by the smeared charge $\delta_p^{(\eta)}$ and the point charge $\delta_p$ coincide outside of $B(p,\eta)$. A consequence of this is that there exists a radial function $f_\eta$ solution to
\begin{equation}\label{eqf0}
\begin{cases}
 -  \Delta f_\eta= c_d\left( \delta_0^{(\eta)} - \delta_0 \right) \quad \text{in} \ \mr^d
\\   
 f_\eta\equiv 0 \quad  \text{in} \ \mr^d \backslash B(0,\eta).
\end{cases}
\end{equation}
and it is easy to define the field $\j_\eta$ generated by a jellium with smeared charges starting from the field of the jellium with (singular) point charges, using $f_\eta$:

\begin{defi}[\textbf{Smeared electric fields}]\label{def:smear field}\mbox{}\\
For any vector field $\j=\nab h$ satisfying
\begin{equation}\label{dij}
-\div \j= c_d\Big(\sum_{p \in \Lambda}N_p \delta_p -m\Big)
\end{equation}
in some subset $U$ of $\mr^d$, with $\Lambda \subset U$ a discrete set of points, we let
$$\j_\eta := \nab h + \sum_{p\in \Lambda} N_p \nab f_\eta(x-p) \qquad h_\eta= h + \sum_{p \in \Lambda} N_p f_\eta(x-p).$$
We have
\begin{equation}
\label{delp}
-\div  \j_\eta  = - \Delta h_\eta =  c_d\Big(\sum_{p \in \Lambda}N_p \delta_p^{(\eta)} -m\Big)\end{equation}
and denoting by $\Phi_\eta$ the map $\j \mapsto \j_\eta$,  we note that $\Phi_\eta$ realizes a bijection from the set of vector fields satisfying \eqref{dij} and which are gradients,  to those satisfying \eqref{delp} and which are also gradients.
\end{defi}

Note that the above definition  in principle  depends implicitly on the set $U$, whose choice will be clear from the context in the sequel (most of the time we will take $U= \R ^d$).

For any fixed $\eta>0$ one may then define the electrostatic energy per unit volume of the infinite jelium with smeared charges as
\begin{equation}\label{eq:Weta pre}
\limsup_{R\to \infty} \dashint_{K_R}  |\j_\eta|^2 := \limsup_{R\to \infty} |K_R| ^{-1} \int_{K_R}  |\j_\eta|^2 
\end{equation}
where $\j_\eta$ is as in the above definition and $K_R$ denotes the cube $[-R,R]^d$. Note that the quantity $\dashint_{K_R}  |\j_\eta|^2$ may  not have a limit (this does occur for somewhat pathological configurations). This motivates the use of the $\limsup$ instead, which gives the strongest possible control on the configuration (see e.g. Lemma~\ref{lembornnu} below).   In addition, for the configurations we obtain as asymptotic limits, the ergodic theorem which we use to exhibit these quantities will always ensure the existence of a true limit.

This energy is now well-defined for $\eta>0$ and blows up as $\eta \to 0$, since it includes the self-energy of each smeared charge in the collection, absent in the original energy  (i.e. in the Hamiltonian \eqref{wn}). One should thus \emph{renormalize} \eqref{eq:Weta pre} by removing the self-energy of each smeared charge before taking the limit $\eta \to 0$. We will see that  the leading order  energy of a smeared charge is $\kappa_d \g(\eta)$, and this is the quantity that should be removed for each point. But  in order for the charges to efficiently screen the neutralizing background,  configurations will need to have  the same charge density as the neutralizing background (i.e. $m$ points per unit volume). We will prove in Lemma \ref{lembornnu} that this holds. We are then led to the definition  
\begin{defi}[\textbf{The renormalized energy}]\label{def:renorm ener}\mbox{}\\
For any $\j \in  \bam$, we define
\begin{equation}\label{We}
\W_\eta(\j) = \limsup_{R\to \infty} \dashint_{K_R}  |\j_\eta|^2 - m(\kappa_d  \g(\eta)+\gamma_2 \indic_{d=2})
\end{equation}
and the renormalized energy is given by \footnote{As in \cite{gl13} we could define the renormalized energy with averages on more general nondegenerate (Vitali) shapes, such as balls, etc, and then prove that the minimum of $\W$ does not depend on the shapes.}
\begin{align}\label{weta}
\W(\j) = \liminf_{\eta\to 0}\W_\eta(\j) \nonumber =\liminf_{\eta\to 0}\left(\limsup_{R\to \infty} \dashint_{K_R}   |\j_\eta|^2 - m(\kappa_d \g(\eta)+\gamma_2 \indic_{d=2})\right).
\end{align} 
\end{defi}
It is easy to see that if $\j \in \bam$, then $\j' := m ^{1/d - 1} \j (m ^{-1/d} .)$ belongs to $\bai$ and 
\begin{equation}\label{eq:scale renorm}
\begin{cases}
\W_\eta (\j) = m ^{2-2/d} \mathcal{W}_{\eta m^{1/d}}  (\j ') & \  \W (\j) = m ^{2-2/d} \W (\j ')  \: \mbox{ if } d\geq 3 \\
\W_\eta (\j) = m \left( \mathcal{W}_{\eta m^{1/d}}  (\j ')-\frac{\kappa_2}{2} \log m \right)  & \ \W  (\j) = m \left(\W (\j ')-\frac{\kappa_2 }{2}\log m \right)  \: \mbox{ if } d=2,
\end{cases}
\end{equation}thus the same scaling formulae hold for $\inf_{\bam} \W$.
 One may thus reduce to the study of $\W(\j)$ on $\bai$, for which we have the following result:

\begin{theo}[\textbf{Minimization of the renormalized energy}]\label{thm:renorm ener}\mbox{}\\
The infimum 
\begin{equation}\label{eq:min ren ener}
\alpha_d:= \inf_{\j \in \bai} \W (\j) 
\end{equation}
is achieved and is finite. Moreover, there exists a sequence $(\j_n)_{n\in \mathbb{N}}$ of periodic vector fields (with diverging period in the limit $n\to \infty$) in $\bai$ such that 
\begin{equation}\label{eq:min period}
\W (\j_n) \to \alpha_d \mbox{ as } n \to \infty. 
\end{equation}
\end{theo}

We should stress at this point that the definition of the jellium (renormalized) energy we use is essential for our approach to the study of equilibrium states of \eqref{wn}. In particular it is crucial that we are allowed to define the energy of an \emph{infinite} system via a local density (square norm of the electric field). For a different definition of the jellium energy, the existence of the thermodynamic limit was previously proved in \cite{LN}, using ideas from \cite{LiLe1,LiLe}. This approach does not transpose easily in our context however, and our proof that $\alpha_d$ is finite follows a different route, using in particular an unpublished result of Lieb \cite{Lie}. The existence of a minimizer is in fact a consequence of our main results below (see the beginning of Section \ref{sec7}). A direct proof is  also provided in Appendix \ref{sec:appendix} for convenience of the reader.

In \cite{gl13,ss2d} a different but related strategy  was used for the definition of $W$ (in dimension~2). Instead of smearing charges out, the ``renormalization" was implemented by cutting-off the electric field $\j$ in a ball of radius $\eta$ around each charge, as in \cite{BBH}.
This leads to the following definition, which we may present in arbitrary dimension (the normalizing of constants has been slightly modified in order to  better fit with the current setting):

\begin{defi}[\textbf{Alternative definition \cite{gl13,ss2d}}]\label{def:renorm SS}\mbox{}\\
We let $\admissible $ denote the subclass of $\bam$ for which all the points are simple (i.e. $N_p=1$ for all $p \in \Lambda$.) For each $\j \in \admissible$, we define 
\begin{equation}\label{WWold}W(\j)= \limsup_{R\to \infty} \frac{1}{|K_R|} W(\j, \chi_{K_R})\end{equation}
where, for any $\j$  satisfying a relation of the type \eqref{dij}, and any nonnegative continuous function $\chi$, we denote
\begin{equation}\label{Wold}
W(\j, \chi)= \lim_{\eta\to 0} \int_{\mr^d \backslash \cup_{p\in \Lambda } B(p,\eta)} \chi |\j|^2 -  c_d \g(\eta) \sum_{p \in \Lambda} \chi(p),
\end{equation}and $\{\chi_{K_R}\}_{R>0}$ denotes a family of cutoff functions, equal to $1$ in $K\widetilde{\W}_{R-1}$, vanishing outside $K_R$ and of universally bounded gradient.
\end{defi}

In addition to the way the renormalization is performed, between the two definitions the order of the limits $\eta \to 0$ and $R \to \infty$ is reversed. It is important to notice that for a given discrete configuration, the minimal distance between points is bounded below on each compact set, hence the balls $B(p,\eta)$ in $K_R$ appearing in \eqref{Wold} are disjoint as soon as $\eta $ is small enough, for each fixed $R$. On the contrary, the smeared out charges $\delta_p^{(\eta)}$ in the definition of $\j_\eta$ may overlap. In fact, we can prove (see Section \ref{sec3}) at least in dimension 2, that   if the distances between the points is bounded below by some uniform constant (not depending on $R$) --- we will call such points ``well separated", in particular they must all be simple ---, then the order of the limits $\eta\to 0$ and $R\to \infty$ can be reversed and $W $ and $\W$ coincide (in addition the value of $\W$ then does not depend on the choice of $\ro$ by Newton's theorem).  An easy example is 
the case of a configuration of points which is periodic. But if the points are not well separated, then in general $W$ and $\W$ do not coincide (with typically $\W\le W$). An easy counter example of this is the case of a configuration of well-separated points except  one multiple point. Then computing the limit $\eta \to 0$ in $W$ immediately yields $+ \infty$. On the contrary, the   effect of the multiple point gets completely dissolved when taking first the limit $R \to \infty$ in the definition of $\W$. At least an immediate consequence of \eqref{eq:gse second order} (or Theorem \ref{th1} below), by comparison with the result of \cite{ss2d}, is that  we know that  $\min_{\bai} \W=\min_{\mathcal A_1} W$ in dimension 2.

One of the advantages of $W$ is that it is more precise: in 2D it can be derived as the  complete $\Gamma$-limit of $\w$ at next order, while $\W$ is not (it is too ``low"), see below. It also seems more amenable to the possibility of charges of opposite signs (for example it is derived in \cite{gl13} as the limit of the vortex interactions in the Ginzburg-Landau energy where vortices can a priori be positive or negative). The main advantage of $\W$ is that it is much easier to bound it from below: $\W_\eta$ is bounded below from its very definition, while proving that $W$ is turned out to be much more delicate. In \cite{gl13} it was proven that even though the energy density that defines $W(\j, \chi)$ is unbounded below, it can be shown to be very close to one that is. This was a crucial point in the proofs, since most lower bound techniques (typically Fatou's lemma) require energy densities that are bounded below.  The proof in \cite{gl13} relied on the sophisticated techniques of ``ball construction 
methods", which originated with Jerrard \cite{jerrard} and Sandier \cite{sandier} in Ginzburg-Landau theory (see \cite[Chapter 4]{SSbook} for a presentation) and which are very much two-dimensional, since they exploit the conformal invariance of the Laplacian in dimension~$2$. It is not clear at all how to find a replacement for this ``ball construction method" in dimension $d \ge 3$, and thus not clear how to prove that the local energy density associated to $W$ is bounded below then (although we do not pursue this, it should however be possible to show that $W$ itself is bounded below with the same ideas we use here). Instead the approach via $\W$ works by avoiding this issue and replacing it with the use of Onsager's lemma, making it technically much simpler, at the price of a different definition and a less precise energy. However we saw that $W$ and $\W$ have same minima (at least in dimension $2$) so that we essentially reduce to the same limit minimization problems.

\subsection{Main results: ground state}\label{sec:results 0 T}

Our results on the ground state put on a rigorous ground the informal interpretation of the two-scale structure of minimizing configurations we have been alluding to in the introduction. To describe the behavior of minimizers at the microscopic scale we follow the same approach as in \cite{ss2d} and perform a blow-up: For a given $(x_1, \dots , x_n)$, we let $x_i'= n^{1/d} x_i$ and
\begin{equation}
\label{hpn}
h_n'(x')=  w \ast \left(\sum_{i=1}^n \delta_{x_i'} - \mu_0(n^{-1/d}x')\right),
\end{equation}
Note here that the associated electric field $\nab h_n'$ is in $\Lp$ if and only if $p<\frac{d}{d-1}$, in view of the singularity in $\nab w$  around each point. One of the delicate parts of the analysis (and of the statement of the results) is to give a precise averaged formulation, with respect to all possible blow-up centers in $\E$, and in this way to give a rigorous meaning to the  vague sentence ``around almost any point
$x$ in the support of $\mu_0$ there are approximately $\mu_0(x)$ points per unit volume, minimizing $\W$  with background density $\mu_0(x)$''. 
% First of all it is clear that perturbing the minimizer in a neighborhood of a single point $x$ will not change the energy drastically, so the preceding sentence should rather begin with ``around \emph{most} points $x$ ...'' and we should give a precise meaning to ``most''. Second, as we saw in the preceding section, the jellium energy is most conveniently defined from the electric field of a configuration so the second part of the sentence should be interpreted as ``the blown-up electric field around $x$ minimizes the renormalized energy''. 

Our formulation of the result uses the following notion, as in  \cite{gl13,ss2d}, which allows to embed $(\mr^d)^n$ into the set of probabilities on $X=\E\times\Lp$, for some $1<p<\frac{d}{d-1}$.  For any $n$ and $\xbf =(x_1,\dots, x_n)\in(\mr^d)^n$ we let $i_n(\xbf) = P_{\nu_n},$
where $\nu_n=\sum_{i=1}^n \delta_{x_i}$ and $P_{\nu_n}$ is the push-forward of the normalized Lebesgue measure on $\E$ by
$$x\mapsto \left(x, \nab h_n'(n^{1/d}x + \cdot)\right).$$ 
Explicitly:
\begin{equation}\label{eq:Pnun}
i_n (\xbf) = P_{\nu_n} = \dashint_{\E} \delta_{(x,\nabla h'_n (n ^{1/d} x + \cdot) )} dx.
\end{equation}
This way  $i_n(\xbf)$ is an element of $\P(X)$, the set of probability measures on $X = \E\times\Lp$ (couples of blown-up centers, blown-up electric fields) that measures  the probability of having a given blown-up electric field around a given blow-up point in $\E$. As suggested by the above discussion, the natural object we should look at is really $i_n (\xbf)$ { and its limits up to extraction,~$P$}. 

Due to the fact that the renormalized jellium functional describing the small-scale physics is invariant under translations of the electric field, and by definition of $i_n$, we should of course expect the objects we have just introduced to have a certain translation invariance, formalized as follows:

\begin{defi}[\textbf{$T_{\lambda(x)}$-invariance}]\label{invariance}\mbox{}\\
We say a probability measure $P$ on $X$ is $T_{\lambda(x)}$-invariant if $P$ is invariant  by $(x,\j)\mapsto\left(x,\j(\lambda(x)+\cdot)\right)$, for any $\lambda(x)$ of class $C^1$  from $\E$ to $\mr^d$.
\end{defi}
Note that $T_{\lambda(x)}$-invariant implies translation-invariant (simply take $\lambda \equiv 1$). 

\begin{defi}[\textbf{Admissible configurations}]
We say that $P \in \mathcal{P}(X)$ is admissible if its first marginal is the normalized Lebesgue 
measure on $\E$, if it holds for $P$-a.e. $(x,\j)$ that $\j \in \overline{\mathcal{A}}_{\mu_0(x)}$, and if $P $ is $T_{\lambda(x)}$-invariant.
\end{defi}

Assumption~\ref{ass2} ensures here that $\mu_0 (x)>0$  for any $x\in\Sigma$, so that the class $\overline{\mathcal{A}}_{\mu_0(x)}$ makes sense.

Our main result  on the next order behavior  of minimizers of $\w$ is  that 
$$\min_{\xbf} n^{2/d-2} \left(\w (\xbf)- n^2 \I[\mu_0] +\Big(\frac{ n}{2} \log n\Big)\indic_{d=2}\right) {\rightarrow} \min_{P \ \text{admissible}} \widetilde{\W} (P)$$
 where  we define
\begin{equation}\label{wtilde}
\widetilde{\W}(P):= \frac{|\E|}{c_d} \int_{X} \W(\j) \, dP(x,\j),
\end{equation} 
if $P$ is admissible, and $+\infty$ otherwise.
The function $\widetilde{\W}$ is precisely what we meant by ``the average with respect to blow-up centers" of $\W$.

We note that, in view of the scaling relation \eqref{eq:scale renorm}, we may guess that 
\begin{equation}\label{eq:gamma d}
\min_{P \ \text{admissible}} \widetilde{\W}=
\xi_d:=  \begin{cases}\displaystyle
             \frac{1}{c_d} \min_{\bai}  \W  \int_{\R   ^d} \mu_0 ^{2-2/d} \quad \mbox{if } d\geq 3\\
             \displaystyle \frac{1}{2\pi} \min_{\bai}  \W - \hal \int_{\R   ^2} \mu_0 \log \mu_0  \quad \mbox{if } d=2.
            \end{cases}
\end{equation} 
Note that the minima on the right-hand side exist thanks to Theorem~\ref{thm:renorm ener}. Also, the left-hand side is clearly larger than the right-hand side in view of the definition of $\widetilde{\W}$, that  of admissible, and \eqref{eq:scale renorm}. That there is actually equality is a consequence of the following theorem
on the behavior of ground state configurations (i.e. minimizers of \eqref{wn}) :

\begin{theo}[\textbf{Microscopic behavior of ground state configurations}]\label{th1}\mbox{}\\
Let $(x_1, \dots, x_n)\in (\mr^d)^n$ minimize $\w$. Let $h_n' $ be associated via \eqref{hpn} and  $P_{\nu_n}= i_n(x_1, \dots, x_n) $ be defined as in \eqref{eq:Pnun}.  

Up to extraction of a subsequence we have $P_{\nu_n} \wto P$ in the sense of probability measures on $X$, {where $P$ is admissible and, if
 $d\ge 3$,
\begin{equation}\label{eq:intro 222}
\lim_{n\to \infty} n^{2/d-2} \left(\w(x_1, \dots, x_n) -n^2\I[\mu_0] \right)  =   \widetilde{\W}(P) = \xi_d=  \displaystyle\frac{1}{c_d} \min_{\bai}  \W \int    \mu_0^{2-2/d}(x) \, dx,
\end{equation}
respectively for $d=2$,
\begin{multline}\label{eq:intro 223}
\lim_{n\to \infty} n^{-1}\left(\w(x_1, \dots, x_n) -n^2\I[\mu_0]+ \frac{n}{2}\log n \right)\\ = \widetilde{\W}(P) =\xi_2= \frac{1}{c_2}\min_{\bai}  \W -\frac{1}{2} \int \mu_0(x) \log \mu_0(x) \,dx,
\end{multline}}
  $P$ is a minimizer of $\widetilde{\W}$ and $\j$ minimizes $\W$ over $\overline{\mathcal{A}}_{\mu_0(x)}$  for $P$-a.e. $(x,\j)$.
%\item We have $\sum_{i=1}^n \zeta(x_i) = o(n^{1-2/d}),$ as $n \to \infty$.
\end{theo}
Note that the configuration of points itself depends on $n$, i.e. we have slightly abused notation by writing $(x_1, \dots, x_n)$ instead of $(x_1^n, \dots, x_n^n)$.

The statement above implies \eqref{eq:gse second order} and  is a  rigorous formulation of our vague sentence in the introduction about the double scale nature of the charge distribution. Note  that since $P$ is always admissible, the result   includes the fact that the blown-up field around $x$ is in the class $\overline{\mathcal{A}}_{\mu_0(x)}$, i.e. has a local density $\mu_0(x)$.

 If the crystallization conjecture is correct, and if the only minimizers of $\W$ are periodic configurations, this means that after blow-up around (almost) any point in $\E$, one should see a crystalline configuration of points, the one minimizing $\W$, packed at the scale corresponding to $\mu_0(x)$.   Note that in \cite{rns} a stronger result is proved in dimension $2$: exploiting completely the minimality  of the configuration, it is shown that this holds after blow up around {\it any} point in $\E$  (not too close to $\p \E$ however), and that the renormalized energy density as well as the number of points are equidistributed (modulo the varying density $\mu_0(x)$). The same results are likely to be proveable in dimension $d\ge 3$ combining the present approach and the method of \cite{rns}.

The proof of Theorem \ref{th1} relies on the following steps:
\begin{itemize}
\item First we show the result as a more general lower bound, that is, we show that  for an arbitrary configuration $(x_1, \dots, x_n)$, it holds that {
$$ \liminf_{n\to \infty} n^{2/d-2} \left(\w(x_1, \dots, x_n) -n^2\I[\mu_0] + \Big(\frac{n}{2}\log n\Big)\indic_{d=2}\right) \ge   \widetilde{\W}(P) - o(\eta)$$}
with $P = \lim_{n\to \infty} P_{\nu_n}$. This relies on the splitting formula alluded to before, based on Onsager's lemma, and the general  ergodic framework introduced in \cite{gl13, ss2d} and suggested by Varadhan, which allows to bound first from below by $\W_\eta$ instead of $\W$, for fixed $\eta$. 
\item  Taking the $\liminf_{\eta \to 0}$ requires showing that $\W_\eta$ is bounded from below by a constant independent of $\eta$ (and then using Fatou's lemma).
This is accomplished in a crucial step, which consists in proving that minimizers of $\W_\eta$ have points that are well-separated: theirdistances are bounded below by a constant depending only on $\|\mu_0\|_{L^\infty}$. This relies on an unpublished result of Lieb \cite{Lie}, which can be found in dimension $d=2$ in \cite[Theorem 4]{rns} and which we readapt here to our setting.
This allows to  complete the proof of the lower bound.
\item We prove  the corresponding upper  bound inequality only at the level of the minimal energy.  The reason is that the proof relies on an explicit construction, based on the result of the previous step: we take a minimizer of $\W$ and we ``screen" it, as done in \cite{gl13,ss2d,ss1d}. This means we truncate it over some large box and periodize it, in order to be able to copy-paste it after proper rescaling in order to create an optimal configuration of points. This screening uses crucially  the preliminary result that the  points are well separated. 
At this point, we see that arbitrary configurations cannot  be screened:  for example if a configuration has a multiple point, then truncating and periodizing it would immediately result in an infinite $\W$, and obviously to an infinite $\w$!
We note also that if the upper bound inequality was holding for any general admissible $P$,  this would mean that {$\widetilde{ \W}$} is the complete $\Gamma$-limit (at next order) of $\w$. But in dimension $2$ at least, it has been  proven in \cite{ss2d} that the corresponding $\Gamma$-limit is  {$\widetilde{ \W}$}. By uniqueness of the $\Gamma$-limit, this would mean that $\int W\, dP= \int \W \, dP$ for all admissible probabilities $P$. But we know that  (at least in dimension 2)  $W $ and $\W$ are not equal, so there remains to know whether their average over admissible probabilities can coincide.
 
 \item Combining the lower bound and the matching upper bound immediately yields the result for minimizers.
\end{itemize}

\subsection{Main results: finite temperature case}\label{sec:results finite T}

We now turn to the case of positive temperature and study the Gibbs measure \eqref{eq:defi Gibbs state}. We shall need the common assumption that there exists $\beta_1>0$ such that
\begin{equation}\label{integrabilite}
\begin{cases}
\int e^{-\beta_1 V(x)/2} \, dx <\infty  & \mbox{ when } d\geq 3\\
\int e^{-\beta_1( \frac{V(x)}{2}- \log |x|)} \, dx <\infty & \mbox{ when } d=2. 
\end{cases}
\end{equation}
Note that this is only a slight strengthening of the assumption \eqref{eq:trap pot}.

By definition, the Gibbs measure minimizes the $n$-body free energy functional
\begin{equation}\label{eq:free ener N}
\Fnbeta [\mubf] := \int_{\R ^{dn}} \mubf(\xbf) H_n (\xbf) d\xbf + \frac{2}{\beta} \int_{\R ^{dn}} \mubf(\xbf) \log (\mubf (\xbf)) d\xbf
\end{equation}
over probability measures $\mubf \in \P (\R ^{dn})$. It is indeed an easy and fairly standard calculation to show that the infimum is attained\footnote{Remark that the symmetry constraint \eqref{eq:symmetry} is automatically satisfied even if do not impose it in the minimization.} at $\Q$ (cf. \eqref{eq:defi Gibbs state}) and that we have 
\begin{equation}\label{eq:free ener N min}
\Fnbetae := \inf _{\mubf \in \P (\R ^{dn})} \Fnbeta [\mubf]  =  \Fnbeta [\Q] = - \frac{2}{\beta} \log \Z.
\end{equation}
where $\Z$ is the partition function normalizing $\Q$.

\medskip

As announced in the introduction, we will mainly focus on the low temperature regime $\beta \gtrsim n ^{2/d-1}$ where we can get more precise results. We expect that crystallization occurs when $\beta \gg n ^{2/d-1}$ and that $\beta\propto n ^{2/d-1}$ should correspond to the liquid-crystal transition regime. Since they are the main basis of this conjecture, let us first state precisely our estimates on the $n$-body free energy defined above. In view of \eqref{eq:free ener N min}, estimating $\log \Z$ and $\Fnbetae$ are equivalent and we shall work with the later. We need to introduce the mean field free energy functional 
\begin{equation}\label{eq:MF free ener func}
\F [\mu] =  \En [\mu] + \frac{2}{n\beta} \int \mu \log \mu
\end{equation}
with minimizer, among probabilities,  $\mu_{\beta}$. It  naturally arises when taking $\mu ^{\otimes n}$ as a trial state in \eqref{eq:free ener N}. Note that for $\beta n \gg 1$, this functional is just a perturbation of \eqref{eq:def MF ener} and $\mubet$ agrees with $\mu_0$ to leading order (results of this kind are presented in \cite[Section 3]{RSY2}). For convenience we will restrict the high temperature regime in Theorem \ref{thm:partition} to $\beta \geq C n ^{-1}$ in order that the last term (entropy) in the mean-field free energy functional stays bounded above by a constant \footnote{The opposite case is in fact somewhat easier, see \cite[Section 3]{RSY2}.}. The results below do not depend on $C$. 

\begin{theo}[\textbf{Free energy/partition function estimates}]\label{thm:partition}\mbox{}\\
The following estimates hold.
\begin{enumerate}
\item (Low temperature regime).  {Let $\bar{\beta}:=\limsup_{n\to+\infty} \beta n^{1-2/d}$, and assume $0<\bar{\beta}\le + \infty$. There exists $C_{\bar{\beta}}>0$ depending only on $V$ and $d$, with $\lim_{\beta \to \infty} C_{\bar{\beta}}=0$ and  $C_{\infty}=0$  such that  
\begin{equation}\label{eq:partition low T}
\limsup_{n\to \infty}  n^{2/d-2} \left| \Fnbetae - n ^2 \En [\mu_0]  - n ^{2-2/d} \xi_d    \right| \leq    C_{\bar{\beta}},
\end{equation} } 
respectively if $d=2$ and $\beta \geq c (\log n) ^{-1} $ for $n$ large enough with $c>0$,  
\begin{equation}\label{eq:partition low T 2d}
\limsup_{n\to \infty} n^{-1} \left| \Fnbetae - n ^2 \En [\mu_0] + \frac{n}{2}\log n - n\xi_2  \right| \leq  C_{\bar{\beta}} ,
\end{equation} where $\xi_d$ is as in \eqref{eq:gamma d}.
\item (High temperature regime). If $d\geq 3$ and $ c_1 n ^{-1} \leq \beta \leq c_2 n ^{2/d - 1}$ for some $c_1,c_2>0$,  we have for $n$ large enough,
\begin{equation}\label{eq:partition high T}
\left| \Fnbetae - n ^2 \En [\mubet]  - \frac{2 n}{\beta} \int_{\R ^d} \mubet \log \mubet \right| \leq C n ^{2-2/d},
\end{equation} 
respectively if $d=2$ and $c_1 n ^{-1} \leq \beta \leq c_2 (\log n) ^{-1}$ for some $c_1, c_2>0$
\begin{equation}\label{eq:partition high T 2d}
\left| \Fnbetae - n ^2 \En [\mubet] - \frac{2 n}{\beta} \int_{\R ^d} \mubet \log \mubet \right| \leq C n \log n,
\end{equation} 
where $C$ depends only on $V$ and $d$.
\end{enumerate}
\end{theo}

The leading order contribution to \eqref{eq:partition low T}-\eqref{eq:partition high T 2d} has been recently derived, along with the corresponding large deviation principle,  in \cite{CGZ} under the assumption $\beta \gg n ^{-1} \log n $ in our units. In this regime one may use either $\En[\mubet]$ or $\En[\mu_0]$ for leading order considerations, but when $\beta \propto n ^{-1}$ it is necessary to use the former.

In the regime $\beta \gg n ^{2/d - 1}$ \eqref{eq:partition low T} says that the free energy agrees with the ground state energy (compare with \eqref{eq:gse second order}) up to a negligible remainder, and this is the only case where we prove the existence of a  thermodynamic limit. We conjecture that this corresponds to a transition to a crystalline state. On the other hand, \eqref{eq:partition high T} shows that for $\beta \ll n ^{2/d - 1}$ the effect of the  entropy at the macroscopic scale prevails over the jellium energy at the microscopic scale. This alone is no indication that there is no crystallization in this regime: One may perfectly well imagine that a term similar to the subleading term in \eqref{eq:gse second order} appears at a further level of approximation (one should at least replace $\mu_0$ by $\mu_{\beta}$ but that would encode no significantly different physics at the microscopic scale). We 
believe that this is not the case and that the entropy should also appear at the microscopic scale, but this goes beyond what we are presently able to prove. Note that when $d\ge 3$, the fixed  $\beta $ regime is deep in the low-temperature regime, and so we do not expect any particular transition to happen at this order of inverse temperature.

The case $d=2$, in which we re-obtain the result of \cite{ss2d} is a little bit more subtle due to the particular nature of the Coulomb kernel, which is at the origin of the $n\log n$ term in the expansion. As in higher dimensions, comparing \eqref{eq:gse second order} and \eqref{eq:partition low T 2d} we see that free energy and ground state energy agree to subleading order when $\beta \gg 1$, which we conjecture to be the crystallization regime. The estimate \eqref{eq:partition high T 2d} shows that the entropy is the subleading contribution only when $\beta \ll (\log n) ^{-1}$ however. It is not clear from Theorem \ref{thm:partition} what exactly happens when $(\log n) ^{-1} \lesssim \beta \lesssim 1$. We expect  entropy terms at both macroscopic and microscopic levels to enter. For the regime $\beta \propto 1$ (the most studied regime, e.g. \cite{BZ}), one can see the further discussion in \cite{ss2d} and Theorem \ref{th4} below.

\medskip 

Our next result exposes the consequences of \eqref{eq:partition low T}-\eqref{eq:partition low T 2d} for the Gibbs measure itself. Roughly speaking, we prove that it charges only configurations whose renormalized energy  {$\widetilde{\W}$} is below a certain threshold, in the sense that other configurations have exponentially small probability in the limit $n\to \infty$. This is a large deviations type upper bound at speed $n^{2-2/d}$, but a complete large deviations principle is missing.
The threshold of renormalized energy vanishes in the limit $\beta \gg n ^{2/d-1}$, showing that the Gibbs measure charges only configurations that minimize the renormalized energy at the microscopic scale, which is a proof of crystallization, modulo the question of proving that minimizers of $\W$ are really crystalline. 
Our precise statement is {again} complicated by the double scale nature of the problem and the fact that the renormalized energy takes electric fields rather than charge configurations as argument. Its phrasing uses the same framework as Theorem \ref{th1} {in terms of the limiting probability measures $P$, which should now be seen as random, due to the temperature. In fact we also prove the existence of a limiting ``electric field process", as $n\to \infty$, i.e. a limiting probability on the $P$'s, which is the limit of the push-forwards of $\Q$ by $i_n$.}

  We will consider the limit of the probability that the system is in a state $(x_1,\ldots,x_n)\in A_n$ for a given sequence of sets $A_n \subset (\R ^{d})^n$ in configuration space. Associated to this sequence we introduce
\begin{equation}\label{ainf} 
A_\infty = \left\{ P \in \P(X) : \exists \xbf_n \in A_n, \, P_{\nu_n} \wto P \  \text{up to a subsequence}\right\}
\end{equation}
where $P_{\nu_n}$ is as in \eqref{eq:Pnun} with $\xbf=\xbf_n$ and the convergence is weakly as measures.  $A_\infty$ can be thought of as the limsup in the sense of sets of the sequence of sets $i_n(A_n)$.

\begin{theo}[\textbf{Microscopic behavior of thermal states in the low temperature regime}]\label{th4}
For any $n>0$ let $A_n\subset (\mr^d)^n$ and $A_\infty$ be as above.
{Let $\bar{\beta}>0 $ be as in Theorem~\ref{thm:partition}, $\xi_d$ as in~\eqref{eq:gamma d}. There exists $C_{\bar{\beta}}$ depending only on $V$ and $d$ such that $C_{\bar{\beta}}=0$ for $\bar{\beta} = \infty$ and 
\begin{equation}
\label{ldr}
\limsup_{n\to \infty} \frac{ \log \Q(A_n)}{n ^{2-2/d}} \le -\frac{\beta}{2} \left(\inf_{P\in A_\infty}\widetilde{\W}(P)  - \xi_d  - C_{\bar{\beta}}  \right).
\end{equation}}
Moreover, for fixed $\beta>0$,  letting $\widetilde{ \Q}$ denote the push-forward of $\Q$ by $i_n$ (defined in \eqref{eq:Pnun}), $\{\widetilde{\Q}\}_{n}$ is tight and converges as $n\to \infty$, up to a subsequence, to a probability measure on $\mathcal P(X)$ which is concentrated on admissible probabilities satisfying  $\widetilde{\W}(P) \le \xi_d + C_{\bar{\beta}}$. 
\end{theo}
Note that the error term {$C_{\bar{ \beta}}$} becomes negligible when $\beta \gg n ^{2/d-1}$. This means as announced that the Gibbs measure concentrates on minimizers of  {$\widetilde{\W}$} in that regime. When $\beta =\bar{\beta} n^{2/d-1}$, we have instead a threshold phenomenon: the Gibbs measure concentrates on configurations whose {$\widetilde{ \W} $} is below the minimum plus {$C_{\bar{\beta}}$}.
 For the 2D case, a partial converse to \eqref{ldr} is proved in \cite{ss2d}, establishing a kind of large deviation principle at speed $n$ in the regime $\beta \gg 1$, but with $W$ instead of $\W$.  A full LDP should involve an additional entropy term and  is still missing. 
 % Due to the fact that 
%{$\widetilde{\W}$} is not expected to be the full $\Gamma$-limit of $\w$, such a converse cannot be proven here, and one can expect that a complete LDP should involve $W$ instead of $\W$.
 
 \subsection{Estimates on deviations from mean-field theory}\label{sec:results deviations}
  
Our methods also yield some quantitative estimates on the fluctuations from the equilibrium measure. We again focus on the low temperature regime for concreteness but if need arises our methods could work just as well in the opposite regime, with less hope of optimality however. The following estimates are insufficient for a crystallisation result, as noted before, but they set bounds on the possible deviations from mean-field theory in a stronger sense than what has been considered so far, namely by estimating charge deviations and  correlation functions.

We will estimate the following quantity, defined for any configuration $(x_1,\ldots,x_n)\in\R^{dn}$ 
\begin{equation}\label{eq:intro charge fluctu}
D(x,R):=  \nu_n (B(x,R)) - n \int_{B(x,R)} \mu_0, 
\end{equation}
i.e. the deviation of the charge contained in a ball $B(x,R)$ with respect to the prediction drawn from the equilibrium measure. Here 
\begin{equation}\label{eq:empirical measure}
\nu_n  = \sum_{i=1} ^n  \delta_{x_i}.
\end{equation}
Note that since $\mu_0$ is a fixed bounded function, the second term in \eqref{eq:intro charge fluctu} is of order $n R^d$, and typically the difference of the two terms is much smaller than this, since the distribution of points at least approximately follows the equilibrium measure.

For pedagogical purposes we only state two typical results:
\begin{enumerate}
\item the probability of a large deviation of the number of points in a ball $B(x,R)$ (i.e. a deviation of order $n R ^d$) is exponentially small as soon as $R\gtrsim n ^{-1/(d+2)}$.
\item the probability of a charge deviation of order $n ^{1-1/d}$ in a macroscopic ball $B(x,R)$ with $R = O(1)$ is exponentially small.
\end{enumerate}
We deduce these results from sharper but more complicated estimates similar to \cite[Eq. (1.47) and (1.49)]{ss2d}, see Section \ref{sec:fluctu charge}.  Actually these contain a whole continuum of results of the type above: the larger the scale at which one considers the deviations, the smaller they are allowed to be. If need be, any scale $R$ satisfying $n ^{-1/(d+2)} \lesssim R \leq C$ can be considered, and deviations on this scale are exponentially small, with a rate depending on the scale. We believe that the two results mentioned above are a sufficient illustration of our methods and we proceed to state them rigorously, along with an estimate on the discrepancy between $\nu_n$ and the equilibrium measure $n\mu_0$ in weak Sobolev norms. Below $W^{-1, q}(\Omega)$ denotes the dual of the Sobolev space $W^{1,q'}_0(\Omega)$ with $1/q+1/q'=1$, in particular $W^{-1,1}$ is the dual of Lipschitz functions.

\begin{theo}[\textbf{Charge fluctuations}]\label{thm:charge fluctu}\mbox{}\\
Assume there is a constant $C>0$ such that $\beta \geq C n ^{2/d - 1}$. Then the following holds for any $x\in \R ^d$.
\begin{enumerate}
\item (Large fluctuations on microscopic scales). Let $R_n$ be a sequence of radii satisfying $R_n \geq C_R \: n ^{-\frac{1}{d+2}}$ for some  constant $C_R$. Then, for any $\lambda>0$ we have, for $n$ large enough,
\begin{equation}\label{eq:fluctu charge micro}
\Q \left( |D(x,R_n)| \geq \lambda n R_n ^{d}\right) \leq C e ^{-C \beta n ^{2-2/d}(C_R\lambda^2 - C)},
\end{equation}
for some $C$ depending only on dimension.
\item (Small fluctuations at the macroscopic scale). Let $R>0$ be a fixed radius. There is a constant  $C$ depending only on dimension such that for any $\lambda>0$, for $n$ large enough, 
\begin{equation}\label{eq:fluctu charge macro}
\Q \left( |D(x,R)| \geq  \lambda n^{1-1/d} \right) \leq C e ^{-C \beta n ^{2-2/d}( \min(\lambda^2  R^{2-d}, \lambda^4 R^{2-2d})   -C) }. 
\end{equation}
\item (Control in weak Sobolev norms). Let $R>0$ be some fixed radius. Let $1 \leq q < \frac{d}{d-1}$ and define
\begin{align*}
t_{q,d} &= 2 - \frac{1}{d} - \frac{1}{q}<1\\
\tilde{t}_{q,d} &= 3   -\frac{1}{d} - \frac{2}{dq}  >0.
\end{align*}
There is a constant $C_R>0$ such that the following holds for $n$ large enough, and any $\lambda$ large enough, 
\begin{equation}\label{eq:fluctu field}
\Q \left( \left\Vert \nu_n - n\mu_0 \right\Vert_{W^{-1,q} (B_R)} \geq \lambda n ^{t_{q,d}}\right) \leq C e ^{- \beta C_R \lambda^2 n ^{ \tilde{t}_{q,d}}}.
\end{equation}
\end{enumerate}
\end{theo}

As regards \eqref{eq:fluctu field}, note that only estimates in $W^{-1,q}$ norm with $q<\frac{d}{d-1}$ make sense, since a Dirac mass is in $W^{-1,q}$  if and only if $q<\frac{d}{d-1}$  (hence the same for $\nu_n$). In view of the values of the parameters $t_{q,d}$ and $\tilde{t}_{q,d}$ defined above, our results give meaningful estimates for any such norm: a large deviation in a ball of fixed radius would be a deviation of order $n$. Since $t_{q,d}<1$, equation \eqref{eq:fluctu field} above implies that such deviations are exponentially unlikely in $W^{-1,q}$ for any $q<\frac{d}{d-1}$, in particular in $W^{-1,1}$ for any space dimension.

The proof of Item 3 is based on a control in $W^{-1,2}$ of $\tilde{\nu}_n - \mu_0$, where $\tilde{\nu}_n$ is the regularization of $\nu_n$ with charges smeared-out on a scale $n ^{-1/d}$. This is another instance where the method of smearing out charges makes the proof siginificantly easier than in \cite{ss2d}, replacing the use of a ``displaced" energy density \cite{SSmass} and Lorentz space estimates \cite{SeTi}. 

\medskip

Finally we state some consequences for the marginals (reduced densities) of the probability \eqref{eq:defi Gibbs state}. Let us denote 
\begin{equation}\label{eq:marginal Gibbs}
\Qk (x_1,\ldots, x_k) = \int_{ \xbf' \in \R ^{d(n-k)}}  \Q (x_1,\ldots, x_k,\xbf') \, d\xbf'.
\end{equation}
Remark that since $\Q$ is symmetric w.r.t. exchange of variables, it does not matter over which $n-k$ particles we integrate to define $\Qk$ (particles are indistinguishable). The value $\Qk (x_1,\ldots, x_k)$ is interpreted as the probability density for having one particle at $x_1$, one particle at $x_2$, $\ldots$, and one particle at $x_k$. 

\begin{coro}[\textbf{Marginals of the Gibbs measure in the low temperature regime}]\label{thm:marginals}\mbox{}\\
Let $R>0$ be some fixed radius, $1\leq q<\frac{d}{d-1}$. Under the same assumptions as Theorem \ref{thm:charge fluctu}, there exists a constant $C>0$ depending only on the  dimension such that the following holds:  
\begin{itemize}
\item (Estimate on the one-particle reduced density).
\begin{equation}\label{eq:result marginal 1}
\left\Vert \Qone - \mu_0 \right\Vert_{W^{-1,q} (B_R)} \leq C n^{1-1/d-1/q}=o_n(1). 
\end{equation}
\item (Estimate on $k$-particle reduced densities). Let $k\geq 2$ and $\varphi : \R ^{dk} \mapsto \R$ be a smooth function with compact support, symmetric w.r.t. particle exchange. Then we have
\begin{multline}\label{eq:result marginal k}
\left| \int_{\R ^{dk}} \left( \Qk - \mu_0 ^{\otimes k}\right) \varphi \right| \leq C \left( k n ^{1-1/d-1/q} + k ^2 n ^{-1} \right) 
\\ \sup_{x_1 \in \R ^d} \ldots \sup_{x_{k-1} \in \R ^d }   \left\Vert \nabla \varphi (x_1,\ldots,x_{k-1}, \:. \: ) \right\Vert_{L ^p (\R ^d)},
\end{multline}
where $1/p = 1 - 1/q$.
\end{itemize}
\end{coro}

As our other results, Corollary \ref{thm:marginals} concerns the low temperature regime. One may use the same technique to estimate the discrepancy between the $n$-body problem and mean-field theory in other regimes. When $\beta n $ becomes small however (large temperature), in particular when $\beta \sim n ^{-1}$ so that entropy and energy terms in \eqref{eq:MF free ener func} are of the same order of magnitude, a different method can give slightly better estimates. Since this regime has been considered for related models in several works \cite{MS,Kie1,Kie2,CLMP}, it is worth mentioning that quantitative estimates in the spirit of \eqref{eq:result marginal 1}-\eqref{eq:result marginal k} can be obtained in total variation norm. We refer to Remark \ref{rem:TV estimates} in Section \ref{sec:partition} for details on this approach, which is in a slightly different spirit from what we have presented so far.

\subsection{Organization of the paper}

In Section \ref{sec4}, we smear out the charges at scale $\eta$ for fixed $\eta$ and we use Onsager's lemma to obtain a sharp ``splitting formula" for the Hamiltonian, in which the leading order and next order contributions decouple. 
We also obtain a control  on the charge fluctuations  and on the electric field in terms of the next order term in the Hamiltonian.

In Section \ref{sec5} we start taking the limit $\eta \to 0$ in the estimates of the previous section,   and prove the lower bound on the energy via the ergodic  framework of Varadhan presented in  \cite{gl13}. This assumes the lower bound on $\W_\eta$.

In Section \ref{sec:univ low bound}, we use the result of \cite{Lie} to reduce to points that are  well-separated and deduce a lower bound on $\W_\eta$ independent on $\eta$. This requires  a ``screening result", which will also be used in Section  \ref{sec7}, where we prove an upper bound for the minimal energy by the construction of a precise test-configuration.
This finishes the proof of Theorem \ref{th1}.

In Section \ref{sec8}, we apply all the previous results to the case with temperature, and deduce Theorems \ref{thm:partition}, \ref{th4} and \ref{thm:charge fluctu}.

Appendices \ref{sec:appendix} and \ref{sec:appendix 2} respectively contain a direct proof of the existence of a minimizer for the renormalized energy functional and a discussion of the relation of the two versions of the renormalized energy functional.

\section{Splitting formulae and control on fluctuations}\label{sec4}
 
In this section, we start to exploit the idea of smearing out the charges and Onsager's lemma in a way similar to \cite{RSY2}. We also explore easy corollaries that can be obtained for fixed $\eta$. Since the smearing our procedure turns out to have no effect on the energy of configurations with well-separated points we start by discussing these particular configurations in the following subsection.

\subsection{Preliminaries and well-separated configuraions}\label{sec3}

We start with a lemma that shows that if $\j\in \bam$ and $\W(\j) <\infty$ then the density of points is indeed equal to that of the neutralizing background, i.e. $m$.
From now on, for any $\j \in \bam $ we denote by $\nu$ the corresponding measure of singular charges, i.e. $\sum_{p\in\Lambda} N_p \delta_p$.

\begin{lem}[\textbf{Density of points in finite-energy configurations}]\label{lembornnu}\mbox{}\\
Let $\j \in \bam$ be such that $\W_\eta(\j) <\infty$ for some $\eta\le 1$, and let $\nu = -\div \j+m.$
Then we have $\lim_{R\to \infty}\frac{\nu(K_R)}{|K_R|} =m$. \footnote{We could easily prove the same result with averages on balls or other reasonable shapes.}
\end{lem}

\begin{proof}
First we show that 
\begin{equation}\label{contnu}
\nu(K_{R-2})\le  m |K_R|+ C R^{\frac{d-1}{2}} \|\j_\eta\|_{L^2(K_R)} \qquad 
\nu(K_{R+1}) \ge m |K_{R-1}| - C R^{\frac{d-1}{2}} \|\j_\eta\|_{L^2(K_R)}. \end{equation}
To prove this, first by a mean value argument, we   find 
$t \in [R-1,R]$ such that
\begin{equation}\label{bbord}
\int_{\p K_t} |\j_\eta|^2 \le  \int_{K_R} |\j_\eta|^2.\end{equation}
Let us next  integrate \eqref{delp} over $K_t$ and use Stokes's theorem to find
\begin{equation}
\int_{ K_t} \sum_{p \in \Lambda } N_p \delta_p^{(\eta)}  - m |K_t| =  - \int_{\p K_t} \j _\eta\cdot \vec{\nu},
\end{equation}where $\vec{\nu}$ denotes the outer unit normal.
Using the Cauchy-Schwarz inequality and \eqref{bbord}, we deduce that
\begin{equation}\label{feg}
\left|\int_{ K_t} \sum_{p \in \Lambda } N_p \delta_p^{(\eta)}  - m |K_t| \right|\le C R^{\frac{d-1}{2}} \| \j_\eta\|_{L^2(K_R)} .\end{equation}
Since $\eta\le 1$, by definition of $\nu$ and since the $\delta_p^{(\eta)}$ are supported in $B(p,\eta)$, we have $\nu(K_{R-2})\le \int_{ K_t} \sum_{p \in \Lambda } N_p \delta_p^{(\eta)} \le \nu(K_{R+1}) $ in view of the definition of $\nu= \sum N_p \delta_p$.  The claim \eqref{contnu} follows. 
Since $\W_\eta(\j) <+\infty$ we have $\int_{K_R} |\j_\eta|^2 \le C_\eta R^d$ for any $R>1$.  Inserting this into \eqref{contnu}, dividing  by $|K_R|$ and letting $R \to \infty$, we easily get the result.\end{proof}

We next turn to configurations with well-separated points. 
We will need the following scaling relation, which can be obtained from \eqref{eq:def Coul ener} and \eqref{kapd} by a change of variables:
\begin{equation}\label{chvard}
\begin{cases}
& D(\delta_0^{(\ell)}, \delta^{(\ell)}_0) =\frac{\kappa_d}{c_d} \g(\eta) \quad \text{if } d\ge 3\\
& D(\delta_0^{(\ell)}, \delta^{(\ell)}_0)=   \g(\eta)+\frac{\gamma_2}{c_2} = \frac{\kappa_2}{c_2} \g(\eta)+\frac{ \gamma_2 }{c_2}\quad \text{if } d=2.\end{cases}\end{equation}

\begin{lem}[\textbf{The energy of well-separated configurations}]\label{lemequiv}\mbox{}\\
Assume  that $\j=\nab h$ satisfies
 \begin{equation}\label{eqce}- \Delta h = c_d \Big( \sum_{p \in \Lambda } \delta_p - a(x)\Big)\end{equation}
 in some subset $U\subset \mr^d$, for some $a(x) \in L^\infty(U)$, and $\Lambda $ a discrete subset of $U$, and
\begin{equation}\label{wellsep}\min\left(\min_{p\neq p'\in \Lambda} |p-p'|, \min_{p\in \Lambda} \dist (p, \p U)\right)\ge  \eta_0>0.\end{equation}
Then, we have
\begin{equation}\label{lequi}
 \int_{U} |\j_\eta|^2 -  \# ( \Lambda\cap U)(  \kappa_d \g(\eta)+  \gamma_2 \indic_{d=2})= W(\j, \indic_U) + \# (\Lambda \cap U) o_\eta(1)\|a\|_{L^\infty(U)}, \end{equation}
where $o_\eta(1) \to 0$ as $\eta \to 0$ is a function that depends only on the dimension. Moreover, \begin{equation}\label{eq:lowbound sep}
W(\j, \indic_U) \ge - C \#(\Lambda \cap U), 
\end{equation}
where $C>0$ depends only on the  dimension, $\gamma_2$ (hence the choice of smearing function $\ro$), $\|a\|_{L^\infty}$  and $\eta_0$.
\end{lem}

\begin{proof}
We recall that by definition of $\j_\eta$ (cf. Definition \ref{def:smear field}) we have 
$$\j_\eta= \j + \sum_{p\in \Lambda} \nab f_{\eta}(x-p).$$
Since the $B(p, \eta_0) $ are disjoint and included in $U$, and $f_{\eta}$ is identically $0$  outside of $B(0, \eta)$ we may write for any $\eta<\eta_0$, and any $0<\alpha<\eta$,
\begin{multline}\label{pu12}
\int_{U\backslash \cup_{p\in \Lambda} B(p, \alpha)}  |\j_\eta|^2 = \int_{U
 \backslash \cup_{p\in \Lambda} B(p, \alpha)}
 |\j|^2 + \# (\Lambda \cap U) \int_{B(0,\eta)\backslash B(0,\alpha)}  |\nab f_{\eta} |^2 \\ + 2 \sum_{p \in \Lambda} \int_{B(p, \eta)\backslash B(p,\alpha)}  \nab f_\eta(x-p) \cdot \j.
 \end{multline}
First we note that, using Green's formula,  and $\vec{\nu}$ denoting the outwards pointing unit normal to $\p B(0, \alpha)$ we have
\begin{equation*}
\label{}
\int_{B(0,\eta)\backslash B(0,\alpha)} |\nab f_{\eta}|^2 = -\int_{\p B(0, \alpha)} f_\eta \nab f_\eta \cdot \vec{\nu} + c_d \int_{B(0,\eta)\backslash B(0,\alpha)}f_\eta \delta_0^{(\eta)}.\end{equation*}
By Green's formula again and the definition of $f_\eta$ we have  
$$\int_{\p B(0, \alpha)}
\nab f_\eta \cdot \vec{\nu}= - c_d \int_{B(0, \alpha)} \delta_0^{(\eta)} + c_d=  c_d + o_\alpha(1)$$
as $\alpha \to 0$,  and combining with the fact
that $f_\eta= \g*\delta_0^{(\eta)} - \g$  (see its definition \eqref{eqf0}) we find
\begin{equation}\label{partief0}
\int_{B(0,\eta)\backslash B(0,\alpha)} |\nab f_{\eta}|^2 = -
 c_d f_\eta (\alpha)  + c_d \int_{\mr^d} \left( \g*\delta_0^{(\eta)}\right) \delta_0^{(\eta)} - c_d \int_{\mr^d} \g \delta_0^{(\eta)} +o_\alpha(1) .
\end{equation}
We next observe that $c_d \int_{\mr^d} \left( \g*\delta_0^{(\eta)}\right) \delta_0^{(\eta)}=c_d D(\delta_0^{(\eta)}, \delta_0^{(\eta)})$
and
 $\int_{\mr^d} \g \delta_0^{(\eta)}=\g* \delta_0^{(\eta)}(0),$ thus, inserting into \eqref{partief0}, we find
\begin{multline}\label{partief}
\int_{B(0,\eta)\backslash B(0,\alpha)} |\nab f_{\eta}|^2 =
- c_d \g* \delta_0^{(\eta)}  (\alpha) + c_d \g(\alpha) + \kappa_d \g(\eta) - c_d\g* \delta_0^{(\eta)}(0) +o_\alpha(1)\\=  -2c_d \g* \delta_0^{(\eta)} (0)+     c_d \g(\alpha )+ c_d D(\delta_0^{(\eta)}, \delta_0^{(\eta)})   + o_\alpha(1),  \end{multline}
in view of the fact that for fixed $\eta$,  $\g* \delta_0^{(\eta)}$ is continuous at $0$.
On the other hand,  using Green's theorem and \eqref{eqce} we have
\begin{multline*}\int_{B(p, \eta)\backslash B(p,\alpha)}  \nab f_\eta(x-p) \cdot \j= - c_d \int_{B(p,\eta)\backslash B(p,\alpha) }  f_\eta(x-p) a(x)\, dx
- f_\eta(\alpha) \int_{\p B(p, \alpha)} \j \cdot \vec{\nu}.\end{multline*}
First we note that
\begin{equation*}
\left|\int_{B(p,\eta)\backslash B(p,\alpha) }  f_\eta(x-p) a(x)\, dx
\right|
 \le
c_d \|a\|_{L^\infty}\int_{B(0, \eta)} |f_\eta|(x)
dx\le  \|a\|_{L^\infty}o_\eta(1),\end{equation*}
where $o_\eta(1)$ depends only on $\ro $ and $d$. To see this, just notice that  we have  $|f_\eta|\le |\g|$ and the Coulomb kernel $\g$ is integrable near the origin.
Secondly, by Green's theorem again  we have
$$- \int_{\p B(p, \alpha)} \j \cdot \vec{\nu}=  c_d + O(\|a\|_{L^\infty} \alpha^d)  .$$
Inserting these two facts we deduce
\begin{equation*}\int_{B(p, \eta)\backslash B(p,\alpha)}  \nab f_\eta(x-p) \cdot \j= c_d f_\eta(\alpha) +  O(\|a\|_{L^\infty} \alpha^d \g(\alpha) ) +  \|a\|_{L^\infty}o_\eta(1).\end{equation*}
Combining this and \eqref{partief}, \eqref{pu12},  \eqref{chvard}, and again $f_\eta(\alpha)=
 \g* \delta_0^{(\eta)} (0) - \g(\alpha) +o_\alpha(1)$,  we find
\begin{multline*}
\int_{U\backslash \cup_{p\in \Lambda} B(p, \alpha)}  |\j_\eta|^2 = \int_{U
 \backslash \cup_{p\in \Lambda} B(p, \alpha)}
 |\j|^2 + \# ( \Lambda\cap U) (  \kappa_d \g (\eta) +  \gamma_2 \indic_{d=2} -c_d \g (\alpha) +o_\alpha(1)) \\ +\|a\|_{L^\infty} o_\eta(1)\# (\Lambda\cap U)  +   O(\|a\|_{L^\infty} \alpha^d f_\eta(\alpha) ) \# (\Lambda\cap U) .
\end{multline*}
Letting $\alpha \to 0$, \eqref{lequi} follows by the  definition \eqref{Wold}.

The proof of \eqref{eq:lowbound sep} is a toy version of ``ball construction" lower bounds in Ginzburg-Landau theory, made much simpler by the separation of the points. From \eqref{eqce} and the Cauchy-Schwarz inequality, we have,  for any $p\in \Lambda$,
\begin{multline}\label{minfacile}\int_{B(p,\hal \eta_0 ) \backslash B(p, \eta)} |\j|^2 \ge \int_{\eta}^{\eta_0/2} \frac{1}{|\mathbb{S}_{d-1}|   t^{d-1} } \left(\int_{\p B_t }  \j  \cdot \vec{\nu}\right)^2\, dt  \\
\ge  c_d^2 \int_{\eta}^{\eta_0/2} \frac{1}{|\mathbb{S}^{d-1}|  t^{d-1} }  (1- \|a\|_{L^\infty} |\mathbb{B}^d  | t^d)^2 \, dt \ge c_d( \g(\eta)- \g(\eta_0/2))  - C
\end{multline}  where $|\mathbb{B}^d|$ is the volume of the unit ball in dimension $d$, and we have used   the definition of $c_d$, and $C$ depends only on $\|a\|_{L^\infty}$ and $d$. We may then absorb $c_d \g(\eta_0/2)$ into a  constant $C>0$ depending only on $\|a\|_{L^\infty}$, $\eta_0$ and $d$.
Since the $B(p, \hal\eta_0)$ are disjoint and included in $U$,
we may add these lower bounds, and obtain the result.
\end{proof}

\subsection{Splitting lower bounds}

We start by discussing the problem of minimization of $\En$ defined in \eqref{eq:def MF ener}. Direct variations  of the form $(1-t) \mu_0+ t \nu$ for $\nu\in \mathcal P(\mr^d)$ and $t\in [0,1]$  yield that the unique minimizer of $\En$ (note that $\En$ is strictly convex), denoted $\mu_0$, solves
\begin{equation}\label{variaineq}
\begin{cases}
h_{\mu_0} + \frac{V}{2} = c:= \hal (\I[\mu_0] + D(\mu_0,\mu_0) ) &  q.e. \ \text{in} \ \supp\,  \mu_0\\
h_{\mu_0}+ \frac{V}{2}\ge c:= \hal ( \I[\mu_0] + D(\mu_0,\mu_0))  &  q.e.
\end{cases}
\end{equation}
where $q.e. $ means ``outside of a set of capacity $0$", and  $h_{\mu_0} = \g* \mu_0$ is the potential generated by $\mu_0$. More precisely, the variations first yield that \eqref{variaineq} holds for some constant $c$ on the right-hand side of both relations. Then, integrating the first relation against $\mu_0$ yields that 
$$c= \int_{\mr^d} (h_{\mu_0} + \frac{V}{2})\, d\mu_0= D(\mu_0, \mu_0) + \int \frac{V}{2}\, d\mu_0$$
which  identifies $c$ as the right-hand side in \eqref{variaineq}.
It turns out that these relations can also be shown to characterize uniquely $\mu_0$ (by convexity). For more details, as well as a proof of the existence of $\mu_0$ if the assumptions \eqref{eq:trap pot} are verified, one can see \cite{frostman},  \cite[Chap. 1]{ST} in dimension $2$, or \cite{ln}.
It turns out, although this is rarely emphasized in the literature, that the solution $\mu_0$ is also related to an obstacle problem, in the following sense: the potential $h_{\mu_0}$ generated by $\mu_0$ can be shown (for a proof, one can refer to \cite{asz}, the proof is presented in dimension $2$  and for $V$ quadratic but carries over to dimension $d\ge 3$ and general $V$  with no change)  to solve the following variational inequality:
\begin{equation}\label{varineqobs}
\forall u\in \mathcal K, \quad  \int_{\mr^d} \nab h_{\mu_0} \cdot \nab (u - h_{\mu_0}) \ge 0 
\end{equation}
where 
$$\mathcal K= \{ u \in H^1_{loc} (\mr^d), u  - h_{\mu_0}\ \text{ has bounded support and}\ u \ge - \frac{V}{2} +c, \ \mbox{q.e.}\}.$$
This happens to characterize a classical obstacle problem with obstacle $\vp:=- \frac{V}{2}+c$ (for general background on the obstacle problem, one can see \cite{ks}).
We then denote
\begin{equation}\label{eq:zeta}
\z= h_{\mu_0}+ \frac{V}{2}-\hal \left( \I[\mu_0] +  D(\mu_0,\mu_0)\right),
\end{equation} and note that in view of \eqref{variaineq}, $\zeta\ge 0$ and $\zeta=0 $ in $\E$.
Because $h_{\mu_0}$ is a solution to the obstacle problem with obstacle $\vp$, the support of $\mu_0$, that we denote $\E$, is {contained in} the so-called ``coincidence set" where $h_{\mu_0}= \vp$ i.e. the set $\{\zeta=0\}$. {Equality between these sets happens for example  $V\in C^2$ and  $\Delta V$ does not vanish on the coincidence set.
We observe that in all cases the set $\{\zeta=0\}$ is bounded since $h_{\mu_0}$ behaves like $\g$ at infinity and \eqref{eq:trap pot} holds.}

\medskip

We next turn to recalling our blow-up procedure: $x'=n^{1/d} x$, $x_i'= n^{1/d} x_i$.
For a configuration of points $(x_1, \dots , x_n)$ we  let here and in the sequel, as in \eqref{hpn},
\begin{equation}\label{hh}
h_n= w* \Big (\sum_{i=1}^n \delta_{x_i} -n \mu_0\Big) \qquad h_n'(x')= n^{2/d-1} h_n(x)= w*  \Big (\sum_{i=1}^n \delta_{x_i'} - \mu_0(n^{-1/d} x') \Big),
\end{equation}
and for  $\ell$ and $\eta$ related by $\ell= n^{-1/d} \eta$, we let
\begin{equation}\label{defh}
\hel:= \g * \Big(\sum_{i=1}^n \delta^{(\ell)}_{x_i} -n\mu_0\Big) \qquad  \het (x') = n^{2/d-1} \hel (x)= \g *   \Big (\sum_{i=1}^n \delta_{x_i'}^{(\eta)} - \mu_0(n^{-1/d} x') \Big).
\end{equation}
We also denote \begin{equation}\nu_n'= \sum_{i=1}^n \delta_{x_i'}.\end{equation}

Our first splitting result is the equivalent of \cite[Lemma 2.1]{ss2d}.

\begin{lem}[\textbf{Exact splitting}]\mbox{}\\
\label{split1}
For any $n\ge 1$ and any $(x_1,\dots, x_n) \in (\mr^d)^n$, letting $h_n'$ be as in \eqref{hh},  we have
\begin{align}
\w(x_1,\dots, x_n) &= n^2 \I[\mu_0]+ 2n \sum_{i=1}^n \zeta (x_i) - \hal n \log n + \frac{1}{c_d} W(\nab h_n', \indic_{\mr^2}) \  &\text{if} \ d=2 \nonumber\\
  &= n^2 \I[\mu_0]+ 2n \sum_{i=1}^n \zeta (x_i)+ \frac{n^{1-2/d}}{c_d} W(\nab h_n', \indic_{\mr^d})  \  &\text{if} \ d\ge 3.
  \end{align}
\end{lem}
\begin{proof} Exactly as in  
 \cite[Lemma 2.1]{ss2d}, we can show that  
 \begin{equation}\label{presplit}\w(x_1, \dots, x_n)= n^2 \I[\mu_0]+ 2n \sum_{i=1}^n \zeta (x_i) +\frac{1}{c_d} W(\nab h_n, \indic_{\mr^d}) .\end{equation}
 A change of variables yields that 
 $$\int_{\mr^d\backslash \cup_{i=1}^n B(x_i, \eta)} |\nab h_n|^2 
 = n^{1-2/d}\int_{\mr^d\backslash \cup_{i=1}^n B(x_i', \eta n^{1/d} )}|\nab h_n'|^2$$
 and subtracting off $n c_d \g(\eta)$ from both sides, writing $\g(\eta)= \g(\eta n^{1/2}) + \frac{1}{2}\log n $ in dimension $2$ and $\g(\eta)= \g(\eta n^{1/2}) n^{d/2-1}$ in dimension $d\ge 3$, 
  and letting $\eta \to 0$,  we are led to 
 \begin{equation}
 \begin{cases}
 & W(\nab h_n, \indic_{\mr^d}) = W(\nab h_n', \indic_{\mr^d}) - \frac{c_d n}{2} \log n \quad \text{if }  d=2\\
  & W(\nab h_n, \indic_{\mr^d}) = n^{1-2/d} W(\nab h_n', \indic_{\mr^d}) \quad   \text{if } \ d=3 .\end{cases}\end{equation} Inserting this into \eqref{presplit} yields the result.
\end{proof}

In \cite{ss2d} $W$ was then bounded below locally. As already mentioned, we do not know how to do this in dimension $d\ge 3$ (due to the lack of an efficient ``ball construction" method) so we resort to
a second type of splitting, using the smearing out of charges.

%For any $x\in \mr^d$, and any $\ell>0$,  we let $\delta_x^{(\ell)}:=\dashint_{B(x, \ell)} \, dx,$ be the Dirac mass smoothed out radially at the scale $\ell$.

A crucial ingredient in our approach is that the electrostatic energy of a configuration of positive smeared  charges is always a lower bound for the energy of the corresponding configuration of point charges, with equality if the smeared charges do not overlap. This is the so-called Onsager lemma, which one can find for example in \cite[Chapter 6]{LieSei}. We reproduce it here for sake of completeness.

\begin{lem}[\textbf{Onsager's lemma}]\label{lem:Onsager}\mbox{}\\
For any nonnegative  distribution $\mu$  such that $\int_{\mr^d}  \mu= n$, any $n$, any  $x_1,\ldots,x_n\in \mr^{d}$, and any $\ell>0$, we have 
\begin{equation}\label{eq:Onsager}
\sum_{i\neq j} \g(x_i-x_j) \geq D\left( \mu - \sum_i \delta_{x_i}^{(\ell)},\mu - \sum_{i=1}^n \delta_{x_i}^{(\ell)} \right) - D(\mu,\mu) + 2 \sum_{i=1} ^n D(\mu,\delta_{x_i}^{(\ell)})- n  D(\delta_0^{(\ell)}, \delta_0^{(\ell)}) 
\end{equation}
with equality if $B(x_i,\ell) \cap B(x_j, \ell) = \varnothing$, for all $i\neq j$.
\end{lem}

\begin{proof}
It is based on Newton's theorem, see \cite[Theorem 9.7]{LiLo}, which can easily be adapted to any dimension, and asserts  in particular that  for any $\ell$, $\g* \delta_0^{(\ell)} \le \g= \g * \delta_0$ pointwise.
It follows from this that
\begin{equation}\label{minonewton}
 \sum_{i\neq j} \g(x_i-x_j) \geq \sum_{i\neq j} D(\delta_{x_i}^{(\ell)},\delta_{x_j}^{(\ell)})
\end{equation} with equality if $\min_{i\neq j} |x_i-x_j|\ge 2 \ell$.
Indeed, by Newton's theorem we have
$$\int \left(\g * \delta_{x_i}^{(\ell)}\right) \delta_{x_j}^{(\ell)} \le \int \left(\g * \delta_{x_i}\right) \delta_{x_j}^{(\ell)} = \int \left(\g * \delta_{x_j}^{(\ell)}\right) \delta_{x_i} \le \int \left(\g * \delta_{x_j}\right) \delta_{x_i}.$$
We then write
$$D\left(\mu - \sum_{i=1}^n \delta_{x_i}^{(\ell)}, \mu - \sum_{i=1}^n \delta_{x_j}^{(\ell)}\right)= D(\mu,\mu) - 2\sum_{i=1}^n D(\mu, \delta_{x_i}^{(\ell)}) + \sum_{i\neq j} D(\delta_{x_i}^{(\ell)}, \delta_{x_j}^{(\ell)}) + \sum_{i=1}^n D(\delta_{x_i}^{(\ell)}, \delta_{x_i}^{(\ell)}) $$
and from this relation and \eqref{minonewton},   the lemma easily follows.
 \end{proof}

Smearing-out charges comes with a cost, that we quantify in the following lemma.

\begin{lem}[\textbf{The cost of smearing charges out}]\mbox{}\label{lem:Lieb}\\
For any $\mu \in L ^{\infty}(\R ^d)$ and any point $x$, we have
\begin{equation}\label{eq:Lieb}
\left|D\left( \mu, \delta_{x} - \delell_{x} \right) \right| \leq C \ell ^2 \left\Vert \mu \right\Vert_{L ^{\infty}},
\end{equation}where $C$ depends only on the choice of $\ro $ in \eqref{eq:smeared charge}  and the dimension.
\end{lem}

\begin{proof}
 Without loss of generality, we may assume $x=0$, and write
$$D(\mu, \delta_{0}- \delta_{0}^{(\ell)} )= \int (\g - \g* \delta_0^{(\ell)} ) (x) \mu (x) \,dx.$$
By Newton's theorem, the function $\g- \g* \delta_0^{(\ell)}$ is nonnegative and  supported in $B(0,\ell)$. In dimension $3$, we may just write that  in particular it is smaller than $\g \indic_{B(0, \ell)}$, and so 
 we may write
$$\left| D(\mu, \delta_{0}- \delta_{0}^{(\ell)} )\right| \le \|\mu\|_{L^\infty} \int_{B(0,\ell)} \frac{dx}{|x|^{d-2}} \le C \ell^2 \|\mu\|_{L^\infty} .$$
In dimension 2, we write 
\begin{multline*}
\int_{\mr^2} |\g - \g* \delta_0^{(\ell)} | = \int_{B(0,\ell)}\left|-\log |x| + \frac{1}{\ell^2} \int \log |x-y|\ro(\frac{y}{\ell}) \, dy \right| dx \\=\ell^2  \int_{B(0,1)} \left|- \log |x'|+ \int \log |x'-y'| \ro(y') \, dy'\right|\, dx= C \ell^2\end{multline*}
where we have used the changes of variables $x=\ell x'$, $y=\ell y'$, and we conclude in the same way.
\end{proof}
From these two lemmas we deduce
\begin{lem}[\textbf{Splitting lower bound}]\label{lemsplit}\mbox{}\\
For  any $n$, any $x_1,\ldots,x_n \in \mr^d $, letting $\hel$ be as in \eqref{defh}, we have
\begin{equation}\label{eq:split Onsager}
\w(x_1,\dots,x_n) \ge n^2 \I[\mu_0] + \frac{1}{c_d}\left(\int_{\mr^d} |\nab \hel|^2 -n (\kappa_d\g(\ell) + \gamma_2 \indic_{d=2})\right) +2n\sum_{i=1}^n \z(x_i)-Cn^{2-2/d}\eta^2,
\end{equation}
where  $C$ depends  only on  the dimension. Moreover, equality holds if $\min_{i\neq j} |x_i-x_j| \geq  2\ell$.
\end{lem}
Note that compared to Lemma \ref{split1} it is an inequality, and not an equality, and it has an error term, however it achieves the same goal. Indeed, as we will see, points in minimizing configurations are well-separated so that there is equality in \eqref{eq:split Onsager} and no information is lost in the end.

\begin{proof}
We proceed as in \cite[Proof of Thm 3.2]{RSY2}. First, applying Onsager's lemma above with $\ell$  and $\mu= n\mu_0$,  and using \eqref{chvard}, we find
\begin{multline}\label{eq:appl Onsager}
\sum_{i\neq j} \g( x_i-x_j)   \geq D\left( n\mu_0  - \sum_{i=1}^n\delta_{x_i}^{(\ell)} ,n\mu_0  -   \sum_{i=1}^n\delta_{x_i}^{(\ell)}   \right)\\
 - n^2 D(\mu_0 ,\mu_0) + 2 n \sum_{i=1}^n D(\mu_0, \delta_{x_i}^{(\ell)}    )- n 
 D( \delta_{0}^{(\ell)} ,\delta_{0}^{(\ell)} ),
\end{multline}
with equality if the $B(x_i, \ell)$ are disjoint.
We then use Lemma \ref{lem:Lieb}, the fact that $\mu_0$ is a fixed $L ^{\infty}$ function and
\eqref{eq:zeta} to obtain
\begin{align}\label{eq:split pot terms}
2 n \sum_{i=1}^n D(\mu_0, \delta_{x_i}^{(\ell)}) &= 2 n \sum_{i=1}^n D(\mu_0, \delta_{x_i}) + O\left(n ^2 \ell ^2 \right) =  2 n \sum_{i=1}^n h_{\mu_{0}} (x_i) + O\left(n ^2 \ell ^2 \right)\nonumber \\
&= n ^2 \left( \I[\mu_0] + D(\mu_0,\mu_0) \right) + 2 n \sum_{i=1}^n \z (x_i) - n \sum_{i=1}^n V (x_i) + O\left(n ^2 \ell ^2 \right).
\end{align}
Combining \eqref{eq:appl Onsager} and \eqref{eq:split pot terms}, and observing that  in view of \eqref{eq:ener field} and \eqref{defh}, we have  
$$D \left( n\mu_0  - \sum_{i=1}^n\delta_{x_i}^{(\ell)} ,n\mu_0  -   \sum_{i=1}^n\delta_{x_i}^{(\ell)}   \right)= \frac{1}{c_d}  \int_{\mr^d} |\nabla \hel|^2$$
we thus obtain
\begin{multline*}
H_n (x_1,\ldots,x_n) = \sum_{i\neq j} \g( x_i-x_j)+ n \sum_{i=1} ^n V (x_i)  \geq \\
 n ^2 \I[\mu_0] + \frac{1}{c_d}\int_{\mr^d} |\nab \hel|^2 + 2n\sum_{i=1}^n \z(x_i)- n  D( \delta_{0}^{(\ell)} ,\delta_{0}^{(\ell)} ) + O\left(n ^2 \ell ^2 \right)
.\end{multline*}
Since $\ell= n^{-1/d} \eta$ and \eqref{chvard} holds,  we obtain the result.

 \end{proof}

Combining this lower bound with our blow-up and noting that 
by a change of variables we  have
\begin{equation}\label{chvarh}
\int_{\mr^d} |\nab \hel|^2 =n^{1-2/d} \int_{\mr^d} |\nab \het|^2.\end{equation}
we obtain the following rephrasing:
\begin{equation}\label{pred2}
\w(x_1,\dots,x_n)-n^2 \I[\mu_0]  +\D \left( \frac{n}{2}\log n\right)\indic_{d=2}  \ge n^{2-2/d}\left(J_n(x_1, \dots, x_n) - C\eta^2\right)  + 2n \sum_{i=1}^n \zeta(x_i) 
\end{equation}
where we have written
\begin{equation}\label{F}
J_n (x_1, \dots, x_n) :=  \frac{1}{c_d}\left(\frac{1}{n}\int_{\mr^d} |\nab \het|^2 - (\kappa_d \g(\eta)+ \gamma_2 \indic_{d=2})\right).\nonumber
\end{equation}

\begin{remark} 
Taking for example $\eta=1$, and using that $\zeta\ge 0$, it immediately follows from the above that 
\begin{equation}\label{pred0}
 \w(x_1,\dots,x_n)-n^2 \I[\mu_0]  +\D \left(\frac{n}{2}\log n\right) \indic_{d=2} \ge - C n^{2-2/d} 
\end{equation}
where the constant depends only on the dimension. This provides a very simple proof of this fact.
\end{remark}

\subsection{Controlling fluctuations }

The key ingredient in the proof of Theorems \ref{thm:charge fluctu} is the fact that the energy density (square of the local $L^2$ norm) of the electric field generated by the smeared charges provides a control of the deviations we are interested in. We start by formalizing this idea in two lemmas which, thanks to the smearing out method, provide significantly simpler alternatives to the estimates of \cite{ss2d,SeTi}. Note that the following considerations do not require that $\mu_0$ be the equilibrium measure.

We start with a control on the fluctuations of the number of points (i.e. total charge) in a given ball:
\begin{equation}\label{eq:defi fluctu charge}
D(x',R) = \nu'_n (B(x',R)) - \int_{B(x',R)} \mu_0'.  
\end{equation}
where 
$$ \mu_0' = \mu_0 (n ^{- 1/d} \cdot).$$
Note that the quantity \eqref{eq:intro charge fluctu} used in Theorem \ref{thm:charge fluctu} is 
\begin{equation}\label{eq:fluctu charge micro macro}
D(x,R) = D(x',R n ^{1/d}) 
\end{equation}
where $x'$ is the blown-up of $x$. For this reason we do not feel the  need for a new notation and we will use the above for fluctuations either at the microscopic or macroscopic scale.  

\begin{lem}[\textbf{Controlling charge fluctuations}]\label{lem:fluctu charge}\mbox{}\\
For any $x_1, \dots, x_n$ and $\het$ given by \eqref{hh}, 
for any $0 < \eta < 1$, $R>2$ and $x' \in \R ^d$, we have
\begin{equation}\label{eq:control fluctu charge}
\int_{\R ^d} |\nabla \het | ^2 \geq C \frac{D(x',R) ^2 }{R ^{d-2}} \min\left( 1, \frac{D(x',R)}{R ^d} \right),
\end{equation} 
where $C$ is a constant depending only on $d$.
\end{lem}

\begin{proof} 
In the proof, $C$ will denote a constant depending only on $d$ that may change from line to line.
We distinguish two cases, according to whether $D:=D(x',R)>0$ or $D \leq 0$, and start with the former. We first claim that for all $t$ such that 
\begin{equation}\label{eq:defi T}
R+\eta \leq t\leq T := \left( (R+\eta) ^d + \frac{D}{2C} \right) ^{1/d}
\end{equation}
for some well-chosen constant $C$,  it holds that 
\begin{equation}\label{eq:up bound D}
-\int_{\partial B(x',t)} \nabla \het \cdot \vec{\nu} \geq \frac{c_d}{2} D 
\end{equation}
where $\vec{\nu}$ is the outwards pointing normal to $\partial B(x',t)$. Indeed, by Green's  formula 
\begin{align*}
-\int_{\partial B(x',t)} \nab\het\cdot \vec{\nu}  &= \int_{B(x',t)} - \Delta \het = c_d \int_{B(x',t)} \left( \sum_{i=1}^n \deleta_{x'_i}- \mu_0' \right) \\
&\geq c_d D(x',R) - c_d\int_{B(x',t)\setminus B(x',R)} \mu_0 ' \geq c_d D(x',R) - C\left( t ^d - R ^d\right) \geq \frac{c_d}{2} D
\end{align*}
if $t$ satisfies \eqref{eq:defi T}. We have used the positivity of the smeared charges and the fact that $\mu_0'$ is bounded in $L ^{\infty}$ \eqref{ass2}. Then, integrating in spherical coordinates and using the Cauchy-Schwarz inequality  as in \eqref{minfacile}, we find
\begin{align*}
\int_{\R ^d} |\nabla \het | ^2 &\geq \int_{B(x',T) \setminus B(x',R+ \eta)} |\nabla \het | ^2 
 \geq \int_{t= R+\eta} ^T \frac{1}{t ^{d-1}\left|\mathbb{S}^{d-1}\right|} \left(\int_{\partial B(x',t)} \nab \het \cdot \vec{\nu}\right) ^2 dt \\
\\ &\geq C  D^2 \int_{t= R+\eta} ^T \frac{1}{t ^{d-1}\left|\mathbb{S}^{d-1}\right|} = C D ^2 
\begin{cases}
\log\frac{T}{R+\eta} \quad \mbox{ if } d=2 \\
\frac{1}{(R+\eta)^{d-2}} - \frac{1}{T ^{d-2}}\quad  \mbox{ if } d \geq 3.
\end{cases}
\end{align*}
There only remains to note that in view of \eqref{eq:defi T} and the fact that $\eta<R/2$, we have
\begin{align*}
\log  \frac{T}{R+\eta} &= \log \(1+\frac{D}{2C(R +\eta) ^d}\) ^{1/d} \geq C' \min(1,\frac{D}{R^d})\\
\frac{1}{(R+\eta)^{d-2}} - \frac{1}{T ^{d-2}} &= \frac{1}{(R+\eta)^{d-2}} \left( 1 - 
\left( 1+  \frac{D}{2C(R+\eta)^d}\right)^{\frac{d}{d-2}} \right)
  \geq \frac{C'}{R ^{d-2}} \min(1,\frac{D}{R^d})
\end{align*}
to conclude the proof in the case $D>0$.

If $D\leq 0$ the computation is the same, except that we set 
\[
T = \left((R-\eta) ^d - \frac{D}{2C}\right) ^{1/d}, 
\]
and use that for any $ T \leq t\leq R-\eta$ 
\begin{align*}
-\int_{\partial B(x',t)} \nab \het \cdot \vec{\nu} &= \int_{B(x',t)} - \Delta \het  \leq c_d D -C(R ^d - t ^d) \leq \frac{c_d}{2} D 
\end{align*}
so that 
$$ \left( \int_{\partial B(x',t)} \nab \het\cdot \vec{\nu}\right) ^2 \geq C D ^2$$
again. We then integrate the energy density on $B(x',R-\eta)\setminus B(x',T)$ and argue as before. 
\end{proof}

The method of smearing out charges also provides a very convenient way to control the  $L^q$ norms of the electric field.

\begin{lem}[\textbf{Controlling electric field fluctuations}]\label{lem:fluctu field}\mbox{}\\
For any $x_1, \dots, x_n$ and $\het$, $h_n'$ given by \eqref{hh}, 
for any $1\le q <\frac{d}{d-1} $, any $1 > \eta > 0$ and any $R>0$, we have
\begin{align}\label{eq:control Lq field}
 \left\Vert \nabla h'_n \right\Vert_{L ^q (K_R)} &\leq |K_R | ^{1/q - 1/2} \left\Vert \nabla \het \right\Vert_{L ^2 (K_R)} + C_{q,\eta} \nu_n ' (K_{R+\eta})\\
\label{controlLq2} &\leq  C (R^{\frac{d}{q} - \frac{d}{2}}+ R^{\frac{d-1}{2}})  \left\Vert \nabla \het \right\Vert_{L ^2 (K_{R+2})} + C_{q,\eta} \|\mu_0\|_{L^\infty} R^d ,
\end{align}
where $C_{q,\eta}$ depends only on $q$ and $\eta$ and satisfies $C_{q,\eta}\to 0$ when $\eta \to 0$ at fixed $q$, and $C$ depends only on $d$. \end{lem}

\begin{proof}
Recalling the definition \eqref{eqf0} we have 
\[
\nabla h'_n = \nabla \het - \sum_{i=1} ^n \nabla f_{\eta} (x-x_i) 
\]
and thus 
\[
\left\Vert \nabla h'_n \right\Vert_{L ^q (K_R)} \leq \left\Vert \nabla \het \right\Vert_{L ^q (K_R)} + \nu_n' (K_{R+\eta}) \left\Vert \nabla f_\eta\right\Vert_{L ^q (\R ^d)}
\]
where we used that if $x\in K_R$ and $\eta<1$, then $f_{\eta} (x-x_i) = 0$ if $x_i \in (K_{R+\eta})^c$. A simple application of H\"older's inequality then yields 
\[
\left\Vert \nabla \het \right\Vert_{L ^q (K_R)} \leq |K_R| ^{1/q-1/2} \left\Vert \nabla \het \right\Vert_{L ^2 (K_R)} 
\]
and concludes the proof of the first inequality, with 
$C_{q,\eta} :=  \left\Vert \nabla f_\eta\right\Vert_{L ^q (\R ^d)}.$
Bounding then  $\nu_n'(K_{R+\eta})$ exactly as we did for \eqref{contnu}, we deduce 
$$ \left\Vert \nabla h'_n \right\Vert_{L ^q (K_R)} \leq |K_R | ^{1/q - 1/2} \left\Vert \nabla \het \right\Vert_{L ^2 (K_R)} + C_{q,\eta} \|\mu_0\|_{L^\infty} |K_R| + C   R^{\frac{d-1}{2}}\|\nab \het\|_{L^2(K_{R+2})},$$
and \eqref{controlLq2} follows.
\end{proof}

\section{Lower bound to the ground state energy}\label{sec5}

In this section, we will start to take the limits $n\to \infty$ and  $\eta \to 0$, and we provide the lower bound part of Theorem \ref{th1} by proving the following:

\begin{pro}[\textbf{Lower bound to the ground state energy}]\mbox{} \label{promino}\\
Let $(x_1,\ldots, x_n) \in (\mr^d)^n$, let $h_n'$ be associated via \eqref{hh}, and let $P_{\nu_n}=i_n (x_1, \ldots, x_n) $ be defined as in \eqref{eq:Pnun}. Then $\{P_{\nu_n}\}_n$ is compact for the weak topology of  probability measures on $X$ and any $P$ which is the limit of  an extracted subsequence of $P_{\nu_n}$ is admissible and satisfies 
\begin{equation}\label{222}
\liminf_{n\to \infty} n^{2/d-2} \left(\w(x_1, \dots, x_n) -n^2\I[\mu_0]+ \left(\frac{n}{2} \log n \right)\indic_{d=2} \right) \ge \widetilde{ \W}(P).
\end{equation}
\end{pro}
Lower bounds corresponding to \eqref{eq:intro 222} and \eqref{eq:intro 223} follow from the above {and \eqref{eq:gamma d}}.
% Indeed, using the scaling relation \eqref{eq:scale renorm}, along with Items 1 and 2 we obtain 
%\begin{multline}\label{223}
%% \int \W(\j) \, dP(x, \j) \geq \int \Big(\min_{\overline{\ba}_{\mu_0(x)}} \W \Big) \,dP(x, \j) \\
% =  \Big(\min_{\bai} \W \Big)\int \mu_0(x) ^{2-2/d} dP(x, \j) = \Big(\min_{\bai} \W \Big) \dashint_{\E} \mu_0(x) ^{2-2/d} dx
%\end{multline}
%if $d\geq 3$ (the case $d=2$ is treated similarly).

The proof of Proposition \ref{promino} will occupy all this section and a large part of the next. More precisely, this section is devoted to the connection between the original problem and the renormalized energy at fixed $\eta$, using the general strategy for providing lower bounds for  two-scale energies developed in \cite{gl13,ss2d}. For a lower bound, it is then sufficient to use a lower semi-continuity argment (Fatou's lemma) to pass to the  limit $\eta \to 0$. In this last step we will need the fact that  $\W_\eta$ is bounded below  independently of $\eta$. The proof of this  is postponed to Section \ref{sec:univ low bound}.

\subsection{The  local energy}

%By Newton's theorem, we have the following property, which will be useful later
%\begin{equation}\label{coincid}
%h'_n=\het \quad \text{outside}\cup_{i=1}^n B(x_i',\eta).\end{equation}
The next step is to study the $\liminf$ when $n\to 0$ of $\frac{1}{n}\int_{\mr^d} |\nab \het|^2$ and make the two scale structure apparent. This is a general fact, simple application of Fubini's theorem. Let us pick a smooth cut-off function $\chi$ with integral $1$, support in $B(0,1)$ and which equals $1$ in $B(0,1/2)$. We use it like a smooth partition of unity, writing
\begin{align}\label{eq:two scale}
\int_{\mr^d} |\nab \het|^2 &= \int_{\R  ^d} \left( \int_{\R  ^d} \chi(x-y) dx\right) \left|\nab \het\right|^2 dy\nonumber \\
&= \int_{\R  ^d} \int_{\R  ^d} \chi(x) \left|\nab \het (x+y)\right|^2 dy dx \nonumber \\
&\geq \int_{n ^{1/d} \E} \int_{\R  ^d} \chi(x) \left|\nab \het (x+y)\right|^2 dy dx \nonumber \\
&= n |\E| \dashint_{\E} \int_{\R ^d} \chi(x) \left|\nab \het \left(  y n ^{1/d}+x\right)\right|^2 dx dy
\end{align}
where we simply dropped a part of the integral that we guess will be irrelevant and changed variables. The last line should be interpreted as the average over blow-up centers $y$ of the local electrostatic energy around the center. Indeed, because of the cut-off $\chi$ the integral in $x$ is limited to a bounded region, which means we are integrating the square of the original (non blown up) electric field on a region of size $n ^{-1/d}$ around $y$. The last line is in a form that we can treat using the abstract framework of \cite[Theorem 7]{ss2d}. We will interpret the integral in $x$ as a local energy functional, and deduce from a $\Gamma$-convergence and compactness result (with limit depending on $y$) on this functional a lower bound to \eqref{eq:two scale}. Note that this is also reminiscent of ideas of Graf-Schenker \cite{gs} when they average over simplices.

The following result shows the coercivity of the ``local" energy $\int \chi \left|\nab \het \left(  y n ^{1/d}+\cdot \right)\right|^2$, which is needed to apply the framework of \cite{gl13,ss2d}. In particular we need to prove compactness of the electric fields and that their limits are in the admissible classes of Definition~\ref{def:adm field}.

\begin{lem}[\textbf{Weak compactness of local electric fields}]\label{lemprel}\mbox{}\\
We pick a sequence of blow-up centers $y_n \to y \in \mr^d$. Assume that for every $R>1$ and for some  $\eta\in(0,1)$, we have
\begin{equation}\label{eq:asum h}
\sup_n  \int_{B_R} |\nab \het(n^{1/d}y_n + \cdot) |^2 \le C_{\eta, R}.
\end{equation}
Then $\{\nu'_n(n^{1/d}y_n+\cdot)\}_{n}$ is locally bounded and up to extraction converges weakly as $n\to \infty$, in the sense of measures, to
$$\nu = \sum_{p\in\Lambda}N_p \delta_p$$
where $\Lambda $ is a discrete set and $N_p\in \mathbb{N}^*$.

In addition, there exists $\j\in \Lp$, $p<\frac{d}{d-1}$, $\j_\eta\in \Ld$, with $\j_\eta=\Phi_\eta(\j)$,  such that, up to extraction of a subsequence,
\begin{equation}\label{cvh2}
\nab h_n'(n^{1/d} y_n + \cdot)  \rightharpoonup   \j \ \text{weakly in } \ \Lp \ \text{for}\ p<\frac{d}{d-1}, \ \text{as } \ n\to \infty ,
%& \nab \het (n^{1/d}x_n+ \cdot)  \rightharpoonup  \j \ \text{weakly in } \ \Lp  \ \text{for}\ p<2, \ \text{as } \ \eta\to 0  \mbox{ and } n\to \infty.
\end{equation}
and
\begin{equation}\label{cvh}
\nab \het (  n^{1/d} y_n + \cdot)   \rightharpoonup  \j_\eta \ \text{ weakly in } \ \Ld \ \text{as } \  n \to \infty.
\end{equation}
Moreover $\j = \nabla h$, and if $y \notin \p\E$, we have 
\begin{equation}\label{eqhl}
- \Delta h (x)= c_d(\nu(x) -\mu_0(y))
\end{equation} 
hence $\j \in \overline{\ba}_{\mu_0 (y)}$.
\end{lem}

\begin{proof}
Given $R>1$, by a mean-value argument, there exists $t\in (R-1,R)$ such that for all $n$,
$$\int_{\p B_t} |\nab \het(n^{1/d} y_n+\cdot) |^2 \le C_{\eta, R},$$
from which it follows that (with a different constant)
$$\left| \int_{\p B_t}  \nab \het(n^{1/d} y_n + \cdot) \cdot \vec{ \nu}\right|\le  C_{\eta, R}.$$
But, by Green's formula and \eqref{hh}  the left-hand side is equal to 
$$\left| \int_{ B_t}    \sum\delta_{x_i'}^{(\eta)} (n^{1/d} y_n + x')    - \int_{B_t} \mu_0( y_n  +  n^{-1/d}  x')\, dx' \right|.$$
By \eqref{ass2} it follows that, letting $\underline{\nu_n'}:=\nu_n'(n^{1/d} y_n + \cdot) $, we have
$$\underline{\nu_n'}(B_{R-1}) \le C_d \|\mu_0\|_{L^\infty}  R^d  +C_{\eta,R}.$$
This establishes that $\{\underline{\nu_n'}\}$ is locally bounded independently of $n$. In view of the form of $\underline{\nu_n'}$, its limit can only be of the form $\nu= \sum_{p\in \Lambda} N_p \delta_p$, where $N_p$ are positive  integers, and $\Lambda $ is a discrete set.

Up to a further extraction we also have \eqref{cvh} by \eqref{eq:asum h} and weak compactness in $L ^2_{loc}$. The compactness and  convergence \eqref{cvh2}  follows from Lemma \ref{lem:fluctu field}. The weak local  convergences of both $\underline{\nu_n'}$ and $\nabla h_n ' (n^{1/d} y_n + \cdot)  $ together with the continuity of $\mu_0$ away from $\p \E$ (cf. \eqref{ass2}),   imply after passing to the limit  in \eqref{hh} that $\j$ must be a gradient and that \eqref{eqhl} holds. Finally $\j_\eta=\Phi_\eta(\j)$ because $\Phi_\eta$ commutes with the weak convergence in $\Lp$ for the $\nab h_n'  (n^{1/d} y_n + \cdot)$ described above. Indeed by definition
\[
\Phi_\eta (\nab h_n'  (n^{1/d} y_n  + \cdot)) =  \nab h_{n,\eta}'  (n^{1/d} y_n + \cdot)  = \nab h_n'  (n^{1/d} y_n + \cdot) + \sum_{p \in \Lambda_n} \nabla f_\eta (. - p)
\]
where $\Lambda_n$ is the set of points associated with $\nu_n'  (n^{1/d} y_n + \cdot)$. Since by assumption all these points have limits, one may check that the sum in the right-hand side converges to 
$\sum_{p \in \Lambda} \nabla f_\eta (. - p) 
$, at least weakly in $L ^p _{loc}$. Using in addition the convergences \eqref{cvh} and \eqref{cvh2} we deduce
\[
 \j_\eta = \j + \sum_{p \in \Lambda} \nabla f_\eta (. - p), 
\]
i.e. $\j_\eta = \Phi_\eta (\j)$ as desired.
\end{proof}

\subsection{Large $n$ limit: proof of Proposition \ref{promino}}\label{sec:promino}

We now have the tools to pass to the double-scale limit at fixed $\eta$. The subsequent limit $\eta \to 0$ will use the following result, whose proof is postponed to Section \ref{sec:univ low bound}.

\begin{pro}\label{Wbb}
$\W_\eta$ is bounded below on $\bai$ by a constant depending only on the dimension.
\end{pro}

\medskip

\noindent\emph{Proof of Proposition \ref{promino}}. As announced, we first fix $\eta >0$ and let $n\to +\infty$. We start from \eqref{eq:two scale} and apply the framework of Theorem 7 in \cite{ss2d}. In the notation of that paper we let $G=\E$ and $X=\E\times \Ld$, and take $\ep = n^{-1/d}$.
For $\lambda\in \mr^d$ we let $\theta_\lambda$ denote both the translation $ x \mapsto x+ \lambda $ and the action
$$\theta_\lambda \j= \j\circ \theta_\lambda.$$
Accordingly the action $T^n_\lambda$ is defined for $\lambda\in \mr^d$ by
$$T_\lambda^n(x, \j)= \left(x+ \lambda n^{-1/d}, \j \circ \theta_\lambda\right).$$

We then define a functional over $X$ starting from the local electrostatic energy (the $y$-integral in \eqref{eq:two scale})
\begin{equation}\label{fn}
\mathbf{f}_n(x, \mathcal{Y}):=
\begin{cases} \D\int_{\mr^d}\chi(y)|\mathcal{Y}|^2(y)\, dy & \text{if $\mathcal{Y} =\nab \het(n^{1/d} x+\cdot) $,}\\ +\infty & \text{otherwise.}\end{cases}
\end{equation}
Then \eqref{eq:two scale} is naturally associated with an average of $\mathbf{f}_n$ over blow-up centers in $\E$:
\begin{equation}\label{eq:Fn}
\mathbf{F}_n(\mathcal{Y}):=\dashint_{\E} \mathbf{f}_n\left(x, \theta_{xn^{1/d}} (\mathcal{Y})\right) \, dx. 
\end{equation}
Indeed if $\mathbf{F}_n(\mathcal{Y}) \neq +\infty$,  we have
\begin{equation}\label{Ff}
\mathbf{F}_n(\mathcal{Y})\le    \frac{1}{|\E| n} \int_{\mr^d} |\nab \het|^2.  
\end{equation}

Theorem 7 in \cite{ss2d} was precisely designed to deduce results at the macroscopic  scale (on $\mathbf{F}_n$) from input at the lower scale (on $\mathbf{f}_n$). We now check that its assumptions are satisfied, i.e. that coercivity and $\Gamma$-liminf properties follow from
\begin{equation}\label{coercivite}
\forall R, \quad \limsup_{n\to \infty} \int_{B_R} \mathbf{f}_n (T_\lambda^n (\underline{x_n}, \mathcal{Y}_n)  )\, d\lambda<+\infty, \quad \underline{x_n} \in \E. \end{equation}
But using the definitions above this condition is equivalent to
\begin{equation}\label{coercivite2}
\forall R \mbox{ and } \forall n \geq n_0, \, \mathcal{Y}_n=\nab \het(n^{1/d}\underline{x_n}+ \cdot)\, \text{ and } \limsup_{n\to \infty} \int \chi * \indic_{B_R}|\mathcal{Y}_n|^2 <+\infty
\end{equation}
where $n_0$ is a large enough number. This implies that the assumption \eqref{eq:asum h} of Lemma \ref{lemprel} is satisfied. Up to extraction of a subsequence we may also assume, since $\underline{x_n} \in \E$ which is compact,  that $\underline{x_n} \to x_*$, and we may apply Lemma \ref{lemprel}. Thus we have $\mathcal{Y}_n \wto \mathcal{Y}_*$ weakly in $\Ld$, with all the other results of that lemma, which proves that the compactness assumption of \cite[Theorem 7]{ss2d} is satisfied. Moreover, this weak convergence implies that
\begin{equation}\label{eq:Gamma lim inf}
\liminf_{n\to \infty} \mathbf{f}_n (\underline{x_n},\mathcal{Y}_n)\ge \int_{\mr^d}\chi(y) |\mathcal{Y}_*|^2=\mathbf{f} (x_*, \mathcal{Y}_*) 
\end{equation}
by lower semi-continuity, where
\begin{equation}\label{eq:Gamma lim inf 2}
\mathbf{f} (x_*, \mathcal{Y}_*) := \begin{cases}
\D\int_{\mr^d}\chi(y)|\mathcal{Y}_*|^2(y)\, dy & \text{if } x_* \in \E\backslash \p \E \ \text{ and }  \mathcal{Y}_* = \Phi_\eta (\j) \ \text{for some } \j \in \overline{\ba}_{\mu_0 (x_*)},
\\ 0 & \text{if } \ x_*\in \p \E\\
+\infty & \text{otherwise.}
\end{cases}
\end{equation}

This is the $\Gamma$-$\liminf$ assumption of \cite[Theorem 7]{ss2d} and we may thus apply this theorem to pass to the limit $n\to \infty$. Let 
\[  
P_{n,\eta}:= \dashint_{\E} \delta_{(x,\nabla \het (n ^{1/d} x + \cdot) )} dx
\]
or in other words the push-forward of the normalized Lebesgue measure on $\E$ by 
$$x\mapsto \left(x, \nab \het( n^{1/d} x + \cdot)\right).$$
Theorem 7 of \cite{ss2d} yields that 
\begin{itemize}
 \item $P_{n,\eta}$ converges to a Borel probability measure $P_\eta$ on $X$
 \item $P_\eta$ is $T_{\lambda(x)}$-invariant and  its  marginal with respect to $x$ is $\frac{1}{|\E|}dx_{|\E}$
 \item $P_\eta-$a.e. $(x, \j)$ is of the form $\lim_{n\to \infty} (x_n,   \theta_{n^{1/d }x_n }   \mathcal{Y}_n)$,
\end{itemize}
and moreover
\begin{equation}\label{bi1}
\liminf_{n\to \infty} \mathbf{F}_n(\mathcal{Y}_n) \ge \int \mathbf{f}(x,\mathcal{Y}) dP_\eta (x,\mathcal{Y}) = \int \left(\lim_{R\to \infty} \dashint_{K_R}\mathbf{f}(x, \theta_\lambda \mathcal{Y})\, d\lambda \,\right)  dP_\eta (x,\mathcal{Y})
\end{equation} 
where $\mathbf{f}$ is defined in \eqref{eq:Gamma lim inf 2}. It is also a part of the result (consequence of the ergodic theorem) that the limit $R\to \infty$ in the above exists. From \eqref{eq:Gamma lim inf 2}, and since $\p \E$ is of Lebesgue measure $0$ (by assumption \eqref{ass1})  we also see that for $P_{\eta}$-a.e. $(x, \mathcal{Y})$ it must be that $\mathcal{Y}= \Phi_\eta (\j), \j \in \overline{\ba}_{\mu_0(x)}$, and $\mathbf{f}(x, \mathcal Y) =\int_{\mr^d}\chi(y)|\mathcal{Y}|^2(y)\, dy $.
Combining \eqref{Ff} with \eqref{bi1} we get
\begin{equation}\label{pmino}
\liminf_{n\to \infty} \frac{1}{n}\int_{\mr^d} |\nab \het|^2 \ge |\E| \int \left(\lim_{R\to \infty} \dashint_{K_R} \chi* \indic_{K_R} |\mathcal Y|^2\right) \, dP_\eta(x,\mathcal Y).
\end{equation}
Since $\chi * \indic_{K_R} \ge \indic_{K_{R-1}}$ we can replace this by
\begin{equation}\label{pmino2}
\liminf_{n\to \infty} \frac{1}{n}\int_{\mr^d} |\nab \het|^2 \ge |\E| \int \left(\lim_{R\to \infty} \dashint_{K_R} |\mathcal Y |^2\right) \, dP_\eta(x,\mathcal Y).
\end{equation}
Now the above is true for all $\eta>0$. From the convergence $P_{n,\eta}\to P_\eta$,  and since  the bijection $\Phi_\eta$ commutes with the convergence in \eqref{cvh2} (see the proof of Lemma \ref{lemprel}), we may check  that $P_{\nu_n}$ (as defined in \eqref{eq:Pnun}) converges weakly to $P$ (independent of $\eta$), the push-forward of $P_\eta$ by $\Phi_\eta^{-1}$, for any $\eta$. By the above results on $P_\eta$, and Lemma \ref{lemprel}, we deduce that $P$ {is admissible}. %satisfies Items 1 and 2 in Proposition~\ref{promino}.

Moreover, by definition of the push-forward, \eqref{pmino2} means that
\begin{multline*}\liminf_{n\to \infty} \frac{1}{n}\int_{\mr^d} |\nab \het|^2 - (\kappa_d \g(\eta)+ \gamma_2 \indic_{d=2}) \\ \ge
|\E| \int \left(\lim_{R\to \infty}  \dashint_{K_R} |\Phi_\eta(\j)|^2 -( \kappa_d w(\eta)+ \gamma_2 \indic_{d=2}) \mu_0(x) \right)    \, dP(x,\j),\end{multline*}
where we have used  the fact that $\int \mu_0=1$ and the first marginal of $P $ is the normalized Lebesgue measure on $\E$.
Combining with   \eqref{pred2} and the definition of $\W_\eta$, we obtain 
\begin{equation}\label{pre1}
\liminf_{n\to \infty}  n^{2/d-2} \left(\w(x_1, \dots, x_n) -n^2\I[\mu_0] + \left(\frac{n}{2}\log n \right) \indic_{d=2}\right) \ge \frac{|\E|}{c_d} \int \W_\eta(\j) \, dP(x, \j)-C\eta^2.
\end{equation}
To conclude, there remains to let $\eta\to 0$ in the above. Indeed, for $P$-a.e. $(x,\j)$, $\j \in \overline{\ba}_{\mu_0(x)}$. Since $\mu_0$ is bounded and independent of $\eta$, the scaling property \eqref{eq:scale renorm}, \eqref{ass2} and the uniform lower bound on $\bai$ of Proposition \ref{Wbb} thus imply that 
$ \W_\eta(\j)$ is bounded below independently of $\eta$  for $P-\mathrm{a.e.} \, (x,\j). $
We can then simply use Fatou's lemma, to deduce that
\[
\liminf_{n\to \infty} n^{2/d-2}\left(\w(x_1, \dots, x_n) -n^2\I[\mu_0]+\left(\frac{n}{2}\log n\right)\indic_{d=2} \right) \ge\frac{|\E|}{c_d}\int  \liminf_{\eta \to 0} \W_\eta(\j)\, dP(x,\j), 
\]
which is \eqref{222}, by definition  of $\W$ {and $\widetilde{\W}$}. \hfill \qed

\begin{remark} \label{remzeta} Note that the positive term $2n\sum_{i=1}^n \zeta(x_i)$ has been discarded in \eqref{pred2} and never used in the lower bound, so it can be reintroduced as an additive term in the right-hand side  of \eqref{222}, which will be useful when studying the Gibbs measure in  Section \ref{sec8}. \end{remark}

\section{Screening and lower bound for the smeared jellium energy}\label{sec:univ low bound}

In this section we prove Proposition \ref{Wbb}, whose statement we recall: 
\begin{pro}[\textbf{Lower bound on the smeared jellium functional}]\label{pro:Wbb}\mbox{}\\
$\W_\eta$ is bounded below on $\bai$ by a constant depending only on the dimension.
\end{pro}

In view of the definition \eqref{We}, the method will consist in bounding from below the minimal energy over  a sequence of  cubes of size $R\to \infty$. For that we prove 
\begin{itemize}
\item that by minimality we can reduce to configurations with simple and well-separated points (at least away from the boundary of the cube), i.e. at distances bounded below {\it independently of } $\eta$ and $R$. 
This is done in Section \ref{sec:sep points} via an argument adapted from a Theorem of E. Lieb \cite{Lie} for the case $\eta = 0$ (whose statement and proof can be found in \cite{rns}). 
\item that  we can ``screen" efficiently such a configuration, i.e. modify it  close to the boundary of the cubes to make  the normal component of the electric field vanish on the boundary of the cube and make the points well-separated all the way to the boundary of the cube, at a negligible energy cost. (The vanishing normal component will in particular impose the total number of points in the cube).
This is carried out in Section \ref{sec:screening}.
\item Once this is done, we will be able to immediately bound from below the minimal energy on these large cubes via Lemma \ref{lemequiv}.\end{itemize}
\medskip

In all this section $K_R$ is a hyperrectangle whose side-lengths are in $[2R,3R]$. For the actual proof of Proposition \ref{Wbb} we need consider only hypercubes but we will need more general shapes later in the paper.

\subsection{Separation of points}\label{sec:sep points}

Let us consider the following variational problem:
\begin{equation}\label{eq:ener cube R}
\EetR := \inf \Big\{ \int_{K_R} |\nabla h | ^2 ,\: -\Delta h =  c_d \Big( \sum_{p\in \La} N_p \delta_p ^{(\eta)} -1 \Big) \mbox{ in } K_R,\  \text{for some discrete set} \ \Lambda \subset K_R \Big\}.
\end{equation}
Obviously $\EetR \geq 0$ and the infimum exists. We do not address the question of whether there exists a minimizer. Note that the points $p$ may in principle depend on $\eta$.

\begin{pro}[\textbf{Points ``minimizing'' $\EetR$ are well-separated}]\label{pro:sep points}\mbox{}\\
Let $\La$ be a discrete subset of $K_R$ and $h$ satisfy
\begin{equation}\label{eq:Laplace h}
-\Delta h =  c_d \Big( \sum_{p\in \La} N_p \delta_p ^{(\eta)} -1\Big)  := c_d (  \mu_h -1) \quad \text{in} \ K_R.
\end{equation}
Denote $\La_R = \La \cap K_{R-1}$. There exists three positive constants $\eta_0,r_0,C$ such that if $\eta < \eta_0$, $R$ is large enough and one of the following conditions does not hold:
\begin{eqnarray}\label{cond1}& \forall p  \in \La_R, \quad  N_{p} = 1\\
\label{cond2} & \forall p \in \La_R, \ \dist(p,\La_R \setminus \{ p \})\ge r_0,
\end{eqnarray}
then there exists $\Lat$ a discrete subset of $K_R$ and an associated potential $\tilde{h}$ satisfying
\begin{equation}\label{dhtilde}
-\Delta \hti =  c_d \Big( \sum_{p\in \Lat} N_p \delta_p ^{(\eta)} -1\Big)\quad \text{in}\  K_R
\end{equation}
such that
\[
\int_{K_R} |\nabla \hti| ^2 \leq  \int_{K_R} |\nabla h| ^2 - C.
\]
\end{pro}

What this proposition says is that if two points in the configuration are too close to one another (independently of $R$ and $\eta$), one can always find another configuration with strictly less energy. In other words, in a minimizing sequence for $\EetR$ one can always assume that \eqref{cond1}--\eqref{cond2} hold for some $r_0>0$ (depending only on $d$).
Note that the failure of \eqref{cond1} may be seen as an extreme case of the failure of \eqref{cond2}. Without loss of generality we may  thus prove the result when \eqref{cond2} fails and  $N_p = 1$ for all $p$.

\medskip 

We start with a lemma that gives explicitly the variation of $\int_{K_R} |\nabla h|^2$ when one point in the configuration is moved and the variation of  the potential is chosen to satisfy Dirichlet boundary conditions. We shall denote $G_R$ the constant  $c_d $ times the Green-Dirichlet function of the hyperrectangle $K_R$:  
\begin{equation}\label{eq:def Green Dir}
\begin{cases}
-\Delta_x G_R (x,y) = c_d \delta_y (x) \mbox{ if } x \in K_R  \mbox{ and } y\in K_R\\
G_R (x,y) = 0 \mbox{ if } x\in \partial K_R \mbox{ or } y\in \partial K_R. 
\end{cases}
\end{equation}
Note that $G_R$ is a symmetric function of $x$ and $y$: $G_R(x,y) = G_R(y,x)$. Associated to the Green-Dirichlet function is the quadratic form 
\begin{equation}\label{eq:forme Dirich}
D_R (\mu,\nu) = \iint_{K_R \times K_R} \mu (dx) G_R (x,y) \mu(dy)
\end{equation}
defined for measures $\mu, \nu$, that plays a role analogue to \eqref{eq:def Coul ener}.

We will consider specific variations of $\EetR$ that can be expressed in terms of $G_R$:

\begin{lem}[\textbf{Variations of $\EetR$}]\label{lem:vari cube}\mbox{}\\
Let $h$, $\Lambda$  and $\mu_h$ be as in Proposition \ref{pro:sep points}. Let $x\in\La$ and $y\in K_R$ and define 
$$\Lat = (\La \setminus \{ x\} ) \cup \{ y\}$$
with an associated $\hti = h + \hb$ where the variation $\hb$ is defined as the unique solution to
\begin{equation}\label{eq:defi vari}
\begin{cases} -\Delta \hb =  c_d( \deleta_y - \deleta_x)   \mbox{ in } K_R\\
 \hb = 0 \mbox{ on } \partial K_R.
\end{cases}
\end{equation}
We have
\begin{multline}\label{eq:vari cube}
\int_{K_R} |\nabla \hti| ^2 =  \int_{K_R} |\nabla h| ^2  + c_d D_R \left( \deleta _y, \deleta _y \right) - c_d  D_R \left( \deleta _x, \deleta _x \right) \\ + 2  c_d D_R \left( \mu_h - \deleta_x -1, \deleta _y - \deleta_x \right).
\end{multline}
\end{lem}

\begin{proof}
Let us first note that $\hti$ is an admissible trial state for $\EetR$ since we immediately check that it satisfies \eqref{dhtilde}.
Expanding the square, integrating by parts the cross-term and using the Dirichlet boundary condition for $\hb$ in addition to Equation \eqref{eq:Laplace h} we obtain
\begin{align*}
\int_{K_R} |\nabla \hti| ^2 &= \int_{K_R} |\nabla h| ^2 + \int_{K_R} |\nabla \hb| ^2 + 2 \int_{K_R} \nabla \hb \cdot \nabla h \\
&= \int_{K_R} |\nabla h| ^2 + \int_{K_R} |\nabla \hb| ^2 - 2 \int_{K_R} \hb \Delta h  \\
&= \int_{K_R} |\nabla h| ^2 + \int_{K_R} |\nabla \hb| ^2 + 2 c_d \int_{K_R} \hb (\mu_h - 1).
\end{align*}
We now use the Green representation of $\hb$ to write
\[
\hb (z)= \int_{K_R} G_R (z,z') \deleta_y (z')dz' - \int_{K_R} G_R (z,z') \deleta_x (z')dz':= \hb_y (z) - \hb_x (z).
\]
We then have (using the Dirichlet boundary condition again to integrate the cross term by parts)
\begin{align*}
\int_{K_R} |\nabla \hb| ^2 &= \int_{K_R} |\nabla \hb_x| ^2 + \int_{K_R} |\nabla \hb_y| ^2 - 2 \int_{K_R} \nabla \hb_x \cdot \nabla \hb_y \\
&=  c_d D_R (\deleta_x,\deleta_x) + c_d  D_R (\deleta_y,\deleta_y) -2 c_d  D_R(\deleta_x,\deleta_y)
\end{align*}
and
\begin{equation*}
\int_{K_R} \hb (\mu_h - 1) = D_R (\deleta_y - \deleta_x, \mu_h - 1) = D_R ( \deleta_y - \deleta_x, \mu_h - \deleta_x - 1) + D_R ( \deleta_y - \deleta_x, \deleta_x).
\end{equation*}
Putting everything together yields the desired result.
\end{proof}

The next step is to remark that for large $R$ we should have $G_R (x,y)\simeq w(x-y)$. It is then natural to expect that the difference between the self-interactions 
\[
D_R \left( \deleta _y, \deleta _y \right) - D_R \left( \deleta _x, \deleta _x \right)
\]
should be small, at least if the points $x$ and $y$ do not approach the boundary of the domain and thus do not see the Dirichlet boundary condition. On the other hand we can rewrite
\begin{equation}\label{eq:vari intui}
D_R (\deleta_y - \deleta_x, \mu_h - \deleta_x - 1) = \int \delta_y^{(\eta)} h_x - \int \delta_x^{(\eta)} h_x \simeq h_x(y)-h_x(x)
\end{equation}
where
\begin{equation}\label{eq:vari intui 2}
h_x := \int_{K_R} G_R(.,y) \left(\mu_h-\deleta_x - 1\right)(y)dy.
\end{equation}
 Were we dealing with the case of point charges $\eta = 0$, the   variation would thus reduce to $h_x(y)-h_x(x)$.
Lemma \ref{lem:vari cube} is then a rephrasing in our context where charges are smeared out of the well-known  fact that the condition for optimality is that the charge at $x$ lies at the minimum of the potential $h_x$ generated by all the other charges and the neutralizing background.

We thus carry on with estimates of the Green-Dirichlet function that will allow us to make this intuition rigorous in the case of smeared out charges. We will use the well-known fact that $G_R (x,y)$ may be written as the sum of the Green function of the whole space $w(x-y)$ and a regular part. Important for us is the fact that the regular part is uniformly bounded on $K_R$ with a bound independent of $R$, such that we will be able to focus on the singular part $w(x-y)$ when estimating the variation in the right-hand side of \eqref{eq:vari cube}.

Actually, when $d=2$, things are a little bit more subtle due to the fact that the Green function does not decay at infinity. Our estimate \eqref{eq:estim Green 2D} is thus different in this case, but still sufficient for our proof.

\begin{lem}[\textbf{Estimates on the Green function}]\label{lem:Green}\mbox{}\\
$\bullet$ If $d\geq 3$, there exists a constant $C_G$, independent of $R$, such that
\begin{equation}\label{eq:estim Green}
\left| G_R (x,y) - w(x-y)\right| \leq C_G
\end{equation}
for all $x,y\in K_R$ with $\min(\dist(x,\partial K_R),\dist(y,\partial K_R)) \geq 1$.

\noindent $\bullet$ If $d=2$, there exists a constant $C_G$, independent of $R$, such that
\begin{equation}\label{eq:estim Green 2D}
\left| G_R (x,y) - w(x-y) + w(x-y') \right| \leq C_G
\end{equation}
where $y'$ is the reflection of $y$ with respect to $\partial K_R$. In particular, if $\dist(y,\partial K_R) \geq 1$ and $x,x' \in B(y,c)$ for some constant $c$, then 
\begin{equation}\label{eq:estim Green 2D bis}
\left| \left(G_R (x,y) - G_R (x',y)\right) - \left(w(x-y) - w(x'-y)\right) \right| \leq C_c 
\end{equation}
for some $C_c$ depending only on $c$.
\end{lem}

\begin{proof} 
The proof of \eqref{eq:estim Green} is a well-known argument: Since $G_R(x,y) = G_R(y,x)$ we may restrict to the case where $\dist(y,\partial K_R)\geq 1$ and consider $G_R(x,y)$ and $w(x-y)$ as functions of $x$ only. Now, by definition $G_R(.,y)-w(.-y)$ is harmonic in $K_R$ and hence (by the maximum principle) reaches its maximum and minimum on the boundary of $K_R$. The Green-Dirichlet function is zero there and in view of the assumption that $\dist(y,\partial K_R)\geq 1$ we have
$$\min_{x\in \partial K_R} (- w(x-y)) \geq - C$$
independently of $R$ and
$$\max_{x\in \partial K_R} (- w(x-y))\leq - w(C R)\to 0$$
when $R\to \infty$. We used the asumption $d\geq 3$ in the last equation.

To prove \eqref{eq:estim Green 2D}, we fix $y\in K_R$ and note that the function 
\[
f(x,y) :=  G_R (x,y) - w(x-y) + w(x-y') 
\]
is harmonic as a function of $x\in K_R$ since $\Delta_x G_R (x,y) = \Delta_x w(x-y)$, $\Delta_x w (x-y') = - c_d \delta_{y'} (x)$  and $y' \notin K_R$. We thus know that $f(x,y)$ reaches its maximum and minimum on the boundary $\partial K_R$: 
\begin{equation}\label{eq:bound Green 2d}
 \inf_{x\in \partial K_R} \log \frac{|x-y|}{|x-y'|} \leq f(x,y) \leq \sup_{x\in \partial K_R} \log \frac{|x-y|}{|x-y'|},
\end{equation}
where we used the Dirichlet boundary condition for $G_R(.,y)$ and the definition of $w$ when $d=2$. Next we observe that, $y'$ being the reflection of $y$ with respect to $\partial K_R$ it holds that 
\[
 |x-y|\leq |x-y'| 
\]
for any $x\in \partial K_R$ and thus 
$f(x,y) \leq 0$
for any $x,y\in K_R$. For a lower bound we note that 
$|x-y'| \leq |x-y| + |y-y'|  $
with equality when $x,y$ and $y'$ are aligned, in which case 
\[
\frac{|x-y|}{|x-y'|} = 1 - \frac{|y-y'|}{|x-y'|} = 1 - \frac{|y-y'|}{L + |y-y'| /2}   
\]
where $L$ is one of the side-lengths of the hyperrectangle $K_R$. By assumption $2R \leq L\leq 3 R $, which implies that $|y-y'|\leq 3R$ and we easily conclude that $\inf_{x\in \partial K_R} \log \frac{|x-y|}{|x-y'|}$ is bounded below independently of $R$. Plugging in \eqref{eq:bound Green 2d} concludes the proof of \eqref{eq:estim Green 2D}. The estimate \eqref{eq:estim Green 2D bis} follows immediately by noting that, if $\dist(y,\partial K_R) \geq 1$ then $| \nabla w(.-y') |$ is bounded above independently of $R$ in $K_R$.  
\end{proof}

We can now give the main idea of the proof of Proposition \ref{pro:sep points}, following \cite{Lie} for singular charges. Suppose there are two points, say $0$ and $x$ in $\La$, very close to one another\footnote{Translating the whole system we may always assume that $0\in \La$}. Then one can decrease the energy by sending $x$ to some point $y$ with $\dist(0,y) > \dist(0,x)$ and changing the potential as described in Lemma \ref{lem:vari cube}. As we will prove below, it is sufficient to consider only the contribution of the charge  at $0$ and of the background at some fixed distance thereof to bound from above the variation of the potential $h_x$ defined in \eqref{eq:vari intui}. Then, if the distance $\dist(0,x)$ is really small to begin with, Lemma \ref{lem:Green} tells us that we can approximate the variation \eqref{eq:vari intui} in the potential $h_x$ using the Coulomb potential $w(x-y)$. This observation leads us to the definition
\begin{equation}\label{eq:defi Veta}
V_{\eta} =  \int_{\R ^d} w (.-z) \left( \deleta_0 - \1_{B(0,r_3)} \right) (z)dz
\end{equation}
of the potential generated (via the Green function $w$ of the full space) by a unit charge at~$0$ smeared out at scale $\eta$ and a disc of  constant opposite density of charge, where we choose $r_3<1$ so that $|B(0,r_3)|<1$.

We now observe that $V_{\eta}$ is radial and decreasing. It is also obviously independent of $R$ and gets closer and closer to $V_0$, independent of $\eta$, when $\eta$ is small. Since $V_0(r)\to +\infty$ when $r\to 0$, one can make $V_{\eta}(y) - V_{\eta} (x)$ arbitrarily negative if $x\to 0$ while $y$ stays bounded at a fixed distance away from $0$. This is the content of the following

\begin{lem}[\textbf{The jellium potential close to a point charge}]\label{lem:model pot}\mbox{}\\
For any $M>0$ there exists $\min(\hal, r_3) > r_2 > r_1 > 0$ such that for any small enough $\eta$  and any $r<r_1$,  we have
\begin{equation}\label{eq:vari model pot}
V_{\eta}(r_2-\eta) - V_{\eta} (r+\eta) \leq -M.
\end{equation}
\end{lem}

\begin{proof}
By  Newton's theorem, we have  $V_{\eta} (r) = V_{0} (r)$ for $r \geq \eta$, thus since $V_\eta$ is decreasing,   we have  $V_\eta(r) \ge V_0(\eta)$ for $r<\eta$. Thus, since we will always choose $\eta < r_1 < r_2$, we can study only $V_{0}$, which is readily computed by taking advantage of the radiality of the charge distribution:
\begin{align*}
V_0(r) &= \left(1 - |B(0,r)|\right) w(r) + |\mathbb{S}^{d-1}| \int_{B(0,r_3)\setminus B(0,r)} w(t) t ^{d-1} dt \mbox{ if } 0 \leq r \leq r_3 \\
V_0(r) &= \left(1 - |B(0,r_3)|\right) w(r) \mbox{ if } r\geq r_3,
\end{align*} 
by Newton's theorem again. Simple considerations show that $V_0 (r) \to +\infty$ when $r\to 0$ and that $V_0$ is decreasing. The result follows.
\end{proof}

We can now proceed to the

\noindent \textit{Proof of Proposition \ref{pro:sep points}.}

As announced, we work assuming $N_p=1$ and \eqref{cond2} is violated.  
Without loss of generality (we may always translate $K_R$ without changing the energy) we assume that $0\in \La_R$ and that there exists some point $x\in \La_R$ close to $0$. We will prove that  if $\dist(0,x) < r_1$, where $r_1$ is as in Lemma \ref{lem:model pot} one can decrease the energy  by moving $x$ to a certain point $y\in \partial B(0,r_2)$.

Since we assume $\dist(0,\partial K_R) \geq 1$,  and since $r_1<\hal$, we have 
 $B(0,r_1) \subset K_{R-1/2}$.    We pick some $y$, to be chosen later, and  let $\Lat$ be the corresponding modified configuration, with an associated potential $\hti$  as in Lemma \ref{lem:vari cube}. The variation of the energy is given by  \eqref{eq:vari cube}, which we estimate   using \eqref{eq:estim Green} (respectively \eqref{eq:estim Green 2D bis} when $d=2$): for the first two terms we may write
\[
\left| D_R \left( \deleta _y, \deleta _y \right) - D_R \left( \deleta _x, \deleta _x \right)  \right| \leq C \left( \int \delta_x^{(\eta)} + \int \delta_y^{(\eta)} \right) \leq C
\] 
since the contribution of the singular part $w(x-y)$ to the Green function $G_R(x,y)$ yields the same contribution for both terms. We thus have
\begin{equation}\label{eq:estim vari 1}
\int_{K_R} |\nabla \hti| ^2 - \int_{K_R} |\nabla h | ^2 \leq 2 c_d D_R \left( \mu_h - \deleta_x -1, \deleta _y - \deleta_x \right) + C
\end{equation}
where $C$ is independent of $R$ and  $\eta$. We may then focus on proving that the first term of the right-hand side can be made arbitrarily negative. For that purpose,  we start from \eqref{eq:vari intui}, \eqref{eq:vari intui 2} and, following \cite{Lie}, we  decompose 
\begin{align}\label{eq:decomp h_x}
h_x &= \hint + \hext  \nonumber \\
\hint &= \int_{K_R} G_R(.,z) \left( \deleta_0 - \1_{B(0,r_3)} \right) (z) dz \nonumber \\
\hext &= \int_{K_R} G_R(.,z) \left( \sum_{p\in \La \setminus \{ x,0\} )} \deleta_p - 1 + \1_{B(0,r_3)} \right) (z) dz.
\end{align}
 We further define the $\eta$-averaged function $
\hext_a (z) :=   \hext * \delta_0^{(\eta)}$,  
and note that since $-\Delta \hext \geq 0$ in $B(0,r_3)$, and $r_2<r_3$, we have $-\Delta \hext_a \geq 0$ in $B(0,r_2)$, at least if $\eta$ is small enough. By the maximum principle $\hext_a$ reaches its minimum on $B(0,r_2)$ at the boundary and we now choose  $y\in \partial B(0,r_2)$ to be (one of) its minimum point(s). Then, in view of \eqref{eq:vari intui}, we have
\begin{align}\label{eq:estim vari 2}
D_R \left( \mu_h - \deleta_x -1, \deleta _y - \deleta_x \right) &=  \hext_a (y) - \hext_a (x) + \int\delta_y^{(\eta)} \hint - \int \delta_x^{(\eta)}  \hint \nonumber \\
& \leq \int \delta_y^{(\eta)} \hint - \int \delta_x^{(\eta)}  \hint \leq \int \delta_y^{(\eta)}  V_{\eta} - \int \delta_x^{(\eta)} V_{\eta} + C'
\end{align}
where on the second line we have used \eqref{eq:estim Green} (respectively \eqref{eq:estim Green 2D bis} when $d=2$)  by replacing $G_R(x,y)$ by $w(x-y)$ in the terms involving $\hint$. Assuming now that we have chosen $M= 2(C+c_d C') $ in     Lemma \ref{lem:model pot}, where $C$  is the constant in \eqref{eq:estim vari 1} and $C'$ that in \eqref{eq:estim vari 2}, we have 
\[
V_{\eta} (r_2-\eta) - V_{\eta} (r_1+\eta) \leq - 2 C -2 C'.\] 
Using that $V_\eta$ is radial and decreasing, we deduce that, since we assumed $|x|\le r_1$ and $|y|=r_2$, we have 
\[
\int \delta_y^{(\eta)} V_{\eta} - \int \delta_x^{(\eta)}  V_{\eta} \leq -2 C - 2 C'.
\]
Combining this with \eqref{eq:estim vari 2} and \eqref{eq:estim vari 1}, we have thus obtained  $\int_{K_R} |\nab \tilde h|^2- \int_{K_R} |\nab h|^2 \le - C$ for some $C>0$. This proves the proposition, with $r_0=r_1$.

\hfill \qed

\subsection{Screening}\label{sec:screening}

Starting with a configuration of points modified according to Proposition \ref{pro:sep points} we are able to modify it further and ``screen it" as announced. 
\begin{pro}[\textbf{Screening a configuration of points}]\label{proscreen} \mbox{}\\
There exists $\eta_0>0$ such that the following holds for all $\eta<\eta_0$. Let $K_R$ be a hyperrectangle such that $|K_R|$ is an integer and the sidelengths of $K_R$ are in $[2R,3R]$. Let  $\j_\eta=\nabla h_\eta$ be a gradient vector field with $h_\eta$ satisfying  \eqref{eq:Laplace h} with \eqref{cond1}--\eqref{cond2} satisfied and
\begin{equation}\label{eq:bound ener cube}
\int_{K_R} |\j_\eta| ^2 \leq C R ^d.
\end{equation}
Then there exists $\Lah$ a configuration of points and $\bar{\j}$ an associated gradient vector  field (both possibly also depending on $\eta$)  defined in $K_R $ and satisfying
\begin{align}
\label{eq:proscreen field}
-\div  \bar{\j} &= c_d\Big(  \sum_{p\in \Lah} \delta_p - 1\Big) \mbox{ in } K_R\\
\label{eq:proscreen flux}
\bar{\j} \cdot \vec{\nu} &=0 \quad \text{on} \ \p K_R \end{align}
such that for any $p\in \Lah$
\begin{equation}\label{eq:mindistr plus}
\min \left( \dist(p,\Lah\setminus \{p\}), \dist(p,\partial K_R) \right) \geq \frac{ r_0}{10}
\end{equation} with $r_0$ as in \eqref{cond2}, 
and
\begin{equation}\label{majow}
\int_{K_R} |\bar{\j}_\eta |^2 \le \int_{K_R} |\j_\eta| ^2 + o( R^d)
\end{equation}
as $R\to \infty$, where $\bar{\j}_\eta = \Phi_\eta \left( \bar{\j}\right)$ and the $o(R ^d)$ depends only on $\eta$ and the constants in \eqref{cond1}, \eqref{cond2}, \eqref{eq:bound ener cube}.
\end{pro}

Property \eqref{eq:proscreen flux}  is crucial in order to be able to paste together configurations defined in separate hyperrectangles. It also implies (integrating 
\eqref{eq:proscreen field} over $K_R$ and using Green's formula) that  the number of points of the modified configuration is \emph{exactly} equal to the volume of the domain:
\begin{equation}\label{eq:num points pre}
\# \Lah  = |K_R|,
\end{equation}
which is also important for the  proof of Proposition \ref{Wbb}.

The proof proceeds by modifying the configuration of points in the vicinity of $\partial K_R$:
\begin{enumerate}
\item We first select by a mean-value argument  a ``good boundary'' $\p K_t$ on which we control the $L^2$ norm of $\j_\eta$ (i.e. the energy density), and  at a distance $1 \ll L\ll R$ from the boundary of the original hyperrectangle.
\item We will not move the points whose associated smeared charges intersect $\partial K_t$. Instead, we isolate them in small cubes and leave unchanged all the points lying in $\Gamma$= the union of $K_t$ with these small cubes. A new configuration of points is built in $K_R\backslash \Gamma$.  It is built by splitting this set into hyperrectangles and using on each some test vector-fields (obtained by  solving explicit elliptic PDEs) whose energies is evaluated by elliptic estimates similarly to \cite[Sec. 4.3]{gl13}. The test vector-fields are then pasted together. 
\item It is more convenient to first ``straighten'' the boundary of $\Gamma$. We pick some $\alpha$ such that $\Gamma \subset K_{t+\alpha}$ and  put points in $K_{t+\alpha}\setminus \Gamma$ so that neither the energy nor the flux at the boundary are increased too much. To obtain a control on those quantities, it is important that no smeared charge intersects the boundary of $\Gamma$ to begin with, whence the second step. 
\end{enumerate}

We start with the preliminary lemmas containing the elliptic estimates which provide the elementary bricks of our construction.

\begin{lem}[\textbf{Adding a point without flux creation}]\label{lemcs1}\mbox{}\\
Let $\mathcal{K}$ be a hyperrectangle of center $0$ and sidelengths in $[l_1, l_2]$, with $l_2 \ge l_1 >0$,  and let $m=1/|\mathcal{K}|$.
The mean zero solution to
\begin{equation*}
\left\{\begin{array}{ll}
-\Delta u = c_d ( \delta_0 -m )& \text{in} \ \mathcal{K}\\
\nab u \cdot \vec{\nu} =0 & \text{on} \ \p \mathcal{K}
\end{array}\right.
\end{equation*}
satisfies
\begin{equation*}
\lim_{\eta\to 0}\left| \int_{\mathcal{K} } |\nab u_\eta|^2 - \kappa_d \g(\eta) \right|\le C
\end{equation*}
where $C$ depends only on $d, l_1, l_2$ and $m$. \end{lem}
\begin{proof} It is a very simple computation, similar for example to arguments employed in the proof of Lemma \ref{lemequiv}.
\end{proof}

\begin{lem}[\textbf{Correcting fluxes on hyperrectangles}]\label{lemcs2}\mbox{}\\
Let $\mathcal{K}$ be a hyperrectangle    with sidelengths in $[l/2,3l/2]$. Let $g\in L^2(\p \mathcal{K})$ be a function which is $0$ except on one face of $\mathcal{K}$.  Let $m$ be defined by
%$$c_d(m-1) |\mathcal{K} | = - \int_{\p \mathcal{K}} g $$ 
\begin{equation}\label{eq:flux m}
 m = 1 - c_d ^{-1} |\mathcal{K} | ^{-1} \int_{\p \mathcal{K}} g 
\end{equation}
and assume $|m-1|<1/2$. Then the mean zero  solution to
\begin{equation}\label{eqnu}
\left\{\begin{array}{ll}
-\Delta u  =c_d ( m-1)   & \text{in} \ \mathcal{K}\\
\nab u \cdot \vec{\nu} =g & \text{on} \ \p \mathcal{K}
\end{array}\right.
\end{equation}
satisfies
\begin{equation}\label{estlcs2}
\int_{\mathcal{K}} |\nab u|^2 \le  C l \int_{\p \mathcal{K}} g^2,
\end{equation}
where $C$  depends only on $d$.
\end{lem}

\begin{proof}
We may write $u=u_1 + u_2  $ where
\begin{equation*}
\left\{\begin{array}{ll}
-\Delta u_1  = c_d (m-1 ) & \mbox{ in } \ \mathcal{K}\\
\nab u_1 \cdot \vec{\nu} = \bar{g} & \mbox{ on } \ \p \mathcal{K}
\end{array}\right.
\end{equation*}
with $\bar{g}$ a constant on the face of $\p \mathcal{K}$ where $g$ is nonzero, and zero otherwise (integrating the equation we find $\bar{g}= c_d(m-1)l$)  and \begin{equation*}
\left\{\begin{array}{ll}
-\Delta u_2  =0  & \text{in} \ \mathcal{K}\\
\nab u_2 \cdot \vec{\nu} = g- \bar{g}  & \text{on} \ \p \mathcal{K}.
\end{array}\right.
\end{equation*}
First we note that by \eqref{eq:flux m} and the Cauchy-Schwarz inequality we have
\begin{equation}\label{mm1}
c_d |m-1| \le  \frac{1}{|\mathcal{K}|}\int_{\p \mathcal{K}}  |g|  \le C  l^{\frac{-d-1}{2}} \left(\int_{\p \mathcal{K}} g^2 \right)^{1/2}.
\end{equation}

Next we observe that $u_1$ can be solved explicitly: if the face where $g$ is nonzero is  included in $\{x_1=0\}$ (after translation and rotation), then $u_1(x)= \frac{m-1}{2} (x_1-l)^2$ and we may compute
\begin{equation}\label{nrju1}
\int_{\mathcal{K}} |\nab u_1|^2 \le C (m-1)^2 l^{d-1}l^3 = C (m-1)^2 l^{d+2}.
\end{equation}
Secondly,  by scaling we check that
\begin{equation}\label{nrju3}
\int_{\mathcal{K}} |\nab u_2|^2 \le C l \left\|\nab u_2 \cdot \vec{\nu}\right\|_{L^2 (\p \mathcal{K})}^2
 \le C l \left(\|g\|_{L^2(\p \mathcal{K})}^2 + (m-1)^2 l^2 l^{d-1} \right)\end{equation}
Combining \eqref{nrju1},  \eqref{nrju3}, and inserting \eqref{mm1}, we obtain the result.
\end{proof}

\begin{proof}[Proof of Proposition \ref{proscreen}]
We set two lengths $l $ and $L$ such that $1\ll l \ll L \ll R $ as $R \to \infty$, and to be determined later. In all the sequel we denote by $C_1$ a generic constant which is a multiple (by a universal factor) of the constant $C$ of \eqref{eq:bound ener cube}.\smallskip

\noindent
- {\it Step 1: We find a good boundary.}
First we  claim that
there exists $t\in [R-2L, R-L]$ such that
\begin{equation}\label{bonbordg}
\int_{\p K_t} |\j_\eta|^2  \le  C_1 R^d/ L
\end{equation}
and
\begin{equation}
\label{contrautour}
\int_{K_{t+1}\backslash K_{t-1}} |\j_\eta|^2 \le  C_1 R^d/L.
\end{equation}
Indeed, since \eqref{eq:bound ener cube} holds, we may split $K_{R-L}\backslash K_{R-2L}$ into $L/2$ ``annular" regions of width~$2$. On one of them \eqref{contrautour} must hold, and by a mean-value argument we may also shift $t$ so that  \eqref{bonbordg} holds.

Since we assume \eqref{cond1}--\eqref{cond2} and since $L\ge 1$,  the boundary $\p K_t $ intersects balls $B(p, \eta)$ with $p\in \Lambda$ which are all at distance $\ge \hal r_0>0$ from all other such balls (for that it suffices to take $\eta_0<r_0/4$). Without l.o.g. we assume $r_0<1$.
Let $\Lambda_0$ denote the set of $p$'s for which $B(p,\eta) \cap \p K_t\neq \varnothing$.
For each $p\in \Lambda_0$, let $K_s(p)$ denote the hypercube of sidelength $s$ centered at $p$.
By a mean-value argument similar to the above, and using \eqref{contrautour} we can find $ s \in [r_0/16,r_0/8]$ such that
\begin{equation}\label{petitca}
\sum_{p\in \Lambda_0} \int_{\p K_s(p)} |\j_\eta|^2 \le  C_1 R^d/L.
\end{equation}
We now denote $\Gamma= K_t \cup \left( \cup_{p\in \Lambda_0} K_s(p)  \right) $. By construction $\p \Gamma$ does not intersect any $B(p, \eta)$ and we have
\begin{equation}
\label{Gamma}
\int_{\p \Gamma} |\j_\eta|^2 \le  C_1 R^d/L.
\end{equation}
By perturbing this construction a little we may define the resulting set $\Gamma$ such that $|\Gamma| \in \mn$ and~\eqref{Gamma} still holds. 

\medskip

\noindent {\it - Step 2: definition of $\hj$ in $K_{t+\alpha}$.}
First we define $\hj$ to be $\Phi_\eta^{-1}(\j_\eta)$ in $\Gamma$, that is we set 
$$ \hj = \j_\eta - \sum_{p\in\Lambda\cap \Gamma } f_{\eta} (x-p)\quad \text{in} \ \Gamma $$
where $\Lambda$ is the set of points corresponding to $\j_\eta$ as in~\eqref{eq:Laplace h} and $f_\eta$ is defined as in~\eqref{eqf0}. We note that this definition is unambiguous and coincides with the restriction to $\Gamma$ of $\Phi_{\eta}^{-1} (\j_\eta)$ since $\p \Gamma $ does not intersect any $B(p, \eta)$.

 We next wish  to extend $\hj$ to $K_{t+\alpha}\backslash \Gamma$. We take $\alpha>1$ to be such that $|K_{t+\alpha}\backslash \Gamma|$ is an integer. Since $|\Gamma| \in \mn$ by construction and $|K_R| \in \mn$ by construction we thus have 
\begin{equation}\label{eq:quant volume}
|K_{t+\alpha}| \in \mn \mbox{ and } |K_R \setminus K_{t+\alpha}| \in \mn.  
\end{equation}
Then $K_{t+\alpha}\backslash \Gamma$ can be split into  a disjoint union of hyperrectangles of volume 1, whose sidelengths are all bounded below by $r_0/8$, and bounded above by $1$, i.e. by universal constants.
In each of these hyperrectangles we apply Lemma \ref{lemcs1} with $m=1$. Since the  normal derivatives of the functions obtained that way are zero on the boundary, their gradients can be glued together into  a global vector field, and no divergence will be  created at the boundaries. More precisely this means that there exists a vector field $X$  defined in
 $K_{t+\alpha}\backslash \Gamma$ and satisfying
$$-\div X = c_d \Big( \sum_{p\in \Lambda_1} \delta_p -1\Big)  \quad \text{in} \ K_{t+\alpha}\backslash \Gamma$$
for some set $\Lambda_1$ equal to the union of the centers of these hyperrectangles,
and
\begin{equation}\label{nrjf}
\int_{K_{t+\alpha}\backslash \Gamma}|X_\eta|^2 \le C_\eta R^{d-1} \end{equation}
since the energy cost given by Lemma \ref{lemcs1} is proportional to  the volume concerned.
Next we claim that we can find a vector field  $Y$ satisfying $\div Y=0$ in $K_{t+\alpha}\backslash \Gamma$ and $g$ defined over $\p K_{t+\alpha}$ such that
$Y \cdot \vec{\nu}  =g $ on $\p   K_{t+\alpha}$, while $Y \cdot \vec{\nu} = \j_\eta\cdot \vec{\nu} $ on $\p \Gamma$ and
\begin{equation}\label{borg}
\int_{\p K_{t+\alpha}} g^2 \le C\int_{\p \Gamma} |\j_\eta|^2 \le C_1R^d/L,\end{equation}
and
\begin{equation}\label{nrju}
\int_{K_{t+\alpha}\backslash \Gamma} |Y|^2 \le C_1 R^d/L.\end{equation}
To see this,  we may split $K_{t+\alpha}\backslash \Gamma$ into hyperrectangles $\mathcal{K}_i$ which all have one face included in $\partial \Gamma$  and one face included in $\partial K_{t+\alpha}$, and  whose sidelengths are all bounded above by $2$ and below by $r_0/8$.
Then we solve in each $\mathcal{K}_i$
$$\left\{\begin{array}{ll}
\Delta u_i= 0 & \text{in} \ \mathcal{K}_i\\
\nab u_i \cdot \vec{\nu}= \j_\eta\cdot \vec{\nu} & \text{on} \ \p \mK_i\cap \partial \Gamma\\
\nab  u_i \cdot \vec{ \nu}= g_i & \text{on} \ \p \mK_i \cap \p K_{t+\alpha}\\
\nab u_i \cdot \vec{\nu}= 0 & \text{else on } \ \mK_i
\end{array}\right.$$
where $g_i$ is the unique constant that makes the equation solvable. It is straightforward to check that $\int_{\p \mK_i\cap \p K_{t+\alpha}} g_i^2 \le C \int_{ \p \mK_i \cap \p \Gamma} |\j_\eta|^2 $ (again by Green's formula applied to the equation).
Defining $g$  on $\p K_{t+\alpha}$ as $g=g_i$ on $\p K_{t+\alpha}\cap \p \mK_i$, we then have  \eqref{borg}.
Since $u_i$ is harmonic, we also  have the estimate
$$\int_{\mK_i} |\nab u_i|^2 \le C l \left\|\nab u_i \cdot \vec{\nu}\right\|_{L^2(\p K_i)}^2\le C \int_{\p \mK_i \cap \p \Gamma} |\j_\eta|^2 .
$$
Defining now the vector field $Y$ to be $\nab u_i$ on each $\mK_i$, we see that $Y$ satisfies the desired properties and that \eqref{nrju} holds.

We finally set  $\hj = X + Y $ in $ K_{t+\alpha}\backslash \Gamma$. It satisfies
$-\div \hj =  c_d \left( \sum_{p\in \Lambda_1} \delta_p -1\right)  $ and, since
no divergence is created at the interface $\p \Gamma$, the vector field  $\hj$ now defined in the whole on $K_{t+\alpha}$  still satisfies
 \begin{equation}\label{concleq}
 \left\{\begin{array}{ll} 
 -\div \hj = c_d \Big(  \sum_{p\in \Lambda_1 \cup (\Gamma \cap \Lambda) } \delta_p  -1 \Big)  \quad & \text{in} \ K_{t+\alpha}\\
 \hj \cdot \vec{\nu}=g & \text{on} \ \p K_{t+\alpha}.
 \end{array}\right.
 \end{equation}
Moreover,  in view of \eqref{nrjf} and \eqref{nrju}, we have
\begin{equation}\label{franginter}
\int_{K_{t+\alpha}} |\hj_\eta|^2 \le \int_{\Gamma} |\j_\eta|^2 +  C_\eta R^{d-1} + C R^d/L.\end{equation}
 Note that  by construction the distances between the points in $\Lambda_1 \cup (\Gamma \cap \Lambda)$ are all bounded below by   $r_0/8$.
We now discard the notations used for this step, except for  the conclusions \eqref{concleq}--\eqref{franginter} and \eqref{borg}.

\smallskip

\noindent
{- \it Step 3:  correcting the flux}. We construct a new domain\footnote{We have $\Gamma \subset K_{t+\alpha} \subset \Gamma' \subset K_R$. Of these sets, only $K_{t+\alpha}$ and $K_R$ are hyperrectangles.} $\Gamma'$ on the boundary of which we make the flux vanish. This uses Lemma~\ref{lemcs2} in an essential way to construct a configuration in $\Gamma' \setminus K_{t+\alpha}$. 

First we split $\p K_{t+\alpha}$ into  $O((R/l)^{d-1})$ hyperrectangles (of dimension $d-1$) $I_i$ of sidelength $\in [l/2,3l/2]$.
For each of them we consider a hyperrectangle (of dimension $d$) included in $K_R\backslash K_{t+\alpha}$ which has one side equal to $I_i$. By perturbing the sizes of the sides, we may have a hyperrectangle with aspect ratios in $[1/2,3/2]$ (so all sides have sizes in $[l/2, 3l/2]$).  This forms a disjoint collection $\mathcal{K}_i$ (we use the same notation as in the previous step, even though it does not correspond to the same rectangles). We let $g_i$ be the restriction of $g$ on $\p K_{t+\alpha}$ to $I_i$,  and let $m_i$ be defined by
\begin{equation}\label{eq:pick m_i}
c_d (m_i-1) |\mathcal{K}_i | = - \int_{\p \mathcal{K}_i} g_i .
\end{equation}
Let us check that $|m_i-1|<\hal$.
Using the Cauchy-Schwarz inequality and  \eqref{borg}  we have
\begin{equation}\label{eq:estim m_i}
 c_d|m_i-1|\le  l^{\frac{-d-1}{2}} \left(\int_{\p \mathcal{K}_i} g_i^2\right)^{1/2} \le  l^{\frac{-d-1}{2}}  \left(\int_{\p K_{t+\alpha} }g^2\right)^{1/2}  \le C l^{-1/2-d/2}  \frac{R^{d/2}}{L^{1/2}}.
\end{equation}
It is clear that we may choose 
\begin{equation}\label{eq:lengthscales}
1\ll l\ll L\ll R 
\end{equation}
such that this is $o(1)$ as $R\to \infty$, which we do from now on. Since $m_i \sim 1$ and $|\mathcal{K}_i|\sim C l ^d \gg 1$ it is also clear that we may perturb once more the sizes of the sides of $\mathcal{K}_i$ to ensure in addition that $m_i |\mathcal{K}_i| \in \mn$, with the previous conditions preserved. We may then apply Lemma \ref{lemcs2}  over $\mK_i$, this yields a $u_i$ satisfying
$$
\begin{cases}
-\Delta u_i  =c_d ( m_i-1)    \text{ in }  \mathcal{K}_i\\
\nab u_i \cdot \vec{\nu} =g_i  \text{ on }   I_i\\
\nab u_i \cdot \vec{\nu} = 0  \text{ on }  \partial \mathcal{K}_i \setminus I_i
\end{cases}
$$
and
$$
\int_{\mathcal{K}_i} |\nab u_i|^2 \le  C l \int_{I_i} g^2,
$$
where $C$  depends only on $d$.

We next split each $\mathcal{K}_i$ into hyperrectangles of sides $\in [1/2,3/2]$ and of volume $m_i ^{-1}$, which is possible since by construction $m_i |\mathcal{K}_i| \in \mn$. On each of these hyperrectangles we apply Lemma \ref{lemcs1}. Pasting together the $\nab u$'s given by that lemma  yields an $X_i$ defined over $\mathcal{K}_i$, such that
\begin{equation*}
\left\{\begin{array}{ll}
-\div X_i = c_d \left( \sum_k \delta_{x_k} -m_i\right)  & \text{in} \ \mathcal{K}_i\\
X_i\cdot \vec{\nu} =0 & \text{on} \ \p \mathcal{K}_i.
\end{array}\right.
\end{equation*}
Also, denoting $(X_i)_\eta = \Phi_\eta (X_i)$, it follows from the fact that $|m_i -1| < \hal$ that
\begin{equation}\label{nrjf2}
\int_{\mathcal{K}_i} |(X_i)_\eta|^2 \le C_\eta l^d.
\end{equation}
We note that by construction, all of these $x_k$'s are at distance at least $1/2$ of $\p \mathcal{K}_i$ hence, since $\eta \le 1/2$, the balls $B(x_k, \eta)$ do not intersect $\p \mathcal{K}_i$.

Defining now $Y_i=X_i+\nab u_i$ in $\mathcal{K}_i$ we have obtained  a solution of
\begin{equation*}
\left\{\begin{array}{ll}
-\div Y_i = c_d \left( \sum_k \delta_{x_k} -1\right)  & \text{in} \ \mathcal{K}_i\\
Y_i \cdot \vec{\nu} =g_i & \text{on} \ \p \mathcal{K}_i
\end{array}\right.
\end{equation*}
 and
\begin{equation}\label{nrjv2}
\int_{\mathcal{K}_i} |(Y_i)_\eta|^2\le 2 \int_{\mathcal{K}_i} |(\nab f_i)_\eta|^2+2\int_{\mathcal{K}_i} |\nab u_i|^2\le
C_\eta l^d  + C l \int_{\p \mathcal{K}_i} g_i^2.
\end{equation}
The new domain $\Gamma'$ is $K_{t+\alpha} \cup\cup_i \mathcal K_i$ and we may paste the $Y_i$ constructed above to obtain a $Y$ satisfying
\begin{equation*}
\left\{\begin{array}{ll}
-\div Y = c_d \left( \sum_k \delta_{x_k} -1\right)  & \text{in} \ \Gamma' \\
Y_i \cdot \vec{\nu} =g & \text{on} \ \partial K_{t+\alpha}\\
Y_i \cdot \vec{\nu} = 0 & \text{on} \ \partial \Gamma' \partial K_{t+\alpha}
\end{array}\right.
\end{equation*}

\noindent {\it - Step 4: completing the construction.} There remains to complete the construction in $K_R \setminus \Gamma' = K_R \backslash (K_{t+\alpha}\cup \cup_i \mathcal{K}_i)$. To this end, note that~\eqref{eq:pick m_i} implies 
$$ c_d \sum_i m_i |\mathcal{K}_i| - c_d \sum_i |\mathcal{K}_i| = \int_{\partial K_{t+\alpha}} g $$
whereas integrating the first equation of~\eqref{concleq} yields 
$$ \int_{\partial K_{t+\alpha}} g \in c_d \mn$$
because by construction $|K_{t+\alpha}| \in \mn$. Since we have also ensured $m_i|\mathcal{K}_i| \in \mn$ we deduce that $ \sum_i |\mathcal{K}_i| \in \mn$. Recalling~\eqref{eq:quant volume} we deduce that  
$$ \left| K_R \backslash (K_{t+\alpha}\cup \cup_i \mathcal{K}_i) \right| \in \mn.$$
We may thus tile  $K_R \backslash (K_{t+\alpha}\cup \cup_i \mathcal{K}_i)$ by hyperrecangles volume $1$ ands sidelengths $\in [1/2,3/2]$ on which we apply Lemma \ref{lemcs1} with $m=1$.  We paste together the  gradients of the functions  obtained this way. It gives a contribution to the energy proportional to the volume i.e.  of $C R^{d-1} L$.
We also paste this with the $Y_i$'s of the preceding step. It gives a final vector field  $\hj$ defined in $K_R \backslash K_{t+\alpha}$, satisfying
\begin{equation*}
\left\{\begin{array}{ll}
-\div \hj = c_d \left(\sum_k \delta_{x_k}  -1\right) & \text{in} \  K_R \backslash K_{t+\alpha}\\
\hj \cdot \vec{\nu}=g& \text{on} \ \p K_{t+\alpha}\\
\hj \cdot \vec{\nu} = 0 & \text{on} \ \p K_R.
\end{array}\right.
\end{equation*}
Here all the points $x_k$ are at distance $\ge \hal >\eta$ from $\p K_{t+\alpha}$. The energy of $\hj$ can now be controlled as follows: 
\begin{multline}\label{nrjh2}
\int_{K_R \backslash K_{t+\alpha}} |\hj_\eta|^2 \le C R^{d-1}L + C_\eta  R^{d-1} l^{1-d}   l^d + l \sum_i\int_{\p \mathcal{K}_i} g_i^2
=  C_\eta  R^{d-1} L+     l\int_{\p K_{t+\alpha}} g^2\\
\le C_\eta (R^{d-1}L + l R^d/L) .
\end{multline}
Combining with \eqref{concleq}, we have thus obtained an $\hj$  on all $K_R$ satisfying \eqref{eq:proscreen field}--\eqref{eq:proscreen flux}.
Using \eqref{franginter}, \eqref{nrjh2} and the conditions $l\ll L\ll R$ we have 
$$\int_{K_R} |\hj_\eta|^2 \le \int_{K_R} |\j_\eta |^2 + C_\eta( R^{d-1} L + lR^d/L) \le\int_{K_R} |\j_\eta|^2 +o(R^d),$$

\smallskip

\noindent
{\it - Step 5: making $\hj$ a gradient.} $\hj$ satisfies \eqref{eq:proscreen field} but is not necessarily a gradient. 
We claim that we may add  to it a vector-field $\mathcal{X}$ to make it a gradient, while still having \eqref{eq:proscreen field}--\eqref{eq:proscreen flux}, and decreasing the energy.  
It is standard by Hodge decomposition that we may find  $\mathcal{X}$ satisfying  $\div \mathcal{X}=0$ in $K_R$ and $\mathcal{X} \cdot \vec{\nu}=0$ on $\p K_R$, and such that $\bar{\j}:= \hj+ \mathcal{X}$ is a gradient.  Then $\bar{\j}$ still satisfies \eqref{eq:proscreen field}--\eqref{eq:proscreen flux}, and 
we note that 
$$\int_{K_R} |\hj_\eta|^2 =\int_{K_R} |\bar{\j}_\eta |^2 + \int_{K_R}|\mathcal{X}|^2 - 2\int_{K_R} \bar{\j}_\eta \cdot \mathcal{X}.$$
Since $\bar{\j}$ is a gradient and thus $\bar{\j}_\eta$ too,  and since $\div \mathcal{X}=0 $ and $\mathcal{X} \cdot \vec{\nu}=0$ on $\p K_R$, we find that the last integral vanishes, and so
$$\int_{K_R} |\bar{\j}_\eta|^2 \le \int_{K_R} |\hj_\eta|^2.$$
This last operation has thus not increased the total energy and the vector-field $\bar{\j}$ satisfies all the desired conclusions.
\end{proof}

An immediate consequence of the screening process is the following

\begin{pro}[\textbf{Periodic minimizing sequences}]\label{rem5.9}\mbox{}\\
Given  $\eta>0$ small enough, for any $R$ large enough and any  hyperrectangle $K_R$ of sidelengths in $[2R, 3R]$ such that $|K_R|\in \mn$,    there exists an $\bar{\j}$ satisfying \eqref{eq:proscreen field}, \eqref{eq:proscreen flux}, \eqref{eq:mindistr plus} and 
\begin{equation}\label{eqminw1}
\limsup_{R\to \infty} \dashint_{K_R} |\bar{\j}_\eta|^2 - (\kappa_d \g(\eta)+ \gamma_2 \indic_{d=2}) \le \inf_{\bai} \W_\eta\end{equation}
Also,  $\min_{\bai}\W_\eta$  admits a minimizing sequence made of periodic vector fields, and so does $\min_{\bai} \W$.
\end{pro}

The construction of periodic minimizing sequences proves the corresponding claim in Theorem \ref{thm:renorm ener}. The existence of a minimizer is a direct consequence of Theorem \ref{th1} as  discussed in Section \ref{sec7}. We also give  a self-contained proof in Appendix \ref{sec:appendix}.

\begin{proof}
It suffices to take $K_R= [-R,R]^d$ and $\nab h_\eta$ which approximates $F_{\eta, R}$ and apply Propositions  \ref{pro:sep points} and \ref{proscreen}. This yields the desired  $\bar{\j}$ satisfying \eqref{eq:proscreen field}, \eqref{eq:proscreen flux}, \eqref{eq:mindistr plus}, \eqref{eq:num points pre} and 
$$\limsup_{R\to \infty} \dashint_{K_R} |\bar{\j}_\eta|^2 - (\kappa_d \g(\eta)+ \gamma_2 \indic_{d=2}) \le  \limsup_{R\to \infty} \frac{F_{\eta, R}}{|K_R|} -  (\kappa_d \g(\eta)+ \gamma_2 \indic_{d=2})$$
and then  \eqref{eqminw1}  is an obvious consequence of the definitions of $\W_\eta$ and $F_{\eta, R}$.
In addition, $\bar{\j}$  is  the gradient of a function $h_R$ over $K_R$, with $\nab h_R \cdot \vec{\nu}=0$ on $\p K_R$.  We may then reflect $h_R$ with respect to $\p K_R$ to make it into a function over $[-R, 3R]^d$  which can then be extended periodically to the whole $\mr^d$. Its gradient then defines a periodic vector field $\mathcal Y_R$ over $\mr^d$, which belongs to $\mathcal A_1$.  By periodicity of $\mathcal Y_R$ and definition of $\W_\eta$, we obviously have $\W_\eta(\mathcal Y_R) =  \dashint_{K_R} |\bar{\j}_\eta|^2 - (\kappa_d \g(\eta)+ \gamma_2 \indic_{d=2}) $. In view of \eqref{eqminw1}, this implies that  we have 
$\limsup_{R\to \infty}\W_\eta(\mathcal Y_R) \le \inf_{\bai} \W_\eta$.  We thus conclude that $\inf_{\bai} \W_\eta$ 
admits a minimizing sequence made of periodic vector fields. Using a diagonal argument to deal with the $\eta \to 0$ limit, $\inf_{\bai} \W$ also does.
\end{proof}

\begin{proof}[Proof of Proposition \ref{pro:Wbb}]
In view of Proposition \ref{rem5.9}, to bound from below $\W_\eta$,  it suffices to bound from below $\dashint_{K_R}| \bar{\j}_\eta|^2$ with $\bar{\j}$ satisfying \eqref{eq:proscreen field}, \eqref{eq:proscreen flux}, \eqref{eq:mindistr plus} and \eqref{eq:num points pre}. But then, 
 we are in a situation where Lemmas \ref{lemequiv} applies. With \eqref{eq:num points pre}, the combination of these results easily yields 
$$\limsup_{R\to \infty} \dashint_{K_R} |\bar{\j}_\eta|^2 - (\kappa_d \g(\eta)+ \gamma_2 \indic_{d=2}) \ge - C $$ where $C$ depends only on $d$ and on $r_0$, which itself depends also only on $d$. 
In view of \eqref{eqminw1} this concludes the proof.
\end{proof}

\section{Upper bound to the ground state energy}\label{sec7}

In this section we complete the proof of Theorem \ref{th1} by proving the following

\begin{pro}[\textbf{Energy upper bound}]\label{promajo}\mbox{}\\
For any $\ep>0$ there exists $r_1>0$ and for any $n$ a set $A_n \subset (\mr^d)^n$ such that 
\begin{equation}\label{eq:volume min}
|A_n|\ge n! \left(\pi (r_1)^d/n\right)^n  
\end{equation}
and for any $(y_1, \dots, y_n) \in A_n$ we have
\begin{equation}\label{bsup 3d}
\limsup_{n\to \infty}  n^{2/d-2} \left(\w(y_1, \dots, y_n) - n^2 \I[\mu_0] \right) \le 
{\xi_d + \ep}  \quad \text{if } d\ge 3
%\frac{1}{c_d} \min_{\bai} \W \int \mu_0^{2-2/d}(x)\, dx+ \ep\quad \text{if } d\ge 3
\end{equation}
\begin{equation}\label{bsup 2d}
\limsup_{n\to \infty} n^{-1} \left(
\w(y_1, \dots, y_n) - n^2 \I[\mu_0] +  \frac{n}{2}\log n \right)\\ \le {\xi_2+\ep}\quad \text{if }  d=2.
% \left(\frac{1}{2\pi}\min_{\bai }\W -\hal \int \mu_0(x) \log \mu_0(x)\, dx \right) + \ep \quad \text{if }  d=2.
\end{equation}

\end{pro}

To deduce Theorem \ref{th1} it suffices to let $\ep \to 0$ in \eqref{bsup 3d} (respectively \eqref{bsup 2d}). This yields the upper bound corresponding to \eqref{eq:intro 222} (respectively \eqref{eq:intro 223}).  On the other hand, we have a lower bound for the ground state energy, given by Proposition \ref{promino}. Comparing the two yields the result of Theorem \ref{th1}, where  it also comes as a by-product  that $P$ minimizes $\W$, which can only happen if  $\j$ minimizes $\W$ over $\overline{\mathcal{A}}_{\mu_0(x)}$ for $P-$a.e. $(x,\j)\in X$. 
Note that this also proves the existence result in Theorem~\ref{thm:renorm ener}: for $P$ obtained in Theorem~\ref{th1}, we have $P$-a.e. $(x,\j)$ minimizes $\W$ in $\overline{\mathcal{A}}_{\mu_0(x)}$, hence  this implies the existence of a minimizer for $\W$ in some $\overline{\mathcal A}_m$, hence in all $\overline{\mathcal A}_m$ for all $m$ by \eqref{eq:scale renorm}. The difficult part is of course the existence of the probability measure $P$, which employs the ergodic framework of \cite[Theorem 7]{ss2d}. In Appendix \ref{sec:appendix}  we provide a more direct proof where this method is applied directly to the renormalized energy.

\begin{remark}\label{remzeta2}
In view of Remark \ref{remzeta}, comparing upper and lower bounds  also yields that 
for minimizers (or almost minimizers) of $\w$ we have $\sum_{i=1}^n\zeta (x_i)= o(n^{1-2/d}) $ as $n\to \infty$. Since $\zeta$ is expected to typically grow quadratically away   from $\E$, this provides a control on how far from $\E$ the points can be. In fact, arguing as in \cite{rns} one can show that for minimizers, there are no points outside $\E$. 
\end{remark}

The fact that the upper bound holds, up to a small error $\ep$, for any configuration in $A_n$, a set which has a reasonably large volume in configuration space, will be crucial when studying the fluctuations of the Gibbs measure at finite temperature. It would be very interesting to know more precisely the behavior of $\W$ around a minimizing configuration.

\medskip

The proof relies on an explicit construction, using the test charge configurations and electric fields of Corollary \ref{rem5.9}. Roughly speaking, we fix $\eta>0$ small and  split $\E$ into hyperrectangles of sidelengths of order $R n^{-1/d}$ for some $R$ that will ultimately tend to $\infty$. Equivalently we split $\E'$, the blown-up of $\Sigma$ at scale $n ^{1/d}$, into hyperrectangles of size $R$.  In each hyperrectangle $K$ we put a configuration of points constructed via Corollary~\ref{rem5.9}, properly scaled so that the local electric field is in $\overline{\mathcal{A}}_{m_K}$ where $m_K= \dashint_K \mu_0' $ is the mean value of $\mu_0'$ in the hyperrectangle $K$. Thanks to the screening property \eqref{eq:proscreen flux} the electric field is then defined globally in the domain by just gluing together the fields defined in each hyperrectangle. Since all the points in the test configuration so constructed are well-separated, we are in the case of equality of Lemma \ref{lemsplit} so the next to 
leading order in the energy upper bound will exactly 
be given by the electrostatic energy of charges smeared on a length $\eta n ^{-1/d}$ around the positions of our configuration, up to remainder terms that will ultimately become negligible. The conclusion in the limit $R\to \infty$ (tiling on a scale much larger than the interparticle distance) and $\eta \to 0$ (smearing charges on a scale much smaller than the interparticle distance) then follows by splitting this energy into the contribution of each hyperrectangle, using the scaling property \eqref{eq:scale renorm} and a Riemann sum argument to recover the $\mu_0$-dependent factors in the right-hand sides of \eqref{bsup 3d}, \eqref{bsup 2d}. Of course there is a boundary layer near $\p \E'$ that we cannot tile properly, and we will have to complete the construction there in a way that does not cost too much energy.

\begin{proof}
{\it - Step 1: we define the configuration in the interior of $\E'$ and estimate its energy.} 
We need to tile the interior of $\E'$ into hyperrectangles $K\in \cK_n$ and put in each a number of points equal to the charge of the background $\mu_0$. This requires that $\int_K  \mu_0'$ be an integer\footnote{Note that by scaling $\int_{\E'} \mu_0'= n \gg 1$.} for any $K$.
The following lemma provides such a tiling. It corresponds to \cite[Lemma 7.5]{ss2d} (with $q=1$) where the proof is provided in dimension 2. The adaptation to  higher dimension is immediate and left to the reader.

\begin{lem}[\textbf{Tiling the interior of $\E'$}]\label{rect}\mbox{}\\ 
There exists a constant $C_0>0$ such that, given any $R>1$,  there exists for any $n\in\mn^*$ a collection $\cK_n$ of closed hyperrectangles in $\E'$ with disjoint interiors,  whose sidelengths are between  $ 2R$ and $2R+C_0/R$,  and which are such that
\begin{equation}\label{eq:tiling 1}
\left\{x\in \E': d(x,\p \E')\ge C_0 R \right\} \subset  \bigcup_{K\in\cK_n} K:= \E'_{\rm{int}},
\end{equation}
\begin{equation}\label{eq:tiling 2}
\bigcup_{K\in\cK_n} K \subset \left\{ x\in \E': d(x,\p \E')\ge 2R \right\},
\end{equation}and 
\begin{equation}\label{entier}
\forall K\in\cK_n, \quad \int_K \mu_0' \in \mn.
\end{equation}
\end{lem}
Note that the sizes of the hyperrectangles are controlled  thanks to the positive lower bound assumption on the density $\mu_0'$.

We apply this lemma, which yields a collection $\cK_n$.
For each $K\in \cK_n$ we denote 
\begin{equation}\label{eq:mK}
m_{K}:= \dashint_{K} \mu_0'. 
\end{equation}
 We  need to correct for the difference between $m_{K}$ and $\mu'_0$ by setting
$u_{K}$  to be the solution to 
\[
\left\{\begin{array}{ll} 
-\Delta u= c_d  \left(m_{K}- \mu_0'\right) & \text{in} \ K\\
\nab u \cdot \vec{\nu}=0 & \text{ on} \  \p K.\\
\int_K u = 0 \end{array}\right. 
\]
By standard elliptic regularity we have
\begin{equation}\label{ellreg}
\left\|\nab u_{K}\right\|_{L^2(K)} \le C_{ R} \left\|m_{K}- \mu_0'\right\|_{L^\infty(K)}\le C_{R} n^{-1/d}.
\end{equation}
Indeed, $\mu_0$ is assumed to be $C^1$ on its support (this is where we use this part of \eqref{ass2}) hence after scaling $\left\|\nab \mu_0'\right\|_{L^\infty(\E')} \le C n^{-1/d}.$
Next we denote $\sigma_m$ the rescaling to scale $m$: 
\[
\sigma_m \j= m^{1- 1/d} \j(m^{1/d} . )
\]
and define $\j_K $ to be 
\begin{equation}\label{premdef}
\j_K = \nab u_{K} +\sigma_{m_K } \hj\quad \text{if} \ x\in K
\end{equation}
 where $\hj$ is provided  by Proposition~\ref{rem5.9} applied to $ m_K^{1/d}\eta$ over the hyperrectangle $m_K^{1/d} K$ (and suitably translated),  hence 
 satisfies
\begin{equation}\label{njhj}
\int_{m_K^{1/d} K} | \hj_{m_K^{1/d} \eta}|^2-|m_K^{1/d} K|  
(\kappa_d \g(m_K^{1/d}\eta) +\gamma_2\indic_{d=2}) \le |m_K^{1/d} K|( \min_{\bai} \W_{m_K^{1/d} \eta}+o_R(1)). 
\end{equation}
The vector field $\j_K$ then satisfies 
\begin{multline}\label{elreg2} 
\int_K |\left(\j_K\right)_\eta|^2\le \int_{K} \left|\left(\sigma_{m_K} \hj\right)_\eta\right|^2 + \left\|\nab u_K\right\|_{L^2(K)}^2 + 2 \left\|\nab u_K\right\|_{L^2(K)}\|(\sigma_{m_K}\hj)_\eta\|_{L^2(K)}.
\end{multline}
But by change of scales, we have 
$$\int_{K} \left|\left(\sigma_{m_K} \hj\right)_\eta\right|^2 = m_K^{1-2/d} \int_{m_K^{1/d}K} \left|\hj_{ m_K^{1/d}\eta}\right|^2 $$
so \eqref{njhj} gives that  this is bounded above by 
$$m_K^{1-2/d} \left( m_K |K|( \kappa_d\g(\eta m_K^{1/d})+\gamma_2\indic_{d=2}) +m_K |K| \min_{\bai} \W_{ m_K^{1/d}\eta}+o_R(1))\right) .
$$ Indeed since $\mu_0$ is bounded above and below on its support (see \eqref{ass2}), $m_K$ also is.

Combining with the above,  \eqref{ellreg} and \eqref{elreg2}, and using the definition of $\g$, \eqref{eq:scale renorm} and the Cauchy-Schwarz inequality we are led to  
\begin{equation}\label{elreg3} 
\int_K|(\j_K)_\eta|^2\le
 m_K |K| \kappa_d\g(\eta) +  |K| \Big( \min_{\overline{\mathcal A}_{m_K} } \W_{\eta }+o_R(1)\Big) + C_R n^{-2/d} + C_{R, \eta}  n^{-1/d}
\end{equation} 
if $d=3$;  respectively, if $d=2$,
\begin{multline}\label{elreg4} 
\int_K|(\j_K)_\eta|^2\le
m_K |K|\left(\kappa_2\g(\eta)+\gamma_2  -\hal\kappa_2  \log m_K\right)    +  |K|  \left(\min_{  \overline{\mathcal A}_{m_K}  } \W_{\eta}+ \kappa_2 \frac{m_K}{2} \log m_K+o_R(1)\right) \\
=    m_K |K|(\kappa_2 \g(\eta)+\gamma_2)    +  |K| \left( \min_{  \overline{\mathcal A}_{m_K}    }\W_{\eta} +o_R(1)\right) + C_R n^{-1} + C_{R, \eta}  n^{-1/2}
\end{multline} 
The electric field in $\E'_{\mathrm{int}}$ is then set to be 
$\j_{\mathrm{int}}= \sum_{K\in \cK_n}\j_K.$
We can extend it by $0$ outside of $\E'_{\mathrm{int}}$, and it then satisfies 
\begin{equation}
\label{divjjjj}
-\div \j_{\mathrm{int}} = c_d \Big( \sum_{p\in \Lambda_{\mathrm{int}}} \delta_p - \mu_0'\Big)  \quad \text{in} \ \mr^d
\end{equation}
for some discrete set $\Lambda_{\mathrm{int}}$. Indeed, no divergence is created at the interfaces between the hyperrectangles since the normal components of $\nab u_{K}$ and $\hj$ are zero. In view of \eqref{eq:mindistr plus} all the points in $\Lambda_{\rm{int}}$ are of  simple multiplicity and  at a distance $>\frac{r_0}{10 \|\mu_0\|_{L^\infty}} $ from all the others. Note that integrating \eqref{divjjjj} we have 
$$\# \Lambda_{\mathrm{int}}= \int_{\E'_{\mathrm{int}}} \mu_0',$$ 
which is an integer.

\medskip
\noindent
{\it  - Step 2: we define the configuration near the boundary}. Since $\p \E'\in C^1$,  the set 
$$\E'_{\mathrm{bound}}:=\E'\backslash \E'_{\mathrm{int}}$$
is a strip near $\p \E'$ of volume $\le C n^{\frac{d-1}{d}}$ and width $\ge cR$ by  Lemma \ref{rect}. Since $\int_{\E'}\mu_0'= n$, $\int_{\E'_{\mathrm{bound}}} \mu_0'$ is also an integer.
We just need to place $\int_{\E'_{\mathrm{bound}}} \mu_0'$  points in $ \E'_{\mathrm{bound}}$, all separated by distances bounded below by some constant $r_0>0$ independent of $n$, $\eta$, and $R$  (up to changing $r_0$ if necessary).  Proceeding as in \cite[Section 7.3, Step 4]{ss2d}, using the fact that $\p \E'$  is $C^1$, we may split $\E'_{\mathrm{bound}}$ into regions $\mathcal{C}_i$ such that $\int_{\mathcal{C}_i}\mu_0'\in \mn$ and  $\mathcal{C}_i$ is a set with piecewise $C^1$ boundary, containing a ball of radius $C_1 R$  and contained in a ball $B_i$ of radius $C_2R$, where $C_1, C_2>0$ are universal. We then place $\int_{\mathcal{C}_i} \mu_0'$ points in $\mathcal{C}_i$, in such a way that their distances (and their distance to $\p \mathcal{C}_i$)  remain bounded below by $r_0>0$, and call $\Lambda_i$ the resulting set of points.
We then define $v_i$ to solve 
\begin{equation}
\label{divjjjjbound}
\left\{\begin{array}{ll}
-\Delta v_i = c_d \left( \sum_{p\in \Lambda_{i} } \delta_p - \mu_0'\indic_{\mathcal{C}_i} \right)  &  \quad \text{in} \ B_i\\
\nab v_i \cdot \vec {\nu}= 0 & \text{on} \ \p B_i,\end{array}\right.\end{equation}  and extend $\nab v_i$ by $0$ outside  $B_i$, this way, we have globally
$$-\div (\nab v_i)=  c_d \Big( \sum_{p\in \Lambda_{i} } \delta_p - \mu_0'\indic_{\mathcal{C}_i}\Big)\quad \text{in } \mr^d.$$
We then  set
$$\j_{\mathrm{bound}}:= \sum_i \nab v_i.$$

We can also  check that, arguing as in \cite[Section 7.3]{ss2d}, the energy of $\j_{\mathrm{bound}}$ is bounded by a constant depending on the number of points involved, times the volume of the boundary strip, that is
\begin{equation}\label{energybord}
\sum_i \int_{B_i} \left|\left( \j_{\mathrm{bound}}\right)_\eta \right| ^2 \le  C_{R,\eta} n^{1-1/d}.
\end{equation}

\noindent
{\it  - Step 3: we define the configuration globally and evaluate the energy}. We set 
$$\j =\j_{\mathrm{bound}}+ \j_{\mathrm{int}} \mbox{ in } \E' $$   
 and extend it by $0$ outside $\E'$,  and  we let 
$$\Lambda= \Lambda_{\mathrm{int}} \cup \cup_i \Lambda_i.$$ 
Then $\j$ satisfies
\begin{equation}\label{divjglob}
-\div \j = c_d \Big( \sum_{p \in \Lambda} \delta_p- \mu_0'\Big) \quad \text{in} \ \mr^d.
\end{equation} 
We also have $\# \Lambda= \int_{\E'} \mu_0'= n$ and we can define thus our test configuration as 
\begin{equation}\label{eq:test config}
\xbf = \{x_i= n^{-1/d} x_i'\}_{i=1}^n \mbox{ where } \Lambda = (x_1', \dots, x_n').
\end{equation}
Note that it depends on $R $ and $\eta$. 
There remains to bound $\w(x_1,\dots, x_n)$ from above. Since all the points in $\Lambda$ are separated by distances $>r_0$ fixed, the points in the collection $(x_1,\ldots,x_n)$ are separated by distances $>r_0 n ^{-1/d}$, where $r_0$ is independent of $\eta$. We may then  apply Lemma \ref{lemsplit} with $\ell = \eta n ^{-1/d}$, and if $\eta$ is small enough, we have $|x_i- x_j| \ge 2\ell$. We are then in the case of equality  in that lemma:  
\begin{equation}\label{eq:up pred}
\w(x_1,\dots,x_n)-n^2 \I[\mu_0] + \left(\frac{n}{2}\log n\right) \indic_{d=2} \le n^{2-2/d}\left(J_n(x_1, \dots, x_n) + C \eta^2\right),
\end{equation}
where we used that all the points are in $\E$ where the function $\zeta$ vanishes, and by definition of $J_n$,  and letting $h_{n,\eta}'$ be as in \eqref{defh},
\begin{multline}\label{eq:up F}
J_n (x_1, \dots, x_n) 
: = \frac{1}{c_d}\left(\frac{1}{n}\int_{\mr^d} |\nab \het|^2 -( \kappa_d \g(\eta)+\gamma_2\indic_{d=2})\right)
\\ \leq \frac{1}{c_d}\left(\frac{1}{n}\int_{\mr^d} |\j_\eta|^2 - (\kappa_d \g(\eta)+\gamma_2\indic_{d=2})\right).
\end{multline}
The last inequality is a consequence of \eqref{divjglob}:
\begin{align*}
\int_{\mr^d} |\j_\eta|^2 &= \int_{\mr^d} |\nab \het|^2 + \int_{\mr^d} |\j_\eta - \nab \het|^2 + 2 \int_{\R ^d} \nab \het \cdot \left( \j_\eta - \nab \het \right)\\
&= \int_{\mr^d} |\nab \het|^2 + \int_{\mr^d} |\j_\eta - \nab \het|^2 - 2 \int_{\R ^d} \het \,\div\left( \j_\eta - \nab \het \right)\\
&= \int_{\mr^d} |\nab \het|^2 + \int_{\mr^d} |\j_\eta - \nab \het|^2
\end{align*}
where the integration by parts is justified by the decay at infinity \footnote{Remark that the right-hand side of \eqref{defh} always has zero total charge and is compactly supported.} of $\j_\eta$ and $\nab \het$, and we have $\div (\j_\eta - \nab \het)=0$ by definition \eqref{hh}. Next we recall that by construction of  $\j$ we may write \begin{equation}\label{eq:rid boundary}
 \int_{\mr^d} |\j_\eta|^2 \leq \sum_{K \in \mathcal K_n} \int_{K} \left|(\j_K)_\eta\right|^2 + C_R n ^{1-1/d}
\end{equation}
where we used \eqref{energybord} to estimate the contribution of the boundary terms, and  summing  the bounds \eqref{elreg3}--\eqref{elreg4} over $K$ and inserting into \eqref{eq:rid boundary}, we obtain
\begin{multline}\label{eq:up presque 3d}
c_d n J_n (x_1, \dots, x_n) \\
\leq ( \kappa_d \g(\eta) + \gamma_2 \indic_{d=2})  \Big( \sum_{K\in \mathcal K_n} m_K |K|- n \Big) + \sum_{K\in \mathcal K_n} |K|  \Big(\min_{\overline{\mathcal A}_{m_K} } \W_{\eta} +o_R(1) \Big)  + C_R n ^{1-1/d}.
\end{multline}
First we note that from \eqref{eq:mK} 
$$\sum_{K\in \mathcal K_n} m_K |K| - n= \int_{\E'_{\mathrm{int}}  }\mu_0' - n = - \int_{\E'_{\mathrm{bound}} } \mu_0'= o(n).$$
 In view of the regularity of $\mu_0$ (which implies $\| m_K-\mu_0'\|_{L^\infty(K)}\le C_R n^{-1/d}$) and the properties of our tiling $\cK_n$, with \eqref{eq:mK} again, we can then also recognize a Riemann sum to see that 
$$ \sum_{K\in \cK_n} |K|\min_{\overline{\mathcal A}_{m_K} } \W_{\eta} \le \int_{\E'} \min_{\overline{\mathcal A}_{\mu_0'(x)} } \W_\eta\, dx + o_R(n),$$
using the continuity of $m\mapsto \min_{\overline{\mathcal{A}}_m}  \W_\eta$ which can be checked from \eqref{We}, \eqref{eq:scale renorm} and the continuity of $\eta\mapsto \Phi_\eta$. 
The proof is concluded by dividing \eqref{eq:up presque 3d}  by $c_d n$, using \eqref{eq:up pred}, and taking successively $n\to \infty$, $R$ large enough and then $\eta$ small enough (and changing the configuration of points accordingly). Using that 
\[
\min_{\overline{\mathcal A}_m} \W_{\eta} \to  \min_{\overline{\mathcal A}_m} \W
\] 
along a sequence  $\eta \to 0$, by definition, we can make 
\begin{equation}\label{enx}
\w(x_1, \dots, x_n) -n^2  \I[\mu_0] + \left(\frac{n}{2}\log n\right) \indic_{d=2} \le n^{2-2/d}\Big(\frac{1}{c_d} \int_{\E} \min_{\overline{\mathcal A}_{\mu_0(x)} }  \W \, dx+\hal \ep\Big).\end{equation}
Finally, inserting \eqref{eq:scale renorm} gives the upper bound result at the configuration $(x_1, \dots, x_n)$. 

There remains to perturb the points and prove the statement about the volume of the set $A_n$. 
We note that (still for the $R$ large enough and $\eta$ small enough that \eqref{enx} holds), the construction of $(x_1', \dots, x_n')$ has been made 
(at the blown up scale) by gluing together configurations over an order $n/R^d$ of hyperrectangles $K\in \mathcal{K}_n$ (or cells near the boundary $\mathcal C_i$) of volumes $O(R^d)$, each containing a number of points $O(R^d)$ which are at distances bounded below by $r_0>0$. If each of these points is moved by $r_1<r_0/2$ small enough, by continuity of the energy in each box\footnote{which follows from the techniques of Section~\ref{sec:univ low bound}.}, this induces an error  in $\int |\j_\eta|^2 $ in each hyperrectangle which can be made $<\frac{\ep}{2CR^d}$ if $r_1$ is chosen small enough (depending on $\ep$, $R$, $\eta$), for example in each estimate \eqref{elreg3}. Choosing $C$ large enough depending on $R$, and summing these errors over all the  hyperrectangles leads to an error $\hal n\ep$, which, once divided by $n$, yields an error term $\frac{\ep}{2}$ in \eqref{eq:up F}. Continuing on leads to an error $\ep$ instead of $\ep/2$ in \eqref{enx}. We have thus shown that we have the desired estimate 
in the set $A_n$ consisting of those $(y_1, \dots, y_n)$ such that   $\sup_i (y_i'-x_i')<r_1$ 
(equivalent to $\sup_i (y_i-x_i) <  r_1 n^{ -1/d}$)  for $r_1$ is small enough. 
Clearly this set has volume $n! (r_1^d/n)^n$ in configuration space: the $r_1 ^d /n$ term is the volume of the ball $B(x_i,r_1 n ^{-1/d})$, it is raised to the power $n$ because there are $n$ points in the configuration, and multiplied by $n !$ because permuting $y_1,\ldots,y_n$ does not change the energy.
\end{proof}
%
%The only missing ingredient in the above proof is given by

%\begin{lem}[\textbf{Small perturbations of the construction}]\label{lem:volume A_n}\mbox{}\\%
%Let $\xbf=(x_1,\ldots,x_n)$ and $\ybf=(y_1,\ldots,y_n)$ be two configurations of points satisfying, for some positive constants $r_0,r_1$,
%\begin{enumerate}
%\item $\inf_{i\neq j} \dist (x_i,x_j) \geq r_0 n ^{-1/d}$
%\item $\inf_i \dist (x_i,y_i) \geq r_1 n ^{-1/d}$.
%\end{enumerate}
%Let $h'_{n,\eta,\xbf}$ and $h'_{n,\eta,\ybf}$ be the blown-up potentials generated (according to \eqref{defh}) by $\xbf$ and $\ybf$ respectively. Then, for $r_1$ small enough (depending only on $r_0$)  
%\begin{equation}\label{eq:up bound perturb}
%\frac{1}{n} \left| \int_{\mr^d} |\nab h'_{n,\eta,\xbf}|^2 - \int_{\mr^d} |\nab h'_{n,\eta,\ybf}|^2\right| \leq C (r_1)
%\end{equation}
%with $C(r_1) \to 0$ when $r_1 \to 0$.
%\end{lem}
%
\section{Applications to the partition function and large deviations}\label{sec8}

\subsection{Estimates on the partition function: proof of Theorem \ref{thm:partition}}\label{sec:partition}

In this section we prove Theorem \ref{thm:partition}. 

\subsubsection*{Low temperature regime}

The result is proved by finding upper and lower bounds to $\Fnbetae$, which are themselves direct consequences of the upper and lower bounds we have obtained for $\w$.
%\begin{equation}\label{eq:gamma d 2}
 %\gamma_d := \begin{cases}\displaystyle
  %           \frac{1}{c_d} \left(\min_{\bai}  \W \right) \int_{\R   ^d} \mu_0 ^{2-2/d} \mbox{ if } d\geq 3\\
   %          \displaystyle \frac{1}{2\pi} \left(\min_{\bai}  \W - \int_{\R   ^2} \mu_0 \log \mu_0 \right) \mbox{ if } d=2.
   %         \end{cases}
%\end{equation}
We  recall that $\Fnbetae$ is linked to the partition function via \eqref{eq:free ener N min}, so that
\begin{equation}\label{eq:free ener partition}
\Fnbetae = - \frac{2}{\beta} \log \left(\int_{\R ^{dn}} e ^{-\frac{\beta}{2} H_n } \right). 
\end{equation}
%We treat the case  $d\geq 3$, the case $d=2$ can be treated  exactly in the same way (and is  already dealt with in \cite{ss2d}). 

\medskip

\noindent\textit{Lower bound.} 
Proposition \ref{promino}  and Remark \ref{remzeta} yield
  that for any $(x_1,\ldots,x_n) \in \R ^{dn}$  
\[
 H_n (x_1,\ldots,x_n) \geq n ^2 \En [\mu_0] - \left(\frac{n}{2}\log n \right) \indic_{d=2}+  n ^{2-2/d} \xi_d + 2n \sum_{i=1} ^n \z (x_i) + o_n (n ^{2-2/d})
\]
where $o_n (1) \to 0$ when $n\to \infty$, and $\xi_d$ is defined in  \eqref{eq:gamma d}. Consequently,
\[
 e^{-\frac{\beta}{2}H_n (x_1,\ldots,x_n)} \leq \exp\left( - \frac{n ^2 \beta}{2} \En [\mu_0] +\left(\frac{\beta n}{4}\log n \right) \indic_{d=2} - \frac{n ^{2-2/d}\beta}{2} \left( \xi_d + o_n(1)\right) \right) \prod_{i=1} ^n \exp\left( -n \beta \z (x_i) \right)
\]
and we deduce from \eqref{eq:free ener partition}, separating variables when integrating $\prod_{i=1} ^n \exp\left( -n \beta \z (x_i) \right)$ over $(\R ^{d})^n$,  that
\[
 \Fnbetae \geq n ^2 \En [\mu_0] - \left(\frac{n}{2}\log n \right) \indic_{d=2} + n ^{2-2/d} \left( \xi_d + o_n(1)\right) -\frac{2n}{\beta} \log \left( \int_{\R ^d} \exp(-n \beta \z (x)) \, dx\right). 
\]
On the other hand, with the assumption \eqref{integrabilite}, the dominated convergence theorem yields that 
{
\begin{equation}\label{eq:integr zeta}
 \int_{\R ^d} \exp(-n \beta \z (x)) \, dx\to |\{\zeta=0\}| = C \qquad \text{as} \  n\beta \to +\infty
\end{equation}}
where $C$ is a constant. Note that the assumptions on $\beta$ in Theorem \ref{thm:partition}, Item 1, ensure that $\beta n\to \infty$ as $n \to \infty$. The lower bound corresponding to \eqref{eq:partition low T} --\eqref{eq:partition low T 2d} then follows. 

\medskip

\noindent\textit{Upper bound.} The key tool for the upper bound is Proposition \ref{promajo}. For any $\ep$ we have an $r_1$ (depending on $\ep$) and a set $A_n$ as described therein and we may write
\[
\Fnbetae \leq - \frac{2}{\beta} \log \left(\int_{A_n} e ^{-\frac{\beta}{2} H_n } \right). 
\]
This corresponds to taking as a trial state for the free energy functional $\Fnbeta$ the probability measure  
$ \frac{    \one_{A_n} \Q  }{\Q(A_n)} ,$
where $\one_{A_n}$ is the characteristic function of the set $A_n$. Using \eqref{bsup 3d}--\eqref{bsup 2d} we deduce 
\[
 \Fnbetae \leq - \frac{2}{\beta} \log |A_n| + n ^2 \En [\mu_0]  - \left(\frac{n}{2}\log n \right) \indic_{d=2} + n ^{2-2/d} \left(\xi_d +\ep \right). 
\]
Using the estimate on $|A_n|$ and Stirling's formula we have 
\[
\log |A_n| \geq n \log \frac{\pi r_1 ^d }{e} -C  
\]
and thus 
\begin{equation}\label{eq:up bound final}
\Fnbetae \leq   n ^2 \En [\mu_0] - \left(\frac{n}{2}\log n \right) \indic_{d=2}+ n ^{2-2/d} \left(\xi_d +\ep \right) - \frac{2n}{\beta} \log \frac{\pi r_1 ^d }{e} + \frac{C}{\beta}.  
\end{equation}
The upper bound corresponding to \eqref{eq:partition low T}--\eqref{eq:partition low T 2d} follows, with $C_\ep$ depending only on $\log r_1$ and thus bounded when $\ep $ is bounded away from $0$.

\subsubsection*{High temperature regime}

In this regime we will use the mean-field density at positive temperature $\mubet$. We recall some of its properties in a lemma. The proof uses arguments already provided in \cite[Section 3.2]{RSY2} and is omitted.

\begin{lem}[\textbf{The mean-field density at positive temperature}]\label{lem:MF positive T}\mbox{}\\
For any $\beta > 0$ the functional \eqref{eq:MF free ener func} admits a unique minimizer $\mubet$ among probability measures. It satisfies the bounds 
\begin{equation}\label{eq:unif bound mubeta}
0 < \mubet \leq C  \mbox{ on } \R ^d  
\end{equation}
where $C$ is some constant depending only on the dimension and the potential $V$. Moreover we have the variational equation
\begin{equation}\label{eq:MF equation pos T}
2 h_{\mubet} + V +  \frac{2}{n\beta} \log (\mubet) = \F[\mubet] +  D(\mubet,\mubet) \mbox{ on } \R ^d  
\end{equation}
where 
\begin{equation}\label{eq:defi U mubet}
h_{\mubet} = \g * \mubet. 
\end{equation}
\end{lem}

An upper bound to the $n$-body free energy is obtained by taking the trial state $\mubet ^{\otimes n}$ in \eqref{eq:free ener N}. Independently of the dimension this yields
\begin{equation}\label{eq:up bound high T}
\Fnbetae \leq \Fnbeta [\mubet ^{\otimes n}] = n ^2  \F [\mubet]  - n D(\mubet,\mubet) \leq n ^2  \F [\mubet]  - C n 
\end{equation}
where the second term is due to the fact that there are $n(n-1)/2$ and not $n^2 /2$ pairs of particules.  The upper bound in \eqref{eq:partition high T},  \eqref{eq:partition high T 2d} then follows from the definition \eqref{eq:MF free ener func}.

\medskip

For the lower bound we use a variant of Lemma \ref{lemsplit}. 

\begin{lem}[\textbf{Alternative splitting lower bound}]\label{lemsplit high T}\mbox{}\\
For  any $x_1,\ldots,x_n \in \mr^d $, we have
\begin{equation}\label{eq:split Onsager high T}
\w(x_1,\dots,x_n) \ge n^2 \F [\mubet] - \frac{2}{\beta} \sum_{i = 1 ^n} \log \mubet(x_i) - \left(\frac{n}{2} \log n \right) \indic_{d=2} - Cn^{2-2/d},
\end{equation} where $C$ only depends on the dimension.
\end{lem}

\begin{proof}
We only sketch the proof since it follows exactly that of Lemma \ref{lemsplit}. Applying Onsager's Lemma \ref{lem:Onsager}  with $\ell=n ^{-1/d}$  and $\mu= n\mubet$, we find
\begin{align*}
\sum_{i\neq j} \g( x_i-x_j) & \geq D\left( n\mubet  - \sum_{i=1}^n\delta_{x_i}^{(\ell)} ,n\mubet  -   \sum_{i=1}^n\delta_{x_i}^{(\ell)}   \right) - n^2 D(\mubet ,\mubet) + 2 n \sum_{i=1}^n D(\mubet, \delta_{x_i}^{(\ell)}    )\\
& - n  D(\delta_{0}^{(\ell)} ,\delta_{0}^{(\ell)} ) \\
&\geq - n^2 D(\mubet ,\mubet) + 2 n \sum_{i=1}^n D(\mubet, \delta_{x_i}^{(\ell)}    ) - n \(\frac{\kappa_d}{c_d} \g(\ell)  + \frac{\gamma_2}{c_2}\indic_{d=2} \)     \\
&\geq  - n^2 D(\mubet ,\mubet) + 2 n \sum_{i=1}^n h_{\mubet} (x_i) - n
 \(\frac{\kappa_d}{c_d} \g(\ell)  + \frac{\gamma_2}{c_2}\indic_{d=2} \) 
 - C  n ^2 \ell ^2 \\
&\geq  n ^2 \F [\mubet] - n \left(\sum_{i=1}^n V (x_i) + \frac{2}{\beta n} \sum_{i=1} ^n \log (\mubet (x_i))\right) - n  \frac{\kappa_d}{c_d} \g(n^{-1/d})  
- C n ^{2-2/d}.
\end{align*}
The second inequality is obtained by dropping a positive term and using \eqref{chvard}, the third follows from Lemma \ref{lem:Lieb} and \eqref{eq:unif bound mubeta} as in \eqref{eq:split pot terms} and the fourth from the variational equation \eqref{eq:MF equation pos T} and computing $\g(n^{-1/d}) $ explicitly. Plugging this inequality in the expression of the Hamiltonian, we get the result. 
\end{proof}

Using the results of Lemma \ref{lemsplit high T} and the expression \eqref{eq:free ener N} we deduce that for any probability measure $\mu$ 
\begin{equation}\label{eq:low bound partition high T}
\Fnbeta [\mu] \geq n^2 \F [\mubet] + \frac{2}{\beta} \int_{\R ^{dn}} \mu  \: \log \Big(\frac{\mu}{\mubet ^{\otimes n}}\Big) -
\left(\frac{n}{2} \log n \right) \indic_{d=2} - Cn^{2-2/d}. \end{equation}
The integral term in this equation is (minus) the entropy of $\mu$ relative to $\mubet ^{\otimes n}$ and is thus positive (see e.g. \cite[Lemma 3.1]{RSY2}) for any probability measure $\mu \in \P (\R ^{dn})$. Dropping this term, combining this lower bound with the upper bound \eqref{eq:up bound high T}, and noting that $n ^{2-2/d} \gg n$ for $n$ large and $d\geq 2$, we conclude the proof of Theorem \ref{thm:partition}, Item 2.

\begin{remark}[\textbf{Total variation estimates for reduced densities}]\label{rem:TV estimates}\mbox{}\\
Instead of simply dropping the relative entropy in \eqref{eq:low bound partition high T} for our final estimate of the $n$-body free-energy we may combine our upper and lower bounds to control this term. Interesting estimates then follow from the coercivity properties of the relative entropy.  Indeed, using subadditivity of entropy (see e.g. \cite[Proposition 2]{Kie2})
\begin{equation}\label{eq:subadd 1}
\int_{\R ^{dn}} \mubf \log \mubf \geq \left\lfloor\frac{n}{k}\right\rfloor \int_{\R ^{dk}} \mubf ^{(k)} \log \mubf ^{(k)} + \int_{\R ^{dn[k]}} \mubf ^{(n[k])} \log \mubf ^{(n[k])} 
\end{equation}
where $\mubf$ is a symmetric probability on $\R ^{dn}$, $\mubf ^{(k)}$ is its $k$-th marginal, $\left\lfloor\: . \:\right\rfloor$ stands for the integer part and $n[k]$ is $n$ modulo $k$. On the other hand 
\begin{equation}\label{eq:subadd 2}
\int_{\R ^{dn}} \mubf \log \left(\mubet  ^{\otimes n} \right)= \left\lfloor\frac{n}{k}\right\rfloor \int_{\R ^{dk}} \mubf ^{(k)} \log \left(\mubet ^{\otimes k}\right) + \int_{\R ^{dn[k]}} \mubf ^{(n[k])} \log \mubet ^{\otimes n[k]}.
\end{equation}
By positivity of relative entropies the contribution of the difference of the second terms in equations \eqref{eq:subadd 1} and \eqref{eq:subadd 2} is non negative. One can then use the Csisz\'{a}r-Kullback-Pinsker inequality (see e.g. \cite[Lemma 3.1]{RSY2} and \cite{BV} for a proof) to bound from below the difference of the first terms and obtain 
\[
\int_{\R ^{dn}} \mubf \log \frac{\mubf}{\mubet ^{\otimes n}} \geq \left\lfloor\frac{n}{k}\right\rfloor \int_{\R ^{dk}} \mubf ^{(k)} \log \frac{\mubf ^{(k)}}{\mubet ^{\otimes k}} \geq  \frac12 \left\lfloor\frac{n}{k}\right\rfloor \left\Vert \mubf ^{(k)} - \mubet ^{\otimes k} \right\Vert_{\rm TV} ^2
\]
where $\left\Vert\, . \, \right\Vert _{\rm TV}$ stands for the total variation norm. A control on the $n$-body relative entropy thus provides estimates in the spirit of Corollary \ref{thm:marginals}, but it turns out (combining our free energy upper and lower bounds) that those are meaningful only in the high temperature regime. See \cite{RSY2} (in particular Remark 3.4 and Section 3.4) where this method is used in such a regime. 
\end{remark}

\subsection{The Gibbs measure at low temperature: proof of Theorem \ref{th4}}\label{sec:deviations}
Let $A_n$ be any event, i.e. subset of $(\mr^d)^n$.
Using   \eqref{eq:free ener N min}  we may write
\begin{multline}
\label{pan}
\Q (A_n) =  \frac{1}{\Z}\int_{A_n} \exp \left( - \frac{\beta}{2} H_n (x_1,\ldots,x_n) \right) \, dx_1\ldots dx_n\\
= \int_{A_n} \exp \left( \frac{\beta}{2} \left( \Fnbetae - H_n (x_1,\ldots,x_n) \right) \right) \, dx_1\ldots dx_n .\end{multline}
But the result of Proposition \ref{promino} and Remark \ref{remzeta} give us that  for any $\xbf_n \in A_n$, 
$$\liminf_{n\to \infty} n^{2/d-2}\( \w(\xbf_n) - n^2 \En(\mu_0)+ \left(\frac{n}{2}\log n\right) 
\indic_{d=2} \right) \ge \frac{|\E|}{c_d} \int \W(\j) \, dP(x, \j)$$
where $P=\lim_n P_{\nu_n}$ and $P_{\nu_n}$ is defined as in \eqref{eq:Pnun}. Since $P_{\nu_n}\in i_n(A_n)$ we have $P\in A_{\infty}$, where $A_\infty$ is as in the statement of the theorem,  by definition, and thus   
$$\liminf_{n\to \infty} n^{2/d-2}\( \w(\xbf_n) - n^2 \En(\mu_0)+ \left(\frac{n}{2}\log n\right) 
\indic_{d=2} \right) \ge \frac{|\E|}{c_d} \inf_{P \in A_\infty}\int \W(\j) \, dP(x, \j).$$
 Inserting this into \eqref{pan} and using \eqref{eq:partition low T}--\eqref{eq:partition low T 2d}, we find, for every $\ep>0$
\begin{multline*}
\limsup_{n\to \infty} \frac{\log \Q (A_n)}{n ^{2-2/d}} \leq - \frac{\beta}{2}\left( \frac{|\E|}{c_d} \inf_{P\in A_{\infty}} \int \W(\j) dP(x,\j) - (\xi_d+\ep)  - C_\ep \lim_{n\to \infty} \frac{n ^{2/d-1}}{\beta} \right) \\
+ \limsup_{n\to \infty} \frac{1}{n^{2-2/d}} \log \int_{A_n} \exp\left(-\beta \sum_{i=1} ^n n \zeta(x_i)\right)dx_1\ldots dx_n
\end{multline*}
and there only remains to treat the term involving $\zeta$ as before using \eqref{eq:integr zeta} to arrive at the desired result \eqref{ldr}.

We next turn to the tightness of $\widetilde{\Q}$.
Starting from \eqref{pan} and inserting the   upper bound on $\Fnbetae$ implied by \eqref{eq:partition low T} and the lower bound \eqref{pred2}, we obtain  that for any event $A_n$, for any $\eta\le 1$,
\begin{multline}\label{fhn1}
\Q (A_n) \leq \int_{A_n} \exp \left( - \frac{\beta}{2} n ^{2-2/d} \left( \frac{1}{c_d n} \int_{\R ^d} |\nabla h'_{n,\eta}|^2  -  (\kappa_d \g(\eta)+\gamma_2 \indic_{d=2} )-C\eta^2  \right)\right)\\
 \prod_{i=1} ^n \exp(-n\beta \z(x_i)) dx_1 \dots dx_n. 
\end{multline}
 Let now
$$A_{n,M}:= \left\{\xbf \in (\mr^d)^n, \forall \eta<\hal , \frac{1}{n} \int_{\mr^d} |\nab \het|^2- (\kappa_d \g(\eta)+\gamma_2 \indic_{d=2} )\le M\right\}.$$
Inserting into \eqref{fhn1} and using \eqref{eq:integr zeta}, we obtain 
$$\Q((A_{n,M})^c)  
\le  \exp \(Cn +\frac{\beta}{2}(Cn^{2-2/d}   - \frac{M}{c_d} n^{2-2/d})\right).$$
Using that $\beta \ge cn^{2/d-1}$, it follows that we can find a large enough  $M$ (independent of $n$ and $\beta$) for which $\Q((A_{n,M})^c) \to 0$ as $n\to \infty$. To prove the tightness of $\widetilde{\Q}$, in view of \cite[Lemma 6.1]{ss2d}, it  then suffices to show that  if $P_n \in i_n(A_{n,M})$ then $P_n$ has a convergent subsequence. But this has been  precisely established  in Section \ref{sec:promino}. The fact that the limits of $\widetilde{\Q}$ are concentrated on admissible $P$'s satisfying $\widetilde{\W} (P) \le \xi_d + C_{\bar{\beta}}$ is an easy consequence of what precedes.

\subsection{Charge fluctuations: proof of Theorem \ref{thm:charge fluctu}}\label{sec:fluctu charge}
We start  from \eqref{fhn1} applied with some arbitrary $0<\eta<1$ and insert the result of  Lemma \ref{lem:fluctu charge}.  We deduce that for any event $A_n$, 
\begin{multline*}
\Q (A_n) \leq\int_{A_n} \exp \left( - C \beta n ^{1-2/d} \left( \frac{D(x',R) ^2}{R^{d-2}}\min \left(1,\frac{D(x',R)}{R ^d}\right)  - C n \right)\right) \\\prod_{i=1} ^n \exp(-n\beta \z(x_i))dx_1 \dots dx_n, 
\end{multline*}where we recall that $D(x',R)$ depends on $(x_1,\ldots,x_n) \in A_n$ via the ``empirical measure" $\nu_n'$.
Then, using \eqref{eq:integr zeta}, we find
\begin{equation}\label{eq:charge fluctu main}
\Q (A_n) \leq C \sup_{A_n} \exp \left( - C \beta n ^{1-2/d} \left( \frac{D(x',R) ^2}{R^{d-2}}\min \left(1,\frac{D(x',R)}{R ^d}\right)  - C n \right)\right).
\end{equation}
 Equation \eqref{eq:charge fluctu main} is our main bound on the charge fluctuations. It implies a global control similar to \cite[Eq. (1.49)]{ss2d} on the fluctuations which is in $L^2$ for large fluctuations and in $L ^3$ for small fluctuations. Here we only prove explicitely the statements of Theorem \ref{thm:charge fluctu}, recalling that estimates at intermediate scales follow from \eqref{eq:charge fluctu main} in the same way.
\medskip

\noindent \emph{Proof of Item $1$.} We pick a sequence of microscopic balls of radii $R_n$.   According to \eqref{eq:fluctu charge micro macro} and using \eqref{eq:charge fluctu main}, we have
\begin{align*}
\Q \left( |D(x,R_n)| \geq \lambda n R_n ^{d}\right) &\leq C \exp\left( -C \beta n ^{1-2/d} \left(  \lambda^2 n^2 R_n^{d+2}\min (1, \lambda n)
 - C n \right) \right)\\ 
&\le  C \exp\left( -C \beta n ^{2-2/d} ( C_R \lambda^2- C ) \right),   
\end{align*} if $\lambda \ge 1/n$ and $ R_n \ge C_R n^{-1/(d+2)}$.
If $\lambda\le 1/n$ then the inequality is trivially true for $n$ large enough anyway.
It follows  that \eqref{eq:fluctu charge micro} holds.

\medskip

\noindent \emph{Proof of Item $2$.} The argument is similar as above: Fixing some radius $R$, we have
\begin{multline*}
\Q \left(  |D(x,R)| \geq \lambda n^{1-1/d} \right) \\ 
\leq C \exp\left( -C \beta n ^{1-2/d} \left( \lambda^2 n^{2-2/d}  R ^{2-d}   n ^{(2-d)/d} \min(1,\lambda  n^{1-2/d}  R ^{-d} ) -Cn\right) \right)\\ 
= C \exp\left( -C \beta n ^{2-2/d} \left( \lambda^2  R^{2-d}  - C  \right) \right)   
\end{multline*} if $d\ge 3$,
for $n$ large enough, and $\le C  \exp\left( -C \beta n \min(\lambda^2   R^{2-d}, \lambda^3  R^{2-2d} )- C  \right)$ for $d=2$.
\medskip
 
\noindent \emph{Proof of Item $3$.} We start again from \eqref{fhn1}. Using Lemma \ref{lem:fluctu field} in the ball of radius $R n ^{1/d}$,  the fact that the total mass of $\nu'_n$ is $n$,  and then a change of variables, we have
\begin{align*}
\int_{\R ^d} |\nabla h'_{1,n}|^2 &\geq C(R n^{1/d})^{ 1- \frac{2}{q}}\left(\left\Vert \nabla h'_n \right\Vert_{L^q (B_{Rn ^{1/d}})}^2  -C n^2 \right)  \\
&\geq 
 C(R n^{1/d})^{ 1- \frac{2}{q}}
\left( n ^{\frac{2}{d} + \frac{2}{q} - 2} \left\Vert \nabla h_n \right\Vert_{L^q (B_{R})}^2 -C n^2 \right).
\end{align*}In particular, for $\lambda \ge 2C$, we have
\begin{align*}
\left\{ \left\Vert \nabla h_n \right\Vert_{L^q (B_{R})}    \geq \lambda n ^{ 2- \frac{1}{d}-\frac{1}{q}} \right\} \subset \left\{ n ^{-1} \int_{\R ^d} |\nabla h'_{1,n}|^2 \geq C \lambda^2 R ^{1-\frac{2}{q}   } n ^{\frac{1}{d}- \frac{2}{qd}+1} \right\} .
\end{align*}Since $q \ge 1$, we have $\frac{1}{d}- \frac{2}{qd}+1\ge 0$ and thus, for $n$ large enough,
inserting into \eqref{fhn1} we deduce 
that
\begin{equation}\label{eq:field fluctu final}
\Q \left( \left\Vert \nabla h_n \right\Vert_{L^q (B_{R})} \geq \lambda n ^{t_{q,d}} \right) \leq C \exp \left( - C_R \lambda^2 \beta n ^{ \tilde{t}_{q,d}} \right) 
\end{equation}
where $t_{q,d}$ and $\tilde{t}_{q,d}$ are as in the statement of the theorem. This ends the proof since, in view of the definition \eqref{hh},
$
- \Delta h_n =  c_d \left( \nu_n - n \mu_0\right)$, 
and thus
\[
\left\Vert \nabla h_n \right\Vert_{L^q (B_{R})} = c_d \left\Vert \nu_n - n \mu_0 \right\Vert_{W^{-1,q} (B_{R})}. 
\]

\subsection{Estimates on reduced densities: proof of Corollary \ref{thm:marginals}}

We first prove a simple lemma that formalizes in our setting the idea that a control on the fluctuations of the empirical measure in a symmetric probability measure  implies a control of the marginals of that measure. With this additional ingredient in hand, Corollary \ref{thm:marginals} becomes a consequence of Theorem \ref{thm:charge fluctu}, Item $3$. We have used a similar idea in \cite[Lemma 3.6]{RSY2}. 

\begin{lem}[\textbf{Control of the empirical measure implies control of the marginals}]\label{lem:empirical marginal}\mbox{}\\
Let $\mubf_n$ be a symmetric probability measure over $(\R ^{d})^n$ with reduced densities $\mubf ^{(k)}$. Let 
$$1 \leq q < \frac{d}{d-1}  \mbox{ and } p = \frac{q}{q-1}$$.
Recall the definition \eqref{hh} of $h_n$ as a function of $\xbf=(x_1,\ldots,x_n) $.
\begin{enumerate}
\item For any $\varphi \in C ^{\infty}_c (\R ^d)$ we have
\begin{equation}\label{eq:control 1 body}
\left| \int_{\R ^d} (\mubf ^{(1)} - \mu_0) \varphi \right| \leq \frac{1}{c_d n} \left\Vert \nabla \varphi \right\Vert_{L ^p}  \int_{\xbf \in \R ^{dn}} \left\Vert \nabla h_n \right\Vert_{L ^q}  \mubf (\xbf) d\xbf. 
\end{equation}
\item For any $\varphi \in C ^{\infty}_c (\R ^{dk})$ symmetric in the sense that for any permutation $\sigma$
\[
 \varphi(x_1,\ldots,x_k) = \varphi(x_{\sigma (1)},\ldots,x_{\sigma(k)}) 
\]
we have 
\begin{multline}\label{eq:control k body}
\left| \int_{\R ^{dk}} (\mubf ^{(k)} - \mu_0 ^{\otimes k}) \varphi \right| \leq \left(\frac{k}{c_d n} \int_{\xbf \in \R ^{dn}} \left\Vert \nabla h_n \right\Vert_{L ^q}  \mubf (\xbf) d\xbf + C \frac{k ^2}{n} \right) 
\\ \sup_{x_1 \in \R ^d} \ldots \sup_{x_{k-1} \in \R ^d }   \left\Vert \nabla \varphi (x_1,\ldots,x_{k-1}, \:. \: ) \right\Vert_{L ^p (\R ^d)}   
. 
\end{multline}
\end{enumerate}

\end{lem}

\begin{proof}
Using the symmetry of $\mu$ we write 
\begin{align*}
\int_{\R ^d} (\mubf ^{(1)} - \mu_0) \varphi &= \frac{1}{n} \int_{\xbf \in \R ^{dn}} \mubf (\xbf) \Big( \sum_{i=1} ^n \varphi (x_i) - n \int_{\R ^d} \varphi \mu_0 \Big) d\xbf \\
&= \frac{1}{n} \int_{\xbf \in \R ^{dn}} \mubf (\xbf) \left( \int_{\R^d} \Big( \sum_{i=1} ^n \delta_{x_i} - n  \mu_0 \Big) \varphi \right) d\xbf.
\end{align*}
Then 
\begin{equation}\label{eq:ipp marginals}
\left| \int_{\R^d} \Big( \sum_{i=1} ^n \delta_{x_i} - n  \mu_0 \Big) \varphi \right|= \frac{1}{c_d} \left| \int_{\R ^d} \nabla \varphi\cdot \nabla h_n \right|  \leq \frac{1}{c_d} \left\Vert \nabla \varphi \right\Vert_{L ^p} \left\Vert \nabla h_n \right\Vert_{L ^q}
\end{equation}
where we used \eqref{hh}, the assumption that $\varphi$ has compact support to justify the integration by parts and H\"older's inequality. This proves Item 1 since only the term $\left\Vert \nabla h_n \right\Vert_{L ^q}$ in the right-hand side of \eqref{eq:ipp marginals} depends on $\xbf$.

We now turn to Item 2, for which a little bit more algebra is required. We first note that, using the symmetry under particle exchange, 
\begin{align*}
\int_{\R ^{dk}} \mubf ^{(k)} \varphi &= \int_{\xbf \in \R ^{dn}} \mubf(\xbf) \frac{(n-k)!}{n!}\sum_{1\leq i_1 \neq \ldots \neq i_k \leq n} \varphi(x_{i_1},\ldots,x_{i_k}) d\xbf \\
&= \int_{\xbf \in \R ^{dn}} \mubf(\xbf) \frac{1}{n ^k}\sum_{1\leq i_1, \ldots, i_k \leq n} \varphi(x_{i_1},\ldots,x_{i_k}) d\xbf + \int_{\xbf \in \R ^{dn}} \mubf(\xbf) \tilde{\varphi} (\xbf) d\xbf
\end{align*}
where 
\begin{equation}\label{eq:norm phi tilde}
\left\Vert \tilde{\varphi} \right\Vert_{L ^{\infty}} \leq C\frac{k ^2}{n} \left\Vert \varphi \right\Vert_{L ^{\infty}}.  
\end{equation}
This is a classical combinatorial trick that seems to go back to~\cite{Gru} and can be found e.g. in~\cite[Proof of Lemma 8]{HM} or~\cite[Proof of Theorem~2.2]{Roucdf}. In fact we are here using the de Finetti-type theorem of~\cite{DF} in disguise. Then, in view of the condition on $q$ we have $p>d$ and thus the Sobolev embedding $W^{1,p} \hookrightarrow L ^{\infty}$, so that, in view of~\eqref{eq:norm phi tilde},  
\begin{multline}\label{eq:use deF}
 \left|\int_{\R ^{dk}} \mubf ^{(k)} \varphi - \int_{\xbf \in \R ^{dn}} \mubf(\xbf) \frac{1}{n ^k}\sum_{1\leq i_1, \ldots, i_k\leq n} \varphi(x_{i_1},\ldots,x_{i_k}) d\xbf \right| \leq
 \\ C \frac{k ^2}{n} \sup_{x_1 \in \R ^d} \ldots \sup_{x_{k-1} \in \R ^d }   \left\Vert \nabla \varphi (x_1,\ldots,x_{k-1}, \:. \: ) \right\Vert_{L ^p (\R ^d)}
\end{multline}
follows from the above.

We then use this and the triangular inequality to obtain 
\begin{align}\label{eq:calcul k marginals}
\left|\int_{\R ^{dk}} (\mubf ^{(k)} - \mu_0 ^{\otimes k}) \varphi \right| &\leq  n ^{-k} \left| \int_{\xbf\in \R ^{dn}} \mubf(\xbf) \left(\int_{\R ^{dk}} \Big(\sum_{1\leq i_1, \ldots, i_k \leq n} \delta_{x_{i_1}}\otimes \ldots \otimes \delta_{x_{i_k}} - (n\mu_0) ^{\otimes k}\Big) \varphi\right) d\xbf\right| \nonumber\\
& + C \frac{k ^2}{n} \sup_{x_1 \in \R ^d} \ldots \sup_{x_{k-1} \in \R ^d }   \left\Vert \nabla \varphi (x_1,\ldots,x_{k-1}, \:. \: ) \right\Vert_{L ^p (\R ^d)}
% &= n ^{-k} \int_{\xbf\in \R ^{dn}} \mubf(\xbf) \left(\int_{\R ^{dk}} \Big(\Big(\sum_{i=1} ^n \delta_{x_i} \Big) ^{\otimes k} - (n\mu_0) ^{\otimes k} \Big) \varphi \right)d\xbf. 
\end{align}
and note that  
\begin{multline*}
n ^{-k} \int_{\xbf\in \R ^{dn}} \mubf(\xbf) \left(\int_{\R ^{dk}} \Big(\sum_{1\leq i_1, \ldots, i_k \leq n} \delta_{x_{i_1}}\otimes \ldots \otimes \delta_{x_{i_k}} - (n\mu_0) ^{\otimes k}\Big) \varphi\right) d\xbf = 
\\ n ^{-k} \int_{\xbf\in \R ^{dn}} \mubf(\xbf) \left(\int_{\R ^{dk}} \Big(\Big(\sum_{i=1} ^n \delta_{x_i} \Big) ^{\otimes k} - (n\mu_0) ^{\otimes k} \Big) \varphi \right)d\xbf. 
\end{multline*}

We then write
\begin{multline*}
\Big(\sum_{i=1} ^n \delta_{x_i} \Big) ^{\otimes k} - (n\mu_0) ^{\otimes k} = \Big( \sum_{i=1} ^n \delta_{x_i} - n\mu_0 \Big)\otimes \Big( \sum_{i=1} ^n \delta_{x_i} \Big) ^{\otimes (k-1)} \\  + (n\mu_0) \otimes \Big( \Big(\sum_{i=1} ^n \delta_{x_i} \Big) ^{\otimes (k-1)} - (n\mu_0) ^{\otimes (k-1)} \Big)
\end{multline*}
from which
\begin{equation}\label{eq:algebra k marginal}
\Big(\sum_{i=1} ^n \delta_{x_i} \Big) ^{\otimes k} - (n\mu_0) ^{\otimes k} = \sum_{j=0} ^{k-1} (n\mu_0) ^{\otimes j} \otimes \Big( \sum_{i=1} ^n \delta_{x_i} - n\mu_0 \Big) \otimes \Big( \sum_{i=1} ^n \delta_{x_i} \Big) ^{\otimes (k-j-1)} 
\end{equation}
follows by induction. On the other hand, recalling that $\left\Vert n\mu_0 \right\Vert_{\rm TV}=\left\Vert \sum_{i=1} ^n \delta_{x_i} \right\Vert_{\rm TV} = n$ and using the symmetry of $\varphi$ we have
\begin{multline}\label{eq:final k marginal}
\left|\int_{\R ^{dk}} \varphi \sum_{j=0} ^{k-1} (n\mu_0) ^{\otimes j} \otimes \Big( \sum_{i=1} ^n \delta_{x_i} - n\mu_0 \Big) \otimes \Big( \sum_{i=1} ^n \delta_{x_i} \Big) ^{\otimes (k-j-1)} \right| \\
\leq k n ^{k-1} \sup_{x_1 \in \R ^d} \ldots \sup_{x_{k-1} \in \R ^d }   \left|\int_{\R ^d} \Big( \sum_{i=1} ^n \delta_{x_i} -n\mu_0 \Big)  \varphi (x_1,\ldots,x_{k-1}, \:. \: )\right| \\
\leq \frac{k n ^{k-1}}{c_d} \sup_{x_1 \in \R ^d} \ldots \sup_{x_{k-1} \in \R ^d }   \left\Vert \nabla \varphi (x_1,\ldots,x_{k-1}, \:. \: ) \right\Vert_{L ^p (\R ^d)} \left\Vert \nabla h_n \right\Vert_{L ^q (\R ^d)} 
\end{multline}
by definition of the total variation norm 
%$\left\Vert  \nu \right\Vert_{\rm TV} = \sup_{\left\Vert  \phi \right\Vert_{L ^{\infty}} = 1} |\nu(\phi)|$ 
and an estimate exactly similar to \eqref{eq:ipp marginals}. Combining \eqref{eq:calcul k marginals}, \eqref{eq:algebra k marginal} and \eqref{eq:final k marginal} we obtain the desired result.
\end{proof}

\begin{proof}[Proof of Corollary \ref{thm:marginals}]
From \eqref{eq:field fluctu final} it is easy to deduce that for some large enough $\lambda$
\begin{align*}
\int_{\xbf \in \R ^{dn}} \left\Vert \nabla h_n \right\Vert_{L ^q} d\Q (\xbf) &= \int_{\left\{\left\Vert \nabla h_n \right\Vert_{L ^q} \leq \lambda n^{t_{q,d}} \right\} }  \left\Vert \nabla h_n \right\Vert_{L ^q} d \Q(\xbf) \\ &+ \int_{\left\{\left\Vert \nabla h_n \right\Vert_{L ^q} \geq \lambda n^{t_{q,d}} \right\} }  \left\Vert \nabla h_n \right\Vert_{L ^q} d\Q(\xbf)
\\&\leq \lambda n^{t_{q,d}} \left( 1+ o(1)\right).
\end{align*}
Indeed, the second term is negligible because of \eqref{eq:field fluctu final} and for the first one we only use the fact that $\Q$ is a probability. The results of Corollary \ref{thm:marginals} then follow from Lemma \ref{lem:empirical marginal}.\end{proof}

\appendix

\section{Existence of a minimizer for $\W$: direct proof}\label{sec:appendix}

We start by proving that there exists a minimizer for $\W_\eta$. Picking a minimizing sequence $\j_{n}$ of admissible electric fields, a standard concentration argument yields a sequence $R_n$ of radii such that 
\begin{equation}\label{eq:ap 1}
\lim_{n \to \infty} \W_\eta (\j_n) = \lim_{n \to \infty} \limsup_{R \to \infty}\dashint_{K_R} |\j_{n,\eta}| ^2- \kappa_d \g (\eta) = \lim_{n\to \infty}  \dashint_{K_{R_n}} |\j_{n,\eta}| ^2- \kappa_d \g (\eta).
\end{equation}
Next we write
\begin{equation}\label{eq:ap 2}
\lim_{n\to \infty}  \dashint_{K_{R_n}} |\j_{n,\eta}| ^2 = \lim_{n\to \infty}  |K_{R_n}| ^{-1} \int_{\R ^d} \chi*\indic_{K_{R_n}} |\j_{n,\eta}| ^2 
\end{equation}
for a fixed compactly supported $\chi$ and introduce 
\[
 \mathbf{f}_n (\j) := \begin{cases}
                       \int \chi(y) |\j | ^2 \mbox{ if } \exists \, \lambda \in K_{R_n} \,:\, \j = \theta_{\lambda} (\j_{n,\eta})\\
                       +\infty \mbox{ otherwise }
                      \end{cases}
\]
where $\theta_\lambda$ is the action of the translation by $\lambda$. Then
\[
\lim_{n\to \infty}  \dashint_{K_{R_n}} |\j_{n,\eta}| ^2 =  \lim_{n\to \infty} \dashint_{K_{R_n}} \mathbf{f}_n (\theta_\lambda (\j_{n,\eta}) ) d\lambda
\]
and we may apply the framework of \cite[Theorem 7]{ss2d}, or even the simpler version in \cite[Theorem 3]{gl13}. We use the latter. Arguing as in the proof of Lemma \ref{lemprel}, Assumptions 1 and 2 about the functional $\mathbf{f}_n$ simply follow from compactness in $L^2_{loc}$ and lower semi-continuity. Indeed, if $\j_n \wto \j$ in $L ^2_{loc} (\R ^d, \R ^d)$ we have 
\[
 \liminf_{n\to \infty} \mathbf{f}_n (\j_n) \geq \mathbf{f} (\j) 
\]
where 
\[
\mathbf{f} (\j) := \begin{cases}
                    \int \chi(y) |\j | ^2 \mbox{ if } \exists \, \j' \,:\, \j = \Phi_\eta (\j')\\
                    +\infty \mbox{ otherwise.}
                   \end{cases}
\]
Applying \cite[Theorem 3]{gl13} we deduce
\[
\lim_{n\to \infty}  \dashint_{K_{R_n}} |\j_{n,\eta}| ^2 =  \liminf_{n\to \infty} \dashint_{K_{R_n}} \mathbf{f}_n (\theta_\lambda (\j_{n,\eta}) ) d\lambda \geq \int  \left(\lim_{R\to \infty} \dashint_{K_R} \mathbf{f} (\theta_\lambda \j)d\lambda \right) dP_\eta(\j).
\]
where $P_\eta$ is a probability measure on $L^2 _{loc} (\R ^d, \R ^d)$. Recalling \eqref{eq:ap 1}  we thus have
\begin{equation}\label{eq:ap 3}
\inf_{\bai} \W_\eta = \lim_{n \to \infty} \W_\eta (\j_n) \geq \lim_{R\to \infty}  \frac{1}{|K_R|}  \chi * \indic_{K_R} |\j|^2 dP_\eta(\j) - (\kappa_d \g(\eta) +\gamma_2 \indic_{d=2}).
\end{equation}
 We denote $P$ the push-forward of $P_\eta$  by $\Phi_\eta ^{-1}$, i.e. $dP (\Phi_\eta(\j)) = dP_\eta (\j)$. Recalling that $P_\eta$ is defined using \cite[Theorem 3]{gl13}, we can argue as in the proof of Proposition \ref{promino} to prove that $P$ does not depend on $\eta$. We then rewrite \eqref{eq:ap 3} as  
\begin{equation}\label{eq:ap 4}
\inf_{\bai} \W_\eta \geq \int  \W_\eta (\j) dP (\j) 
\end{equation}
and since $\W_\eta$ is bounded below independently of $\eta$ by Proposition \ref{pro:Wbb} we can use Fatou's lemma to pass to the $\liminf$ as $\eta \to 0$:
\[
\inf_{\bai} \W = \liminf_{\eta \to 0} \inf_{\bai} \W_\eta  \geq \int  \left( \liminf_{\eta \to 0} \W_\eta (\j) \right) dP (\j) =\int  \W (\j) dP (\j);
\]
from which it follows that $P$ is concentrated on the set of minimizers of $\W$, which has to be nonempty.

\section{Comparison of $W$ and $\W$ and the periodic case}\label{sec:appendix 2}

In this appendix we discuss the relation between the two version of the renormalized energy functionals. In the case of periodic configurations, both co\"incide with a simplified expression dependind only on the points. 

Lemma \ref{lemequiv} gives a convenient way to bound from below the energy of well-separated charge configurations. We can also use it to compare $W$ and $\W$ when $d=2$:

\begin{pro}[\textbf{$W$ and $\W$ coincide in 2D for well-separated points}]\label{propequiv}\mbox{}\\
Assume $d=2$, and let $\j \in \mathcal{A}_1$ be such that $\W(\j) <+\infty$ and the associated set of points satisfies $\min_{p\neq p' \in \Lambda } |p-p'| \ge \eta_0>0$ for some $\eta_0>0$. Then 
$\W(\j)= W(\j)$.
\end{pro}
\begin{proof}
Without loss of generality, we may assume that $\eta_0<\hal$.  By well-separation of the points, for any $R>1$ we may find  a set $U_R$  such that $K_R \subset U_R \subset K_{R+1},$    and \eqref{wellsep} is verified in $U_R$ with $\eta_0/2$.
We may then apply Lemmas \ref{lemequiv}  in $U_R$ and obtain
\begin{equation}\label{fed0}
\int_{U_R} |\j_\eta|^2 - \#(\Lambda\cap U_R)( \kappa_d \g(\eta) + \gamma_2 \indic_{d=2}) = W(\j, \indic_{U_R}) +  o_\eta(1)\# (\Lambda \cap U_R) .\end{equation}
with $C$ depending only on $\eta_0$ and $d$.
Next, let $\chi_{K_R}$ be as in Definition \ref{def:renorm SS}. We have 
\begin{equation}\label{fed1}W(\j,\chi_{K_R})- W(\j, \indic_{U_R}) = W( \j,\chi_{K_R} \indic_{U_R\backslash U_{R-1} }) - W(\j, \indic_{U_R\backslash U_{R-1}}).\end{equation} 
Note that by construction, the $B(p, \eta_0/2)$ do not intersect $\p U_R$ and $\p U_{R-1}$. We may next write, by definition of $W$, for any $0<r\le\eta_0/2$,
\begin{multline*}W(\j,  \indic_{U_R\backslash U_{R-1}})= 
\int_{(U_R\backslash U_{R-1} )\backslash \cup_p B(p,r) } |\j|^2 +
\sum_{p \in \Lambda \cap ( U_R\backslash U_{R-1}   ) } 
 \left(  \lim_{\eta \to 0} \int_{B(p,r)  \backslash  B(p,\eta) } |\j|^2 - c_d \g(\eta)\right), 
 \end{multline*}
and using \eqref{minfacile} (applied to $r$ instead of $\eta_0/2$)  we have 
$$\lim_{\eta \to 0} \int_{B(p,r)  \backslash  B(p,\eta) } |\j|^2 - c_d \g(\eta)\ge - c_d \g (r) - C$$
for each $p$ (where $C$ depends only on $d$). It follows that we may write 
\begin{multline}\label{decm0}
\int_{(U_R\backslash U_{R-1} )\backslash \cup_p B(p,r) } |\j|^2 +
\sum_{p \in \Lambda \cap ( U_R\backslash U_{R-1}   ) } 
 \left|  \lim_{\eta \to 0} \int_{B(p,r)  \backslash  B(p,\eta) } |\j|^2 - c_d \g(\eta)\right|\\
 \le W(\j,  \indic_{U_R \backslash U_{R-1}   })+ ( c_d \g(r)+C ) \# (\Lambda\cap (U_R\backslash U_{R-1})).
 \end{multline}
Similarly as above, we have
\begin{multline}\label{decm1}\left|
W(\j, \indic_{U_R\backslash U_{R-1}})-
W(\j, \chi_{K_R}\indic_{U_R\backslash U_{R-1}  })  \right|\le
\int_{(U_R\backslash U_{R-1} )\backslash \cup_p B(p,\eta_0/2) }(1-\chi_{K_R}) |\j|^2 \\ +
\sum_{p \in \Lambda \cap ( U_R\backslash U_{R-1}   ) } 
(1-\chi_{K_R}(p)) \left(  \lim_{\eta \to 0} \int_{B(p,\eta_0/2)  \backslash  B(p,\eta) } |\j|^2 - c_d \g(\eta)\right)\\
+ \sum_{p \in \Lambda \cap ( U_R\backslash U_{R-1} )}  \lim_{\eta \to 0} \int_{B(p,\eta_0/2)  \backslash  B(p,\eta) } |\chi_{K_R}- \chi_{K_R}(p)||\j|^2\end{multline}
In view of \eqref{decm0} applied with $r=\eta_0/2$ we can bound the first two terms on the right-hand side by 
$W(\j,  \indic_{U_R\backslash U_{R-1}})+ C \# (\Lambda\cap (U_R\backslash U_{R-1}))$, where $C$ depends only on $d$ and $\eta_0$.
We turn to the last term. Let us set 
$$\phi(r)= \int_{(U_R\backslash U_{R-1} )\backslash \cup_p B(p,r) } |\j|^2 .$$
Since $\chi_{K_R}$ is Lipschitz, we may write
\begin{multline*}
\sum_{p \in \Lambda \cap ( U_R\backslash U_{R-1} )}  \lim_{\eta \to 0} \int_{B(p,\eta_0/2)  \backslash  B(p,\eta) } |\chi_{K_R}- \chi_{K_R}(p)||\j|^2\\
\le C \lim_{\eta \to 0}  \sum_{p \in \Lambda \cap ( U_R\backslash U_{R-1} )}  \int_{B(p,\eta_0/2)  \backslash  B(p,\eta) } |x-p||\j|^2= - C \lim_{\eta\to 0} \int_{\eta}^{\eta_0/2} r\phi'(r)\, dr\\
= - C\lim_{\eta\to 0} \left( \hal \eta_0\phi(\eta_0/2) - \eta\phi(\eta) + \int_{\eta}^{\eta_0/2}\phi(r)\, dr \right).
\end{multline*}
Using \eqref{decm0} to bound $\phi(r)$, and using the integrability of $w(r) = - \log r$  near $0$ (this is the one point where we use that  $d=2$), we deduce that the third term in \eqref{decm1} can also  be bounded by  $CW(\j,  \indic_{U_R\backslash U_{R-1}})+ C \# (\Lambda\cap (U_R\backslash U_{R-1}))$.
Combining with \eqref{fed0} and \eqref{fed1} we are led to 
\begin{multline}\label{fed3}
\left|\int_{U_R} |\j_\eta|^2 - \#(\Lambda\cap U_R)( \kappa_d \g(\eta)+\gamma_2 \indic_{d=2}) - W(\j, \chi_{K_R}) \right|\\ \le   o_\eta(1)\# (\Lambda \cap U_R) + CW(\j,  \indic_{U_R\backslash U_{R-1}})+ C \# (\Lambda\cap (U_R\backslash U_{R-1})).\end{multline}

On the other hand, since $\W(\j) <\infty$, Lemma \ref{lembornnu} applies and gives that 
$\lim_{R\to \infty} \frac{1}{|K_R|} \# (\Lambda \cap U_R)=1$ and $\#(\Lambda\cap (U_R\backslash U_{R-1}))=o(R^d)$ as $R \to \infty$.  It also implies, dividing \eqref{fed0} by $|K_R|$ and letting $ R \to \infty$ and then $\eta\to 0$, that $\lim_{R\to \infty} \frac{1}{|K_R|} W(\j, \indic_{U_R}) = \W(\j)$. This in turns implies that 
$W(\j, \indic_{U_R\backslash U_{R-1}})=o(R^d)$. Inserting these into  \eqref{fed3}, dividing by $|K_R|$ and letting $R \to \infty$ and then $\eta \to 0$, we obtain $\W(\j)- W(\j)=0$.
\end{proof}

\begin{remark} The fact that Lemma \ref{lemequiv} holds in any dimension $d \ge 2$ but Proposition \ref{propequiv} only for $d=2$ seems to indicate that the cutoff procedure with $\chi_{K_R}$ is probably not adapted for dimension $d \ge 3$ and the value of $W$ defined in \eqref{WWold}--\eqref{Wold} may depend on the choice of cutoff  $\chi_{K_R}$, contrarily to what happens for $d=2$ as proven in \cite{gl13}. This has no consequence for us however as we will never need the result of  Proposition \ref{propequiv}.
\end{remark}

We now turn to consequences of the previous result: when configurations are periodic with simple points, then these are automatically well-separated, and we can deduce an explicit expression for $\W$ in terms of the points only. 

\begin{pro}[\textbf{The energy of periodic configurations}]\label{pro:periodic config}\mbox{}\\
Let $a_1,\dots a_N$ be  distinct points on a flat torus $\mathbb{T}$ of volume $N$. Let $h$  be the mean-zero $\mathbb{T}$-periodic  function
satisfying
$$-\Delta h= c_d \Big( \sum_{i=1}^N \delta_{a_i} -1\Big) \quad \text{in} \ \mathbb{T}.$$
Then \begin{equation}\label{formuleper}
\W(\nab h) = c_d^2  \frac{1}{N}\sum_{i\neq j} G(a_i-a_j) + c_d^2 R\end{equation}
where $G$ is the Green's function of the torus, i.e. solves
\begin{equation}
\label{greentorus}
- \Delta G= \delta_0- \frac{1}{N} \quad \text{in} \ \mathbb{T}\end{equation}
with $\dashint_{\mathbb{T}} G (x)\, dx=0$ and $R$ is a constant, equal to $\lim_{x\to 0} \left(G(x)- c_d^{-1}\g(x)\right).$
\end{pro}

This result should be compared with~\cite[Lemma 1.3]{gl13}.

\begin{proof} In view of Lemma \ref{lemequiv} and the periodicity of $h$,  we easily check that $\W (\nab h)= \frac{1}{|\mathbb{T}|}W(\nab h, \indic_{\mathbb{T}})$. This is a renormalized energy computation \`a la \cite{BBH}: first, using Green's formula and the equation satisfied by $h$, we compute 
\begin{equation}\label{compper}
\int_{\mathbb{T} \backslash \cup_{i=1}^N B(a_i, \eta)} |\nab h|^2 = -  \sum_{i=1}^N 
\int_{\p B(a_i, \eta)} h \nab h \cdot \vec{\nu}- c_d \int_{\mathbb{T} \backslash \cup_{i=1}^N B(a_i, \eta)}  h ,\end{equation}with $\vec{\nu}$ the outer unit normal.
The second term tends to $0$ as $\eta \to 0$ since $h$ has mean zero.
For the first term we note that  $h= c_d \sum_{i=1}^N G(x-a_i)$, and that  $G(x)=c_d^{-1}\g(x)  +R(x)$ with $R$ a $C^1$ function, and  insert this to find 
\begin{multline*}\int_{\p B(a_i, \eta)} h \nab h \cdot \vec{\nu}=    \Big(  \g(\eta)  + c_d \sum_{j\neq i} G(a_i-a_j)\Big)\int_{\p B(a_i, \eta)} \nab h\cdot \vec{\nu} \\ + \int_{\p B(a_i, \eta)} (c_d R(x-a_i)+f(x-a_i))  \nab h\cdot \vec{\nu}\end{multline*}where $f$ is a $C^1$ function equal to $0$ at $0$. We then use that, by Green's theorem,  
$$\int_{\p B(a_i, \eta)} \nab h\cdot \vec{\nu} =  \int_{B(a_i, \eta)}\Delta h =- c_d + o_\eta(1),$$
and that $|\nab h|(x)\le  C|\nab \g|(x-a_i)+C \le C |x-a_i|^{1-d}$ to conclude that 
$$\lim_{\eta\to 0}  \int_{\p B(a_i, \eta)} - h \nab h \cdot \vec{\nu} - c_d \g(\eta)  = c_d^2\sum_{j\neq i} G(a_i-a_j) + c_d^2 R(0).$$ Inserting into \eqref{compper}, in view of the definition of $W(\nab h, \indic_{\mathbb{T}})$,  we get the result.
\end{proof}

If there is a multiple point in a periodic configuration, it is easy to see that both $\W$ and $W$ are $+\infty$ for this configuration, e.g. as a limit case of the above result. Both ways of computing the renormalized energy are thus perfectly equivalent for any periodic configuration.

Configurations that form a simple lattice correspond to this situation but with only one point in the torus,  in that case
 $\W$ is thus equal to $c_d^2 R=\lim_{x\to 0} \left(c_d^2G(x)-c_d \g(x)\right).$
In addition, we may compute explicitly the Green's function of the torus, using Fourier series.
Using the normalization of the Fourier transform
$$\hat{f}(y)= \int_{\mr^d} f(x) e^{-2i\pi x\cdot y}\, dx$$
 we have
 $$\hat{G}(y)= \frac{1}{4\pi^2 |y|^2} \sum_{p \in \Lambda^*\backslash\{0\}}  \delta_p(y)$$where $\Lambda^* $ is the dual lattice of $\Lambda $ i.e.  the set of $q$'s such that $p\cdot q \in \mathbb{Z} $ for every $p\in \Lambda$.
 By Fourier inversion formula we obtain the expression of $G$ in Fourier series
\begin{equation}
\label{G}G(x)=\sum_{p \in \Lambda^* \backslash\{0\}} \frac{e^{2i\pi p \cdot x}}{4\pi^2 |p|^2}.
\end{equation}
Thus, we obtain that 
$$\W(\Lambda) := \W(\nab h)= c_d^2 \lim_{x\to 0} \sum_{p \in \Lambda^* \backslash\{0\}} \frac{e^{2i\pi p \cdot x}}{4\pi^2 |p|^2}- \frac{ \g(x)}{c_d},$$
where $\Lambda$ denotes the lattice.

The series that appears here is an Eisenstein series $E_\Lambda(x)= \sum_{p \in \Lambda^* \backslash\{0\}} \frac{e^{2i\pi p \cdot x}}{4\pi^2 |p|^2}.$
In dimension $2$, using the ``first Kronecker limit formula" (see also \cite{gl13} for an analytic proof), one can  show that this series is related to   the Epstein Zeta function of the lattice $\Lambda^*$ defined by 
$$\zeta_{\Lambda}(s)= \sum_{p \in \Lambda\backslash \{0\}} \frac{1}{|p|^{2+s}},$$ more precisely that 
\begin{equation}\label{Ez}
 \lim_{x\to 0} E_{\Lambda_1}(x) - E_{\Lambda_2}(x)= C \lim_{s>0, s\to 0} \zeta_{\Lambda_1^*}(s) - \zeta_{\Lambda_2^*}(s).\end{equation}
Thus minimizing $W$ over lattices reduces to minimizing the Zeta  function over lattices $\Lambda$ of volume $1$, and this question was solved by Cassels, Rankin, Ennola, Diananda, in the 60's (see a self-contained proof in \cite{montgomery}) in dimensions $2$, where the unique minimizer is the triangular lattice  (the minimizer is  also identified in dimensions $8$ and $24$). Minimizing the Zeta function over lattices  remains an open question in dimension $\ge 3$, even though the face-centered cubic (FCC) lattice is conjectured to be a local minimizer (cf. \cite{sarns}). In dimension $d\ge 3$ however,  \eqref{Ez} is not proven and it is not clear how to give it a meaning since  the series involved are all divergent (one would at least need to use the meromorphic extension of the Zeta function), see e.g. \cite{lang,siegel}.

\bigskip

\noindent Nicolas Rougerie:
\newline {\tt nicolas.rougerie@grenoble.cnrs.fr} \newline
Sylvia Serfaty:
\newline {\tt serfaty@ann.jussieu.fr}

\end{document}